\documentclass[showpacs,twocolumn,floatfix,superscriptaddress,notitlepage]{revtex4-2}
\usepackage{graphicx}
\usepackage{dcolumn}
\usepackage{bm}
\usepackage[colorlinks=true,urlcolor=blue,citecolor=blue,linkcolor=blue]{hyperref}
\usepackage{suffix}

\usepackage{xr}
\usepackage{xfrac}
\usepackage{multirow}
\usepackage{amsmath}
\usepackage{amssymb}
\usepackage{braket}
\usepackage{physics}
\usepackage{amsthm}
\usepackage{natbib}
\usepackage{mathtools}
\usepackage{diagbox}
\usepackage[utf8]{inputenc}
\usepackage[english]{babel}
\usepackage{scalerel}[2014/03/10]
\usepackage{stackengine}
\usepackage{quantikz}
\usepackage{tikz}
\usetikzlibrary{shapes.geometric}
\usepackage{algorithm}
\usepackage[noend]{algpseudocode} 
\usepackage{hyperref} %
\usepackage[capitalize]{cleveref} 
\usepackage{subcaption} 
\usepackage{caption}
\usepackage{overpic}
\usepackage{ragged2e}
\crefname{section}{Appendix}{Appendix}

\DeclareRobustCommand{\ancillasym}{ \vcenter{\hbox{\tikz[scale=0.5,transform shape]{ \draw[line width=0.8pt] (0.6,0) -- (0.95,0); \draw[line width=0.8pt] (0.95,-0.35) -- (0.95,0.35); \draw[line width=0.8pt] (1.15,-0.25) -- (1.15,0.25); \draw[line width=0.8pt] (1.35,-0.15) -- (1.35,0.15); }}} }

\newcommand{\algmargin}{\the\ALG@thistlm}
\makeatother
\newlength{\whilewidth}
\algdef{SE}[parWHILE]{parWhile}{EndparWhile}[1]
  {\parbox[t]{\dimexpr\linewidth-\algmargin}{%
     \hangindent\whilewidth\strut\algorithmicwhile\ #1\ \algorithmicdo\strut}}{\algorithmicend\ \algorithmicwhile}%
\algnewcommand{\parState}[1]{\State%
  \parbox[t]{\dimexpr\linewidth-\algmargin}{\strut #1\strut}}

\graphicspath{{images/}}

\makeatletter
\newtheorem{theorem}{Theorem}[]
\newtheorem*{remark}{Remark}
\newtheorem{corollary}{Corollary}[]
\newtheorem{lemma}[]{Lemma}

\newtheorem{definition}{Definition}

\theoremstyle{definition}

\makeatletter

\newcommand{\proofsketchname}{Proof Sketch}
\makeatother
\renewcommand{\exp}[1]{\text{exp}\left( #1 \right)}
\newcommand{\Vexp}[1]{\left\langle #1 \right\rangle}

\newcommand{\loss}{L\left(\bm{\theta}\right)}
\newcommand{\lossM}{L^{\mathcal{C}}\left(\bm{\theta},\bm{\theta}_{\mathcal{G}}\right)}
\newcommand{\lossMP}{L^{\mathcal{C}}\left(\bm{\theta},\left(\theta_{\mathcal{G}_{11}},\theta_{\mathcal{G}_{12}},\ldots,\theta_{\mathcal{G}_{K3}},\bm{0}\right)\right)}
\newcommand{\lossMT}{L_T^{\mathcal{C}}\left(\bm{\theta},\bm{\theta}_{\mathcal{G}},\bm{\theta}_{\mathcal{G}'_T}\right)}
\newcommand{\lossMTs}{L_{\{T_1,T_2,\ldots\}}^{\mathcal{C}}\left(\bm{\theta},\bm{\theta}_{\mathcal{G}},\bm{\theta}_{\mathcal{G}'_T}\right)}
\newcommand{\lossMNoi}{\widetilde{L}^{\mathcal{C}}\left(\bm{\theta},\bm{\theta}_{\mathcal{G}}\right)}
\newcommand{\lossNoi}{\widetilde{L}\left(\bm{\theta}\right)}
\newcommand{\gadget}[1]{\mathcal{G}\left(#1\right)}
\newcommand{\gadgetT}[1]{\mathcal{G}'_T\left(#1\right)}
\newcommand{\channelM}[1]{\Phi^{\mathcal{C}}\left(#1\right)}
\newcommand{\channelMT}[1]{\Phi^{\mathcal{C}}_T\left(#1\right)}
\newcommand{\channelMTs}[1]{\Phi^{\mathcal{C}}_{\{T_1,T_2,\ldots\}}\left(#1\right)}
\newcommand{\unitaryM}[1]{\mathbf{U}^{\mathcal{C}}\left(#1\right)}
\newcommand{\unitaryMT}[1]{\mathbf{U}^{\mathcal{C}}_T\left(#1\right)}
\newcommand{\unitaryMTs}[1]{\mathbf{U}^{\mathcal{C}}_{\{T_1,T_2,\ldots\}}\left(#1\right)}
\newcommand{\unitaryMNoi}[1]{\widetilde{\mathbf{U}}^{\mathcal{C}}\left(#1\right)}
\newcommand{\AngleSet}{\left\{0, \frac{\pi}{2}, \pi, \frac{3\pi}{2} \right\}}
\newcommand{\ExpC}{\mathbb{E}_{\bm{\theta}}}
\newcommand{\ExpMC}{\mathbb{E}_{\left(\bm{\theta},\bm{\theta}_\mathcal{G}\right)}}
\newcommand{\Fnum}{f_{j,O}^{\mathcal{C}}}
\newcommand{\Fnumall}{f_{\mathcal{G},O}^{\mathcal{C}}}
\newcommand{\Fnumact}{f_{\text{act}}^{\mathcal{C}}}
\newcommand{\uniquePathM}{\vec{\mathbf{s}}^{\hspace{0.1em}\left((\bm{\theta},\bm{\theta}_\mathcal{G}),\alpha\right)}}
\newcommand{\uniquePathMpos}[1]{\mathbf{s}^{\hspace{0.1em}\left((\bm{\theta},\bm{\theta}_\mathcal{G}),\alpha\right)}_{#1}}
\newcommand{\uniquePathMb}{\vec{\mathbf{s}}^{\hspace{0.1em}\left((\bm{\theta},\bm{\theta}_\mathcal{G}),\beta\right)}}
\newcommand{\uniquePathMT}{\vec{\mathbf{s}}^{\hspace{0.1em}\left((\bm{\theta},\bm{\theta}_\mathcal{G},\bm{\theta}_{\mathcal{G}'_T}),\alpha\right)}}
\newcommand{\uniquePathMTb}{\vec{\mathbf{s}}^{\hspace{0.1em}\left((\bm{\theta},\bm{\theta}_\mathcal{G},\bm{\theta}_{\mathcal{G}'_T}),\beta\right)}}

\newcommand{\sm}{Supplementary Information}

\begin{document}
\preprint{APS/123-QED}

\title{Taming Barren Plateaus in Arbitrary Parameterized Quantum Circuits without Sacrificing Expressibility}
\author{Zhenyu Chen}
\thanks{These authors contributed equally to this work.}
\affiliation{Department of Computer Science and Technology, Tsinghua University, Beijing, China}

\author{Yuguo Shao}
\thanks{These authors contributed equally to this work.}
\affiliation{Yau Mathematical Sciences Center, Tsinghua University, Beijing 100084, China}
\affiliation{Yanqi Lake Beijing Institute of Mathematical Sciences and Applications, Beijing 100407, China}

\author{Zhengwei Liu}
\email{liuzhengwei@mail.tsinghua.edu.cn}
\affiliation{Yau Mathematical Sciences Center, Tsinghua University, Beijing 100084, China}
\affiliation{Yanqi Lake Beijing Institute of Mathematical Sciences and Applications, Beijing 100407, China}
\affiliation{Department of Mathematics, Tsinghua University, Beijing 100084, China}

\author{Zhaohui Wei}
\email{weizhaohui@gmail.com}
\affiliation{Yau Mathematical Sciences Center, Tsinghua University, Beijing 100084, China}
\affiliation{Yanqi Lake Beijing Institute of Mathematical Sciences and Applications, Beijing 100407, China}

\begin{abstract}
Quantum algorithms based on parameterized quantum circuits (PQCs) have enabled a wide range of applications on near-term quantum devices. However, existing PQC architectures face several challenges, among which the ``barren plateaus" phenomenon is particularly prominent. In such cases, the loss function concentrates exponentially with increasing system size, thereby hindering effective parameter optimization. To address this challenge, we propose a general and hardware-efficient method for eliminating barren plateaus in an arbitrary PQC. Specifically, our approach achieves this by inserting a layer of easily implementable quantum channels into the original PQC, each channel requiring only one ancilla qubit and four additional gates, yielding a modified PQC (MPQC) that is provably at least as expressive as the original PQC and, under mild assumptions, is guaranteed to be free from barren plateaus.
Furthermore, by appropriately adjusting the structure of MPQCs, we rigorously prove that any parameter in the original PQC can be made trainable.
Importantly, the absence of barren plateaus in MPQCs is robust against realistic noise, making our approach directly applicable to near-term quantum hardware. Numerical simulations demonstrate that MPQC effectively eliminates barren plateaus in PQCs for preparing thermal states of systems with up to 100 qubits and 2400 layers. Furthermore, in end-to-end simulations, MPQC significantly outperforms PQC in finding the ground-state energy of a complex Hamiltonian.


\end{abstract}

\maketitle

\section{Introduction}

Parameterized quantum circuits (PQCs) play a central role in a wide range of quantum algorithms, {including those for quantum machine learning~\cite{benedetti2019parameterized,biamonte2017quantum,beer2020training,ren2022experimental,caro2023out}, quantum optimization~\cite{farhi2014quantum,zhou2020quantum,kotil2025quantum} and quantum chemistry~\cite{lanyon2010towards,hempel2018quantum,huang2023efficient}}. A typical application of PQCs is in the framework of variational quantum algorithms (VQAs)~\cite{cerezo2021variational,tilly2022variational}: one defines a class of PQCs (also referred to as ansatz), encodes the target problem into a loss function expressed as an observable expectation value measured on the outputs of the PQCs, and then iteratively updates the circuit parameters using a classical optimization algorithm to minimize the loss function. {Parameter updates are often based on gradient information, which can be evaluated using the parameter shift rule~\cite{mitarai2018quantum,schuld2019evaluating}.}

However, the optimization of many PQCs suffers from the problem known as ``barren plateaus''~\cite{mcclean2018barren,larocca2025barren,cerezo2021cost}, where the landscape of the loss function becomes exponentially concentrated.
Mathematically, a PQC is said to exhibit a barren plateau if its loss function $L(\bm{\theta})$, with $\bm{\theta}=(\theta_1,\theta_2,\ldots)$, satisfies that {for all $\bm{\theta}$}, the variance of its partial derivative decays exponentially with the system size $n$, i.e.,
\begin{equation*}\label{eq:BP_definition}
\operatorname{Var}_{\bm{\theta}}\left[\partial_{\theta_i} L(\bm{\theta})\right] \leqslant F(n), 
  \quad \text{with} \quad F(n) \in \mathcal{O}\left(\frac{1}{b^n}\right),
\end{equation*}

\noindent where $b>0$ is some constant. 
By Chebyshev’s inequality, $P_{\bm{\theta}}\left(\left|\partial_{\theta_i} L(\bm{\theta})\right| \geqslant \epsilon\right) \leqslant \frac{\operatorname{Var}_{\bm{\theta}}\left[\partial_{\theta_i} L(\bm{\theta})\right]}{\epsilon^2} \leqslant \mathcal{O}\left(\frac{1}{\epsilon^2 b^n}\right)$. Thus, the probability of encountering a nontrivial gradient decreases exponentially with system size. As a consequnce, the {designed} PQC is not trainable. 

To overcome this problem, various strategies have been proposed, such as the use of shallow circuits~\cite{cong2019quantum,pesah2021absence,cerezo2021cost,zhao2021analyzing,liu2022presence}, correlated parameter initialization schemes~\cite{grant2019initialization,zhang2022escaping,wang2024trainability,sauvage2021flip,cao2025exploiting}, restrictions of the circuit dynamics to small Lie algebras~\cite{larocca2022diagnosing,raj2023quantum,fontana2024characterizing,jing2025quantum}, and non-unitary constructions~\cite{deshpande2024dynamic,zapusek2025scaling,yan2025variational}.
However, most of these PQCs {circumvent  barren plateaus at the cost of} expressibility—typically defined as the ability of a PQC to explore the Hilbert space~\cite{sim2019expressibility,holmes2022connecting,Yu_2024}—or by embedding symmetries into the circuit architecture. 
Consequently, such barren-plateau-free constructions often make the circuit dynamics efficiently simulable on a classical computer~\cite{cerezo2025does,angrisani2024classically}. 
This situation naturally raises a fundamental question: can we design a class of PQCs that achieves high expressibility and trainability simultaneously, while remaining classically intractable?

In this work, we {provide an affirmative answer to this question through the construction of} \textit{modified parameterized quantum circuits}~(MPQCs), which incorporate  trainable quantum channels—referred to as gadgets $\mathcal{G}(\bm{\theta})$—into the original PQC, as illustrated in \cref{fig:procedures}(a). 
Starting from {an arbitrary} PQC that may exhibit barren plateaus, an MPQC is constructed by inserting a layer of gadgets $\mathcal{G}(\bm{\theta})$ acting on each qubit.
It turns out that the resulting circuit architecture is guaranteed to be at least as expressive as the original PQC. Moreover, we prove that classically simulating the MPQC is at least as hard as simulating the original PQC, in both the worst case and the average case, implying that typical MPQCs remain classically intractable.

{Crucially, through rigorous analysis we prove that the MPQC is} free from barren plateaus when the gadget layer is properly configured. Furthermore, we show that the introduction of gadgets universally enhances the trainability of PQCs. Specifically, the gradient variance of parameters following the gadget layer is always lower bounded by $\Omega\left(1/\mathrm{poly}(n)\right)$, while for the remaining parameters, the {gradient} variance retains at least its original scaling. 
{Given that some of the latter} may remain untrainable, we introduce a practical strategy to activate them {to be trainable}, thereby enabling {the optimization of all the parameters in} the circuit (see \cref{fig:procedures}(b)). Notably, we further prove that the trainability of MPQCs is robust to noise, meaning that {they work well} even in the presence of realistic noise. 

We perform numerical simulations to demonstrate the effectiveness of our approach in eliminating barren plateaus. Using a {specific} PQC ansatz for thermal-state preparation~\cite{riera2012thermalization,sagastizabal2021variational,motta2020determining}, we estimate both the variance of the loss function and that of the gradient in {both the original PQC and the MPQC} via the Monte Carlo method~\cite{shao2025diagnosing}. The results show that {our approach successfully eliminates} barren plateaus even for circuits with up to 100 qubits and 2400 layers, in sharp contrast to the exponential gradient {decay} observed in the original PQC. Lastly, we emphasize that when applying MPQC to various variational quantum algorithms in an end-to-end fashion, we consistently observe improved performance compared with the original PQCs.

\begin{figure*}[!t]
  \centering
  \includegraphics[width =\linewidth]{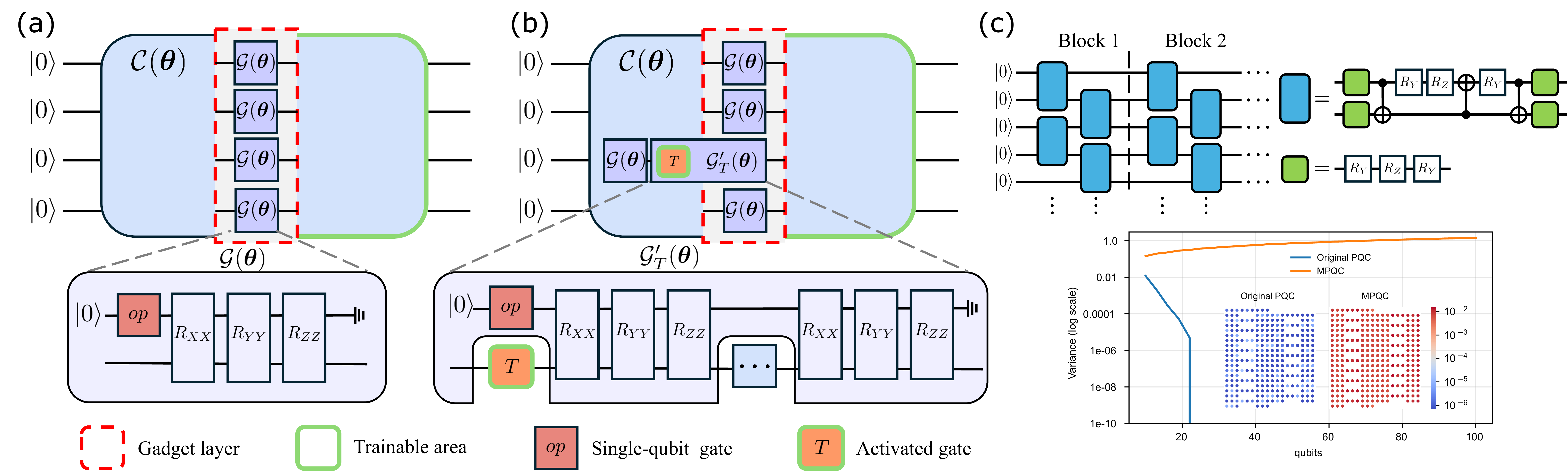}
  \caption{\justifying (a) Structure of an MPQC, where a layer of gadgets $\mathcal{G}(\bm{\theta})$ (outlined by the red dashed box) is inserted into the original PQC $\mathcal{C}(\bm{\theta})$ (indicated by the light-blue region).
Each gadget $\mathcal{G}(\bm{\theta})$ (highlighted in light purple) contains an ancilla qubit initialized in $\ket{0}$, one single-qubit unitary $op$, and three two-qubit rotation gates $R_{XX}$, $R_{YY}$, and $R_{ZZ}$. 
{The symbol $\ancillasym$ denotes the ancilla is discarded.}
(b) {Structure of a $T$-activating MPQC.
The gates denoted by “$\cdots$” represent those in the original PQC located between $T$ and the gadget layer.
In this MPQC, we specifically enlarge the gadget $\gadget{\bm{\theta}}$ acting on the same qubit as $T$, transforming it into $\gadgetT{\bm{\theta}}$.} (c) Top: Ansatz circuit for thermal state preparation, where the number of blocks equals the number of qubits $n$. 
Bottom: Variance comparison between PQCs and MPQCs for thermal state preparation, where the MPQCs are formed by inserting a gadget layer before the final block.
{Yellow and blue curves show the cost-function variances of PQCs and MPQCs, respectively, estimated via the method in Ref.~\cite{shao2025diagnosing}.}
The blue curve is omitted for $n>21$, as its values are extremely close to zero in this regime.
The inset presents the gradient variance of parameters located after the gadget layer for $n=20$.}\label{fig:procedures}
\end{figure*}

\section{Modified Parameterized Quantum Circuits (MPQCs)}

A PQC $\mathcal{C}(\bm{\theta}) = U_{m}(\theta_{m}) \cdots U_1(\theta_1)$ consists of a sequence of unitaries $U_i(\theta_i)$ parameterized by $\bm{\theta} = (\theta_1, \theta_2, \ldots, \theta_{m})$, where each $\theta_i \in \left[0, 2\pi \right)$ specifies a rotation angle and $m$ denotes the number of parameters. In this work, each unitary $U_i(\theta_i)$ is taken to be a Pauli rotation of the form $e^{-i\frac{\theta_i}{2} P}$, where $P \in \{\mathbb{I}, X, Y, Z\}^{\otimes n}$ with $n$ being the number of qubits, followed by a non-parameterized Clifford gate $C_i$. We do not impose any restriction on the form of the input state $\rho$ of $\mathcal{C}(\bm{\theta})$, meaning that it can be a noisy or mixed state. In practice, $\rho$ is typically chosen as the state $\ket{0^n}\bra{0^n}$.
The parameters $\bm{\theta}$ are optimized by minimizing a loss function of the form $\loss = \tr{O \mathcal{C}(\bm{\theta})\rho \mathcal{C}(\bm{\theta})^\dagger}$ with {$O$ being an observable}.
The variance of the loss function and that of the gradient can be expressed as:
\begin{small}
\begin{equation}
  \begin{aligned}
    \operatorname{Var}_{\bm{\theta}} \left[ \loss \right] &= \mathbb{E}_{\bm{\theta}}\left[\loss^2\right] - \left(\mathbb{E}_{\bm{\theta}}\left[\loss\right]\right)^2 \\
    \operatorname{Var}_{\bm{\theta}} \left[ \frac{\partial \loss}{\partial \theta_j} \right] &= \mathbb{E}_{\bm{\theta}}\left[\left(\frac{\partial \loss}{\partial \theta_j}\right)^2\right] - \left(\mathbb{E}_{\bm{\theta}}\left[\frac{\partial \loss}{\partial \theta_j}\right]\right)^2,
  \end{aligned}
\end{equation}
\end{small}

\noindent where each $\theta_i$ is sampled {uniformly from} $\left[0, 2\pi\right)$.

To mitigate barren plateaus, a gadget layer {is inserted at a chosen position of the PQC (the location will be specified later).} This layer consists of $n$ gadgets and can be written as $ \bigotimes_{i=1}^{n} \mathcal{G}_i(\bm{\theta}_{\mathcal{G}_i}) $, where each gadget contains one single-qubit operation $op$ and three two-qubit rotation gates, parameterized by $\bm{\theta}_{\mathcal{G}_i} = (\theta_{\mathcal{G}_{i,1}}, \theta_{\mathcal{G}_{i,2}}, \theta_{\mathcal{G}_{i,3}})$, as illustrated in \cref{fig:procedures}(a).
The single-qubit operation $op$ {can be any quantum operation} that satisfies the following condition:
there exists a constant $\tau > 0$ such that:
\begin{small}
 \begin{equation}\label{eq:gadget_input_parameter}
  \tr{{op \left(\ketbra{0}\right)}P}^2 \geq \tau,  \quad \forall P\in \{X,Y,Z\}.
\end{equation}
\end{small}
In \sm~A, we present two constructions of $op$ {using single-qubit gates.} The first is a fixed unitary gate that achieves the maximal value of $\tau$, while the second introduces {two parameterized single-qubit rotation gates}, rendering $op$ trainable.

It is straightforward to verify that the expressibility of an MPQC is at least as large as that of the original PQC. 
Let $\channelM{\bm{\theta},\bm{\theta}_{\mathcal{G}}}$ denote the channel corresponding to the MPQC obtained by augmenting $\mathcal{C}(\bm{\theta})$ with a gadget layer, where $\bm{\theta}_{\mathcal{G}} = (\bm{\theta}_{\mathcal{G}_1}, \bm{\theta}_{\mathcal{G}_2}, \ldots, \bm{\theta}_{\mathcal{G}_n})$ collects the parameters of the $n$ gadgets. 
For an arbitrary input state $\rho$, we have $\mathcal{C}(\bm{\theta}) \rho \mathcal{C}^\dagger(\bm{\theta}) = \channelM{\bm{\theta},\bm{0}}(\rho)$. 
Hence, the output state ensemble generated by $\mathcal{C}(\bm{\theta})$ is a subset of that generated by $\channelM{\bm{\theta},\bm{\theta}_{\mathcal{G}}}$.

\section{Absence of barren plateaus in MPQC}

We now establish the following theorem, which demonstrates that introducing a gadget layer can eliminate barren plateaus in arbitrary PQCs, thereby restoring their trainability.
{The detailed proof is provided in \sm~E.}

\begin{theorem}\label{thm:absence_BP}[informal]
For an arbitrary $\mathcal{C}(\bm{\theta})$, if the corresponding MPQC $\channelM{\bm{\theta},\bm{\theta}_{\mathcal{G}}}$ satisfies the following conditions:
\begin{itemize}
  \item The observable $O = \sum_{\alpha} c_\alpha P_\alpha$ is local, i.e., $O$ is the sum of Pauli words $\{P_\alpha\}_{\alpha}$ with {each nontrivially acting on at most $\order{1}$ qubits}. 
  \item For each Pauli term $P_\alpha$ in $O$, the support size of its backward light cone at the gadget layer, {i.e., the number of qubits in the layer whose perturbations can affect the measurement outcome of $P_\alpha$}, is upper bounded by $K = \order{\log n}$.
\end{itemize}
Then the variance of its loss function $\lossM \coloneq \tr{\channelM{\bm{\theta},\bm{\theta}_\mathcal{G}}(\rho)O}$ admits the lower bound:
\begin{equation}\label{eq:lower_bound_var}
  \operatorname{Var}_{(\bm{\theta},\bm{\theta}_\mathcal{G})} \left[ \lossM \right]\geq \sum_\alpha c_\alpha^2\left(\frac{\tau}{4}\right)^{K} = \Omega\left(\frac{1}{\mathrm{poly}(n)}\right).
\end{equation}
As a consequence, according to Ref.~\cite{arrasmith2022equivalence}, Eq.~\eqref{eq:lower_bound_var} ensures the absence of barren plateaus in $\channelM{\bm{\theta},\bm{\theta}_{\mathcal{G}}}$.
\end{theorem}

{Later in this section, we will show that the support-size condition can be easily satisfied by appropriately placing the gadget layer.}
Although MPQCs are inherently free from barren plateaus, the trainability of individual parameters needs further investigation. Here, we address this issue by examining $\operatorname{Var}_{(\bm{\theta},\bm{\theta}_{\mathcal{G}})} \left[ \frac{\partial \lossM}{\partial{\theta_j}} \right]$.
The results are summarized in the following theorem, {with the full proof presented in \sm~F.}

\begin{theorem}\label{thm:absence_BP_MPQC_parameters}
  Consider an MPQC $\channelM{\bm{\theta},\bm{\theta}_{\mathcal{G}}}$ and a local observable $O = \sum_{\alpha} c_\alpha P_\alpha$.
If the support-size condition in \cref{thm:absence_BP} holds, and for each $P_\alpha$ the segment of its backward light cone from $P_\alpha$ to the gadget layer contains at most $\order{\log n}$ parameters in gates, then $\lossM$ satisfies the following properties:
  \begin{itemize}
    \item  For parameter $\theta_j$ located after the gadget layer, if $ \operatorname{Var}_{\bm{\theta}} \left[ \frac{\partial \loss}{\partial{\theta_j}} \right] \neq 0$, then we have:
      \begin{equation}\label{eq:gra_var_lower_bound_pars_after}
        \operatorname{Var}_{(\bm{\theta},\bm{\theta}_\mathcal{G})} \left[ \frac{\partial \lossM}{\partial{\theta_j}} \right] \geq \Omega\left(\frac{1}{\mathrm{poly}(n)}\right).
      \end{equation}
    \item For parameter $\theta_j$ located before the gadget layer, there is:
      \begin{equation}\label{eq:gra_var_lower_bound_pars_before}
        \operatorname{Var}_{(\bm{\theta},\bm{\theta}_\mathcal{G})} \left[ \frac{\partial \lossM}{\partial{\theta_j}} \right] \geq  \Omega\left(\frac{1}{\mathrm{poly}(n)}\right)\operatorname{Var}_{\bm{\theta}} \left[ \frac{\partial \loss}{\partial{\theta_j}} \right].
      \end{equation}
  \end{itemize}
\end{theorem}

Theorem~\ref{thm:absence_BP_MPQC_parameters} implies that modifying an original PQC into an MPQC necessarily improves its trainability. 
Specifically, Eq.~\eqref{eq:gra_var_lower_bound_pars_after} guarantees the trainability of parameters located after the gadget layer, while Eq.~\eqref{eq:gra_var_lower_bound_pars_before} ensures that the resulting circuit is not effectively restricted to a shallow architecture: The parameters before the gadget layer remain trainable whenever they are trainable in the original PQC. 
Crucially, as demonstrated by our subsequent numerical simulations, MPQC significantly outperforms shallow circuits, indicating that these parameters continue to play an essential role during training.

To satisfy the conditions of Theorem~\ref{thm:absence_BP_MPQC_parameters}, the placement of the gadget layer can be determined according to the geometric structure of the circuit. 
In Supplementary Material~G, we provide an explicit construction for a broad class of PQCs defined on (hyper)cubic lattices. 
As a specific example, for a one-dimensional brick-wall PQC, the gadget layer should be placed at a distance of order $\mathcal{O}((\log n)^{1/2})$ from the final measurement layer.



\section{Strategy to activate  untrainable parameters}\label{sec:activation}

{Note that} \cref{thm:absence_BP_MPQC_parameters} does not guarantee that parameters located before the gadget layer have nonvanishing gradients.
In the worst case, one may still encounter a single-qubit rotation gate $T = R_{P_T}(\theta_T)$ before the gadget layer whose gradient variance is nearly zero. 
To address this issue, we present a targeted procedure to ``activate'' $T$, which is a strategy that significantly increases the trainability of $T$.

{As illustrated in \cref{fig:procedures}(b), this is accomplished by inserting an additional gadget $\gadget{\bm{\theta}}$ immediately before the target gate $T$ and enlarging one gadget in the gadget layer through the following procedure: we first move the $op$ operation of this gadget layer to the same layer as $T$, and then append three two-qubit parameterized rotation gates---$R_{XX}$, $R_{YY}$, and $R_{ZZ}$---immediately after $op$ and $T$, thereby transforming it into a new type of gadget, denoted as $\gadgetT{\bm{\theta}}$. The position of the enlarged $\gadget{\bm{\theta}}$ can be selected flexibly to suit physical implementation convenience. In fact, any $\gadget{\bm{\theta}}$ located within the backward light cone of some Pauli term $P_\alpha$ in $O$ qualifies as a valid candidate, as elaborated in \sm~H.}

{We refer to the resulting circuit as the \emph{$T$-activating MPQC}.
Let the corresponding quantum channel be $\channelMT{\bm{\theta},\bm{\theta}_{\mathcal{G}},\bm{\theta}_{\mathcal{G}'_T}}$ , where $\bm{\theta}_{\mathcal{G}'_T}$ collects the parameters in the enlarged gadget $\gadgetT{\bm{\theta}}$, and $\bm{\theta}_{\mathcal{G}}$ collects the parameters in all the $\gadget{\bm{\theta}}$ gadgets, including the one inserted before $T$.}
We define its loss function as $\lossMT \coloneqq \tr{\channelMT{\bm{\theta},\bm{\theta}_{\mathcal{G}},\bm{\theta}_{\mathcal{G}'_T}}(\rho)O}$. 
The following theorem guarantees that $\theta_T$ is trainable, and the proof is given in \sm~H.

\begin{theorem}\label{thm:activate_parameters}
  Consider a $T$-activating MPQC $\channelMT{\bm{\theta},\bm{\theta}_{\mathcal{G}},\bm{\theta}_{\mathcal{G}'_T}}$, evaluated with respect to a local observable $O$. Suppose that the conditions stated in \cref{thm:absence_BP_MPQC_parameters} hold.
Let $T = R_{P_{T}}(\theta_{T})$ denote the single-qubit rotation gate to be activated. Then, we have
\begin{equation}\label{eq:gra_var_lower_bound_activate_pars}
  \operatorname{Var}_{(\bm{\theta},\bm{\theta}_{\mathcal{G}},\bm{\theta}_{\mathcal{G}'_T})} \left[\frac{\partial\lossMT}{\partial{\theta_T}}\right]\geq \Omega\left(\frac{1}{\mathrm{poly}(n)}\right).
\end{equation}
\end{theorem}

{In practice, \cref{thm:activate_parameters} provides a strategy to adaptively modify the MPQC architecture, enabling the training of all the parameters in the circuit.}
The procedure can be implemented as follows.
We first train the MPQC $\channelM{\bm{\theta},\bm{\theta}_{\mathcal{G}}}$ to minimize the loss function. 
As stated before, certain parameters before the gadget layer may remain untrainable. 
If such a parameter is identified, we can apply the activation strategy to make it trainable. Moreover, by initializing the newly introduced parameters in $\channelMT{\bm{\theta},\bm{\theta}_{\mathcal{G}},\bm{\theta}_{\mathcal{G}'_T}}$ to zero, the loss function retains the same value as that of $\channelM{\bm{\theta},\bm{\theta}_{\mathcal{G}}}$. This enables us to further minimize the loss funtion. {If multiple untrainable parameters are identified, the same activation procedure can be successively applied.}

Furthermore, this strategy can naturally extend {to} multi-qubit rotation gate and multiple parameters: by inserting several $\gadget{\bm{\theta}}$ and several gadgets of the form $\gadgetT{\bm{\theta}}$, multiple parameters can be simultaneously activated within a single MPQC. {Details of the construction for this strategy and its application can be found in \sm~I and \sm~M, respectively.}

\section{Noise robustness}

MPQCs and its variants still work well even in the presence of noise on quantum devices, making them an applicable tool for quantum machine learning in the NISQ era. 
This property is {formalized in} the following theorem, with the complete proof provided in \sm~J.

\begin{theorem}[informal]\label{thm:absence_BP_robustness}
{Suppose the MPQC is subject to Pauli noise of strength at most $\gamma < 1/2$ after each $U_i(\theta_i)$ and every gate within the gadgets.
Then, the lower bounds on the (gradient) variance of the loss function established in \cref{thm:absence_BP,thm:absence_BP_MPQC_parameters,thm:activate_parameters} deteriorate by at most a multiplicative factor of $(1-2\gamma)^{\order{\log n}} = \Omega\left(\frac{1}{\mathrm{poly}(n)}\right)$.}
\end{theorem}

{As a consequence of \cref{thm:absence_BP_robustness}, the variance and the gradient variance of the loss function of the noisy MPQCs are still lower bounded by $\Omega\left(\frac{1}{\mathrm{poly}(n)}\right)$, implying the merits hold for the MPQC even in the noisy setting.}
It is worth noting, however, that MPQCs {(and, more generally, arbitrary PQCs)} with a constant noise rate are classically simulable in an average sense~\cite{shao2024simulating}; that is, over the MPQC loss function landscape, outputs corresponding to most parameter settings are classically predictable. 
Nevertheless, some specific parameter configurations of {noisy MPQCs} may still retain quantum advantage, {indicating their potential value for deployment on near-term quantum devices.} Similar discussions can also be seen in Ref.~\cite{deshpande2024dynamic}.

\begin{figure*}[!t]
    \centering
    \includegraphics[width=\linewidth]{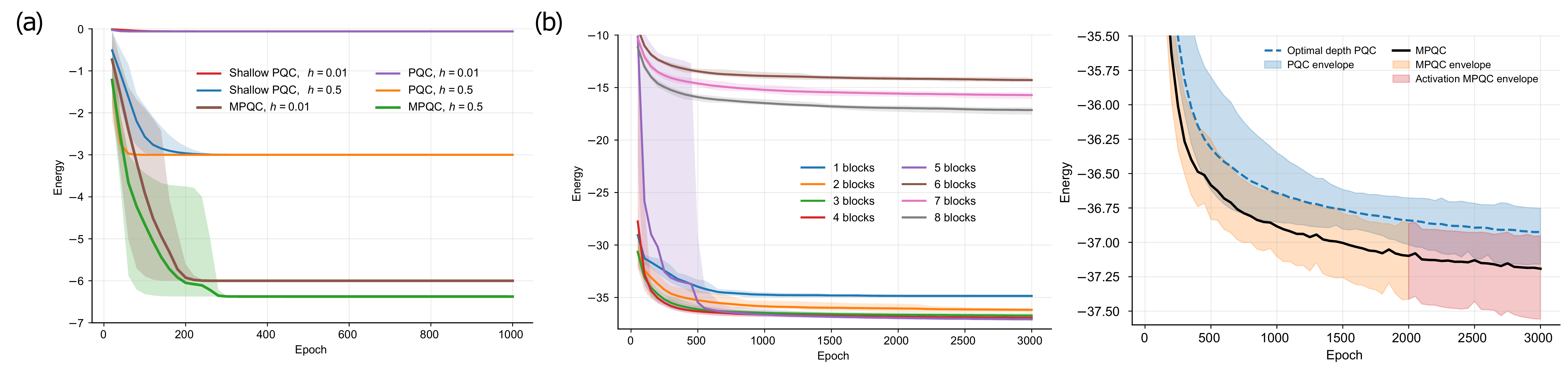}
    \caption{\justifying 
Performance of PQC and MPQC in variational quantum algorithms. 
All results are obtained from ten independent random parameter initializations. 
Shaded regions indicate the min--max envelope of the loss across different initializations, 
and lines represent the corresponding mean values.
(a) Variational training of a poorly designed PQC and the corresponding MPQC for $H_{\mathrm{TFI}}$ with $h=0.01$ and $h=0.5$.
In both cases, MPQC converges to the ground energy, whereas the original PQC fails.
(b) Left: Performance of two-dimensional PQCs with different numbers of repeated blocks in approximating the ground-state energy of $H_G$.
As the number of blocks increases, the lowest achievable energy first decreases and then increases.
Right: Comparison between the best-performing PQC in the left plot and the corresponding MPQC.
During the first $2000$ epochs, MPQC is trained without activation, followed by $1000$ epochs with activation, where newly introduced parameters are initialized to zero.
The final energy achieved by MPQC is lower than that of the best-performing PQC by $0.39$.
}\label{fig:vqa_results}
\end{figure*} 
\section{Numerical simulations}
In this section, we provide numerical evidence showing that MPQCs can effectively eliminate barren plateaus in PQCs. More importantly, we demonstrate that the MPQCs we construct for quite a few variational algorithms outperform the original PQCs significantly. These results highlight that MPQC is a promising approach for constructing trainable and expressive variational parameterized quantum circuits.

\subsection{Evidence that MPQC Eliminates Barren Plateaus}
We {first} conduct numerical experiments that compare the variances of the loss function and {those of the} parameter gradients between PQCs and MPQCs. 
The numerical study focuses on PQCs for thermal state preparation, a task known to be NP-hard~\cite{galanis2016inapproximability}. 
We consider the 2-local transverse field Ising model
\begin{equation}\label{eq:H_TFI}
  H_{\mathrm{TFI}}=-\sum_{j=1}^n X_j X_{j+1}-{h} \sum_{j=1}^n Z_j
\end{equation}
defined on a periodic 1D chain with system sizes ranging from {$n=$ 10 to 100} qubits, {where $h$ denotes the transverse-field strength controlling the relative weight of the single-qubit field term~\cite{sachdev1999quantum,pfeuty1970one}, and is fixed to be $h = 1/2$ in this case.}
The circuits architecture is shown in \cref{fig:procedures}(c).
Following Ref.~\cite{ilin2024dissipative}, which reports the small preparation errors when the number of blocks equals $n$, we set the number of blocks to be $n$. 
To eliminate the barren plateaus present in such a PQC, we construct an MPQC by inserting a gadget layer after the $(n-1)$-th block, followed by an additional block.

{We employ the numerical method of Ref.~\cite{shao2025diagnosing}, which offers efficient classical method to estimate both the variance and the gradient variance of the loss function of arbitrary PQCs, to compare the trainability of the PQC and that of the corresponding MPQC. As shown in \cref{fig:procedures}(c), we first see that the variance of the loss function for the original PQC (blue curve) decreases rapidly as the number of qubits increases. Our numerical results show that it becomes negligibly small when $n > 21$, indicating the onset of barren plateaus.
In stark contrast, the variance of the loss function for the MPQC (yellow curve) remains stable (approximately $1$) and even exhibits a slight increase with the number of qubits. This behavior demonstrates that the MPQC effectively avoids barren plateaus and preserves trainability across increasing system sizes in this case.}

Furthermore, for $n=20$, we evaluate the gradient variances of the gradients associated with the parameters located after the gadget layer. The results are shown in the inset of \cref{fig:procedures}(c), where all the red points corresponding to the MPQC have gradient variances of the order of $10^{-2}$, whereas the blue points for the original PQC fall below $10^{-4}$. Collectively, these numerical results provide strong evidence that MPQCs are highly effective in enhancing the trainability of PQCs.

\subsection{Evidence that MPQC outperforms original PQC in variational algorithms}
In this subsection, through comprehensive numerical simulations across all stages, we demonstrate that MPQC can substantially improve the performance of PQCs in variational quantum algorithms. 
Owing to the limitations of classical numerical simulation and the additional ancilla qubits required by MPQC, 
our simulations are restricted to systems of up to 24 qubits. 
At these system sizes, gradients that vanish exponentially with the number of qubits may not yet be extremely small. 
Nevertheless, even for medium-size PQCs, the cost-function gradient can still be close to zero when the ansatz is poorly designed, which leads to severe training difficulties and gives us chance to test the performance of MPQC.

For this, we first construct a deliberately unfavorable ansatz to approximate the ground-state energy 
of the Hamiltonian in \cref{eq:H_TFI}. 
For this ansatz, the corresponding PQC becomes untrainable when the transverse field strength $h$ is close to zero. 
By inserting a gadget layer into the circuit, we obtain the corresponding MPQC. 
In this example, we let $n=6$ (the number of qubits) and consider two representative values of $h$, which are $h=0.01$ and $h=0.5$. 
Details of the circuit construction and the training procedure are provided in \sm~M.

As shown in \cref{fig:vqa_results}(a), when $h=0.01$, even the shallow PQC cannot be trained properly, indicating the presence 
of vanishing gradients. Moreover, increasing the field strength to $h=0.5$ does not resolve this issue: the poorly designed PQC 
still fails to converge to the ground-state energy, as illustrated by the blue and orange curves. 
In contrast, for the both values of $h$, MPQC consistently converges to the exact ground-state 
energy up to a small error of $0.01$. 
These results demonstrate that MPQC remains effective even when the underlying ansatz is improperly designed.

To provide further evidence that MPQC can outperform the original PQC, we next consider the task of approximating the ground-state energy of a more complex Hamiltonian $H_G$ discussed in Eq.\cite{KempeKitaevRegev2006}, which is QMA-complete. Here, $G$ specifies the underlying geometry of the Hamiltonian.
To generate the ground state, we construct a family of two-dimensional ansatze on $12$ qubits, composed of repeated circuit blocks, with the number of blocks ranging from $1$ to $8$. When the block number is $8$, the corresponding MPQC is obtained by inserting a gadget layer after the fourth block of the PQC. In addition, we employ the activation strategy described in \cref{sec:activation} to further enhance the performance of MPQC, which doubles the qubit number.
For a fair comparison, all the PQCs are trained for $3000$ epochs, and the MPQC is trained for $2000$ epochs without activations, followed by $1000$ more epochs with activations, resulting in the same total number of optimization steps. Details on the Hamiltonian $H_G$, the two-dimensional PQC ansatz, the construction of MPQC, and the training process are provided in \sm~M.

As shown in the left panel of \cref{fig:vqa_results}(b), the final energy obtained by the original PQC initially decreases as the number of blocks increases, reflecting the improved expressibility of deeper circuits. However, when the number of blocks exceeds five, the performance deteriorates, indicating the onset of severe trainability issue.

As a sharp comparison, the right panel of \cref{fig:vqa_results}(b) shows that MPQC can address the trainability issue very well, which remains trainable at all the depths considered here. Moreover, its loss function decreases more rapidly and reaches significantly lower values than those achieved by the best-performing PQC (with five blocks). We further observe that the activation strategy enables additional optimization progress: without activation, the MPQC loss remains nearly constant after approximately $2000$ training epochs, whereas activating additional parameters allows the loss function to decrease further.

\section{Conclusions and Discussions}
In this work, we have introduced {a novel, easily implementable, and universal strategy} to improve the trainability of an arbitrary PQC. By inserting a layer of gadgets, we transform the PQC into an MPQC that is at least as expressive as the original PQC and, importantly, is provably free of barren plateaus.
We further analyze the trainability of parameters in the MPQC, showing that our construction consistently enhances trainability: parameters following the gadget layer are guaranteed to be trainable, {whereas the others retain the same learning behavior as in the original PQC.}
{Moreover, we further propose a targeted strategy to render these remaining parameters trainable, ensuring that all the parameters can be effectively optimized}.

The improvement in trainability brought by constructing MPQCs is supported by our numerical experiments. Focusing on a PQC for thermal state preparation, we find that barren plateaus are absent in the MPQC even for deep circuits with up to 100 qubits and 2400 layers, whereas the original PQC becomes untrainable when the system size reaches 20 qubits. 
Furthermore, by end-to-end numerical simulations we show that MPQC can substantially enhance the performance of the original PQC in variational quantum algorithms. In particular, in some cases we see that MPQC is able to converge to the optimal solution even when the corresponding PQC cannot be trained at all.

Our theoretical analysis {and numerical verifications position} MPQCs as a promising circuit architecture for PQC-based quantum algorithms.
However, several interesting questions remain. 
First, {we have shown that the set of the output state of an MPQC subsumes that of the original PQC, implying that classical simulation of the MPQC is at least as hard as that of the original PQC.
Actually, we have also theoretically demonstrated that the average-case classical simulation of the MPQC leads to that of the original PQC (see \sm~L for details).
These results may shed new light on the relationship between average-case classical simulation complexity and barren plateaus, as recently discussed in Ref.~\cite{cerezo2025does}.
}
Second, as highlighted in Ref.~\cite{deshpande2024dynamic}, the absence of barren plateaus alone does not guarantee that a quantum algorithm will converge to the optimal solution, since the loss landscape can still exhibit significant complexity.
In future work, we will investigate the internal mechanisms of MPQCs to {examine} their effects on the loss function landscape, with the aim of {understanding} the convergence behavior of training MPQCs.

\begin{acknowledgments}

We thank Weikang Li, Zhengfeng Ji, Ruiqi Zhang, Fuchuan Wei and Weixiao Sun for valuable discussions. Z.W was supported by Beijing Science and Technology Planning Project (Grant No. Z25110100810000). Z.L was supported by NKPs (Grant No. 2020YFA0713000). Z.W and Z.L were supported by Beijing Natural Science Foundation (Grant No. Z220002).
Y.S, and Z.L were supported by BMSTC and ACZSP (Grant No. Z221100002722017).  Z.C and Z.W
were supported by the National Natural Science Foundation of China (Grant Nos. 62272259 and 62332009).

\end{acknowledgments}
\bibliography{main}
\clearpage
\setcounter{theorem}{0}
\renewcommand{\thetheorem}{A.\arabic{theorem}}
\renewcommand{\thelemma}{A.\arabic{lemma}}
\renewcommand{\thecorollary}{A.\arabic{corollary}}
\renewcommand{\thedefinition}{A.\arabic{definition}}
\renewcommand{\thefigure}{A.\arabic{figure}}
\renewcommand{\thetable}{A.\arabic{table}}

\renewcommand{\thetheorem}{A.\arabic{theorem}}
\renewcommand{\thelemma}{A.\arabic{lemma}}
\renewcommand{\thecorollary}{A.\arabic{corollary}}
\renewcommand{\thedefinition}{A.\arabic{definition}}
\renewcommand{\thefigure}{A.\arabic{figure}}
\renewcommand{\thetable}{A.\arabic{table}}

\onecolumngrid 
\section*{Supplemental Material}
\appendix

\tableofcontents

\clearpage

\section{Circuit architectures}\label{app:circuit_details}
\subsection{Parameterized quantum circuit (PQC)}
A typical $n$-qubit PQC, denoted as $\mathcal{C}(\bm{\theta})$, consists of a sequence of Pauli rotation gates and non-parameterized Clifford gates.
The Pauli rotation gates are represented as $e^{-i\frac{\theta}{2} P}$, where $P\in\{\mathbb{I},X,Y,Z\}^{\otimes n}$. The Clifford gates are the unitary operators that normalize the Pauli group $Cl_n \coloneqq \{C\in U_{2^n} \mid C\mathcal{P}_nC^\dagger=\mathcal{P}_n\}$, where $\mathcal{P}_n$ is the Pauli group on $n$ qubits. Any unitary operator $U\in Cl_n$ is equivalent to a circuit generated using Hadamard, CNOT, and phase gates $S$~\cite{gottesman2016surviving}.

Without loss of generality, we assume that PQCs follow the form:
\begin{equation}\label{eq:parameterized_circuit}
  \mathcal{C}(\bm{\theta})=U_m(\theta_m)  \cdots {U}_1(\theta_1),
\end{equation}
where $\bm{\theta}=(\theta_1,\cdots,\theta_m)$ are rotation angles and $m$ is the number of the parameters. Each unitary ${U}_i(\theta_i):= R_{P_i}(\theta_i)C_i $ comprises a Clifford operator $C_i$ and a rotation $R_{P_i}(\theta_i):=\exp{-i \frac{\theta_i}{2} P_i}$ on Pauli operator $P_i\in\{\mathbb{I},X,Y,Z\}^{\otimes n}$ with angle $\theta_i$. 

In this context, the quantum circuit $\mathcal{C}(\bm{\theta})$ is applied to an initial state $\rho$, and what we are interested in is the expectation value of an observable $O$, which is given by
\begin{equation}\label{eq:expectation_value_noiseless}
  \Vexp{O} = \tr{O \mathcal{C}(\bm{\theta})\rho \mathcal{C}(\bm{\theta})^\dagger}. 
\end{equation}
Without loss of generality, we assume that the observable is traceless, i.e., $\tr{O}=0$, otherwise we can replace $O$ with $O-\frac{\tr{O}}{2^n}I$. 

Moreover, we restrict the number of Pauli words constituting the observable $O$ is $\order{\mathrm{poly}(n)}$, since measuring an exponential number of expectation values is experimentally infeasible.
This assumption is satisfied for a wide range of variational quantum algorithms (VQAs), such as the Variational Quantum Eigensolver (VQE)~\cite{peruzzo_variational_2014} and the Quantum Approximate Optimization Algorithm (QAOA)~\cite{farhi_quantum_2014}.
Consequently, for $O = \sum_{\alpha} c_\alpha P_\alpha$, we have
\begin{equation}\label{eq:upper_bound_HS}
\sum_{\alpha} c_\alpha^2 \le \max_\alpha \{c_\alpha^2\} \sum_{\alpha} 1 = \order{\mathrm{poly}(n)}.
\end{equation}

\subsection{Modified parameterized quantum circuit (MPQC)}
By introducing some gadgets to any PQC in form of~\eqref{eq:parameterized_circuit}, we obtain a corresponding modified parameterized quantum circuit (MPQC). A schematic illustration of the MPQC is shown in \cref{fig:modified_pqc}.

\begin{figure}[h]
  \centering
  \includegraphics[width = 0.8\columnwidth]{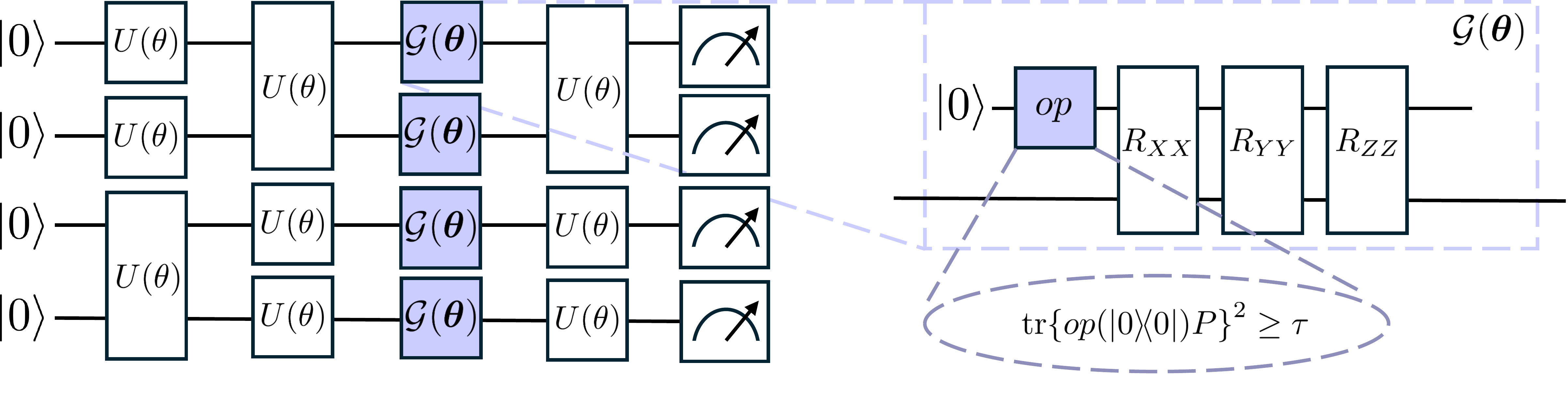}
  \caption{
    An example of an MPQC, where gadgets $\gadget{\bm{\theta}}$ drawn in blue are inserted into the original PQC.
    The gadget contains an ancilla qubit $\ket{0}$, one single qubit gate $op$ and three 2-qubit rotation gates $R_{XX}, R_{YY}, R_{ZZ}$.
  }\label{fig:modified_pqc}
\end{figure}

In \cref{fig:modified_pqc}, the single qubit gate $op$ in gadget $\gadget{\bm{\theta}}$ satisfies the following condition:
\begin{equation}\label{eq:op_condition}
  \min\left\{\tr{op(\ket{0}\bra{0})X}^2,\tr{op(\ket{0}\bra{0})Y}^2,\tr{op(\ket{0}\bra{0})Z}^2\right\} = \tau > 0.
\end{equation}
In the next subsection, we present a construction of \( op \) such that \cref{eq:op_condition} holds with maximum $\tau$ when \( op \) is a unitary. Moreover, we provide an alternative construction that keeps \( op \) trainable. It is easy to see inserting $\gadget{\bm{\theta}}$ to the original PQC will not decrease expressibility, because if the rotation angles in these three 2-qubit gates equal $0$, the PQC in \cref{fig:modified_pqc} is exactly the original PQC. 

We assume that all gadgets $\gadget{\bm{\theta}}$ are inserted after the $l$-th layer of the original circuit, as illustrated in \cref{fig:MPQC_for_proof}. Also for further simpliy the proof, we restrict that the gadget layer is placed after the $L$-th block of the PQC, i.e.:
\begin{equation}\label{eq:Modified_parameterized_circuit}
  \channelM{\bm{\theta},\bm{\theta}_\mathcal{G}}= \mathcal{U}_m(\theta_m)  \circ\mathcal{U}_{m-1}(\theta_{m-1})\cdots \circ\mathcal{U}_{L+1}(\theta_{L+1})\circ \otimes_{i=1}^n\mathcal{G}_i(\bm{\theta}_{\mathcal{G}_i})\circ\mathcal{U}_{L}(\theta_{L})\cdots\circ\mathcal{U}_{1}(\theta_{1}),
\end{equation}
where $\mathcal{U}_{i}(\theta_{i})$ is the channel representation corresponding to the unitary operation $U_{i}(\theta_{i})$, each gadget is parameterized by three angles $\bm{\theta}_{\mathcal{G}_i} = (\theta_{\mathcal{G}_{i,1}}, \theta_{\mathcal{G}_{i,2}}, \theta_{\mathcal{G}_{i,3}})$, and ``$\circ$'' denotes the composition of quantum channels.
\begin{figure}[h]
  \centering
  \includegraphics[width = 0.7\columnwidth]{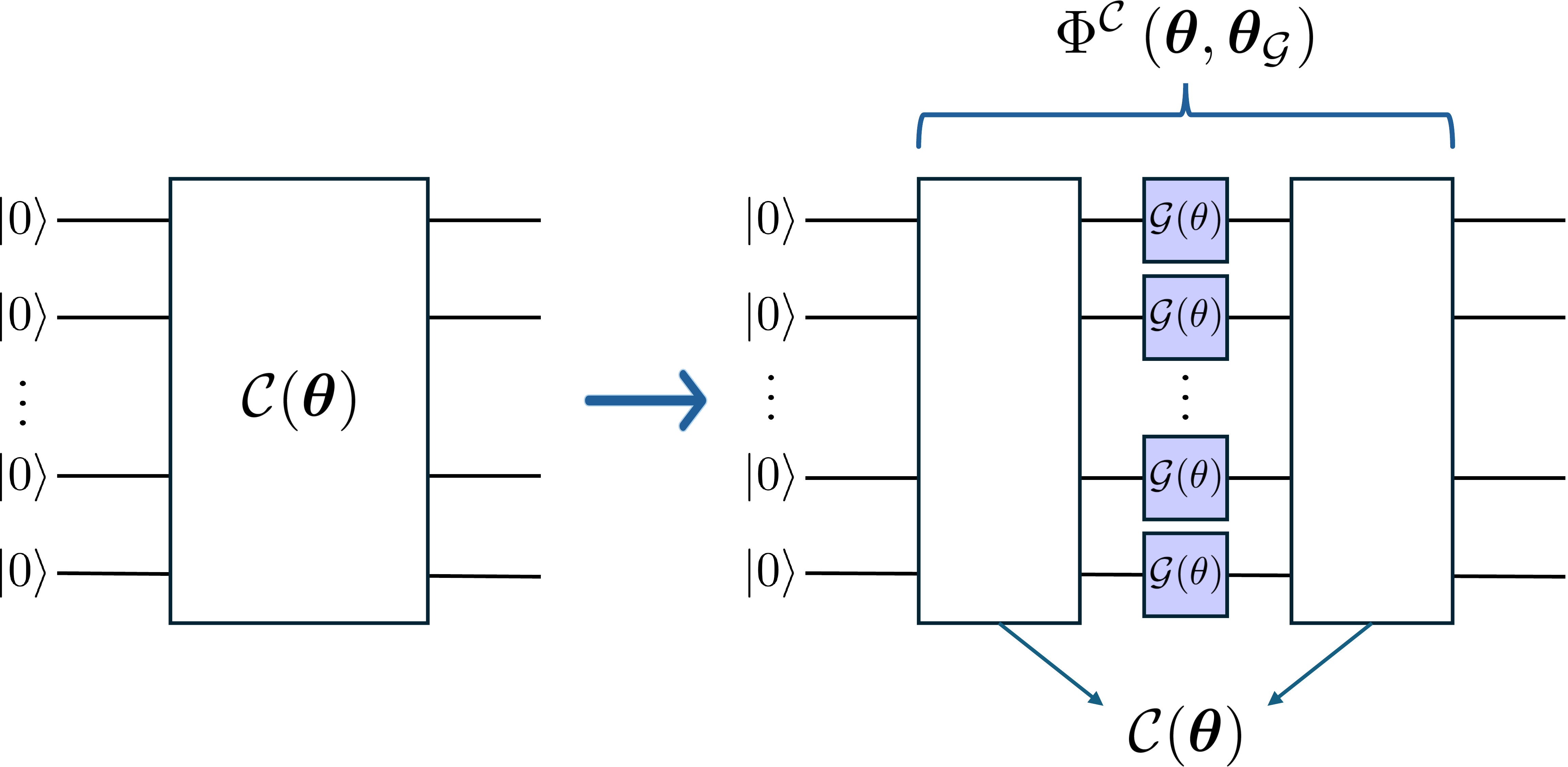}
  \caption{One construction of the MPQC: all gadgets $\gadget{\bm{\theta}}$ are inserted in parallel after the $l$-th layer of the original circuit.
}\label{fig:MPQC_for_proof}
\end{figure}

\subsection{Constructions of $op$}\label{subapp:op}
We now present two constructions of \( op \). The first achieves the maximal value of \( \tau \) using a single unitary gate. The second employs two parameterized single-qubit Pauli rotation gates, offering a hardware-efficient implementation compatible with current quantum devices.

\paragraph{Single qubit unitary} Suppose $op$ is a unitary gate $U$ such that 
\[
U\ket{0}= \cos\psi\ket{0} + \sin\psi e^{i\phi}\ket{1}.
\]
Then we have 
\begin{equation}
  \begin{aligned}
  &\tr{U(\ket{0}\bra{0})U^{\dagger}X}^2 = (\cos\psi\sin\psi)^2(e^{i\phi} + e^{-i\phi})^2 ={\sin^22\psi} (real(e^{i\phi}))^2\\
  &\tr{U(\ket{0}\bra{0})U^{\dagger}Y}^2 = (\cos\psi\sin\psi)^2(ie^{i\phi} - ie^{-i\phi})^2 = {\sin^22\psi} (Im(e^{i\phi}))^2\\
  &\tr{U(\ket{0}\bra{0})U^{\dagger}Z}^2 = (\cos^2\psi - \sin^2\psi)^2 = \cos^22\psi\\
  & \tau = \min\{{\sin^22\psi} (real(e^{i\phi}))^2, {\sin^22\psi} (Im(e^{i\phi}))^2, \cos^22\psi\}.
  \end{aligned}
\end{equation}
It is straightforward to verify that when $2\psi = \arcsin{\sqrt{2/3}}$ and $\phi = \pi/4$, the value of $\tau$ attains its maximum of $1/3$.

\paragraph{Trainable construction} We can further allow $op$ to be trainable. Here, we provide a construction that employs two additional parameterized single-qubit rotation gates, in which $op$ is defined as follows:
\begin{figure}[H]
  \centering
  \begin{quantikz}
   & \gate[1]{R_{X}(\theta_1)} & \gate[1]{R_{Y}(\theta_2)}& \qw
  \end{quantikz}
  \caption{Trainable construction of $op$, in which we allow parameters $\theta_1$ and $\theta_2$ to be trainable.}\label{fig:op_trainable}
\end{figure}
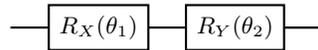

In the subsequent analysis, we demonstrate that MPQCs constructed using either method exhibit the same desirable properties. In the proofs of the main theorems, we assume that the condition in \cref{eq:op_condition} holds. In \cref{app:trainable_op}, we further show that the favorable properties of MPQCs still hold when the operator is trainable, as illustrated in the construction of \cref{fig:op_trainable}.

\section{Technical preliminaries}
In this section, we introduce the mathematical tools used to analyze the variance and gradient variance of parameterized quantum circuits.

\subsection{2-design of parameterized rotation gates}\label{app:2-design}
Let $R_P(\theta)=\exp{-i \frac{\theta}{2} P}=\cos(\frac{\theta}{2})\mathbb{I}-i\sin(\frac{\theta}{2})P$, it is not hard to see that the set
${\{ R_P(\theta) \}}_{\theta \in \left[0, 2\pi\right)}$
forms a group, which is a subgroup of the $n$-qubit unitary group
$\mathbb{U}(2^n)$.
Similar to unitary $t$-design, here we consider the $t$-design over the group $ {\{ R_P(\theta) \}}_{\theta \in \left[0, 2\pi\right)}$, which we call
the \emph{quantum rotation $t$-design}.

\begin{definition}
  A set of  unitary matrices ${\left\{A_i\right\}}_{i=1}^K$ is called a \emph{quantum rotation $t$-design} with respect to
  $R_P(\theta)$, if
  \begin{align}\label{def:two_design}
    \frac{1}{K} \sum_{i=1}^K {\left(A_i \otimes A_i^{\dagger} \right)}^{\otimes t} =
    \frac{1}{2\pi} \int_{0}^{2\pi} {\left(R_P(\theta)
    \otimes R_P(-\theta) \right)}^{\otimes t} d\theta.
  \end{align}
\end{definition}

We now prove that the following gate set forms a quantum rotation $2$-design.
\begin{theorem}\label{thm:two_design}
  ${\left\{R_P(\theta)\right\}}_{\theta=0, \pi/2, \pi, 3\pi/2}$ is a
  quantum rotation $2$-design with respect to
  ${\{ R_P(\theta) \}}_{\theta \in \left[0, 2\pi\right)}$.
\end{theorem}

\begin{proof}
  Utilizing the relations
  \begin{align*}
    &\frac{1}{2 \pi} \int_0^{2 \pi} \cos ^4 \frac{\theta}{2} {\rm d}\theta
      = \frac{1}{2 \pi} \int_0^{2 \pi} \sin ^4 \frac{\theta}{2} {\rm d}\theta
      = \frac{3}{8},\\
    & \frac{1}{2 \pi} \int_0^{2 \pi} \cos \frac{\theta}{2}
      \sin ^3 \frac{\theta}{2} {\rm d}\theta
      = \frac{1}{2 \pi} \int_0^{2 \pi} \cos ^3 \frac{\theta}{2}
      \sin \frac{\theta}{2} {\rm d}\theta=0, \\
    & \frac{1}{2 \pi} \int_0^{2 \pi} \cos ^2 \frac{\theta}{2}
      \sin ^2 \frac{\theta}{2} {\rm d}\theta=\frac{1}{8},
  \end{align*}
  we have
  \begin{align*}
    & \frac{1}{2\pi}\int_{0}^{2\pi} {R_P(\theta)}^{\otimes 2} \otimes
      {R_P(-\theta)}^{\otimes 2} {\rm d}\theta \\
    = & \sum_{i_1,\ldots, i_4=0}^1 \frac{1}{2 \pi} \int_0^{2 \pi} i^{-i_1-i_2+i_3+i_4}
        {\Bigl(\cos \frac{\theta}{2}\Bigr)}^{4-\sum_j i_j} {\Bigl(\sin \frac{\theta}{2}\Bigr)}^{\sum_j i_j}
      \times \biggl(\bigotimes_{j=1}^{4} P^{i_{j}} \biggr)\,
      {\rm d}\theta\\
    = & \frac{3}{8} I^{\otimes 4} + \frac{3}{8} P^{\otimes 4} +
        \frac{1}{8} \sum_{\substack{i_1, \ldots, i_4=0 \\
      i_1+\cdots+i_4=2}}^1 i^{-i_1-i_2+i_3+i_4}
      \bigotimes_{j=1}^{4} P^{i_{j}}.
  \end{align*}

  Meanwhile, it can be verified that
  \begin{align*}
    & \frac{1}{4} \sum_{k=0}^3 {R_P \Bigl( \frac{k\pi}{2}
      \Bigr)}^{\otimes 2} \otimes {R_P \Bigl(\frac{-k\pi}{2}
      \Bigr)}^{\otimes 2} \\
    = & \frac{1}{4} \left(1+\frac{1}{4}+\frac{1}{4}\right)I^{\otimes 4}+
        \frac{1}{4} \left(1+\frac{1}{4}+\frac{1}{4}\right)P^{\otimes 4}+ \frac{1}{4} \left(\frac{1}{4}+\frac{1}{4}\right)
     \sum_{\substack{i_1, \ldots, i_4=0 \\
      i_1+\cdots+i_4=2}}^1 i^{-i_1-i_2+i_3+i_4}
      \bigotimes_{j=1}^{4} P^{i_{j}}\\
    = & \frac{3}{8} I^{\otimes 4} + \frac{3}{8} P^{\otimes 4} +
        \frac{1}{8} \sum_{\substack{i_1, \ldots, i_4=0 \\
      i_1+\cdots+i_4=2}}^1 i^{-i_1-i_2+i_3+i_4}
      \bigotimes_{j=1}^{4} P^{i_{j}},
  \end{align*}
  which concludes the proof.
\end{proof}

Thus for arbitrary operators $A, B, C, D$, we have the following corollary:
\begin{corollary}\label{cor:two_design}
  For any $n$-qubit operators $A, B, C, D$, the following equation holds:
  \begin{equation}\label{eq:two_design_cor}
    \mathbb{E}_{\theta} \tr{A R_P(\theta) B R_P(-\theta)}\tr{C R_P(\theta) D R_P(-\theta)}=
    \frac{1}{4}\sum_{\theta \in \AngleSet}\tr{A R_P(\theta) B R_P(-\theta)}\tr{C R_P(\theta) D R_P(-\theta)}.
  \end{equation}
\end{corollary}
\begin{proof}
  The proof is straightforward by using the definition of the quantum rotation $2$-design in Eq.~\eqref{def:two_design} and the fact of Thm.~\ref{thm:two_design}, we have:
  \begin{equation}
    \mathbb{E}_{\theta} R_P(\theta) \otimes R_P(-\theta) \otimes R_P(\theta) \otimes R_P(-\theta)=\frac{1}{4}\sum_{\theta \in \AngleSet}R_P(\theta) \otimes R_P(-\theta) \otimes R_P(\theta) \otimes R_P(-\theta).
  \end{equation}
  The left-hand side of Eq.~\eqref{eq:two_design_cor} can be expressed as:
  \begin{equation}
    \begin{aligned}
      &\mathbb{E}_{\theta} \tr{A R_P(\theta) B R_P(-\theta)}\tr{C R_P(\theta) D R_P(-\theta)}\\
      =& \mathbb{E}_{\theta} \left(\sum_{i,j} \bra{i}A R_P(\theta) B \ket{j} \bra{j} R_P(-\theta)\ket{i} \right)
      \left(\sum_{k,l} \bra{k}C R_P(\theta) D\ket{l}\bra{l} R_P(-\theta)\ket{k}\right)\\
      =& \mathbb{E}_{\theta} \left(\sum_{i,j} \bra{i}\otimes\bra{j} \cdot (A R_P(\theta) B) \otimes R_P(-\theta)\cdot \ket{j}\otimes\ket{i} \right)
      \left(\sum_{k,l} \bra{k}\otimes\bra{l} \cdot(C R_P(\theta) D) \otimes R_P(-\theta) \cdot\ket{l}\otimes\ket{k}\right)\\
      =& \mathbb{E}_{\theta} \left(\sum_{i,j,k,l} \bra{i}\otimes\bra{j}\otimes\bra{k}\otimes\bra{l} \cdot (A R_P(\theta) B) \otimes R_P(-\theta) \otimes (C R_P(\theta) D) \otimes R_P(-\theta) \cdot\ket{j}\otimes\ket{i}\otimes\ket{l}\otimes\ket{k}\right)\\
      =& \mathbb{E}_{\theta} \left(\sum_{i,j,k,l} \bra{i}A\otimes \bra{j}\otimes\bra{k} C\otimes\bra{l} \cdot R_P(\theta) \otimes R_P(-\theta) \otimes R_P(\theta)  \otimes R_P(-\theta) \cdot B\ket{j}\otimes\ket{i}\otimes D\ket{l}\otimes\ket{k}\right)\\
      =& \sum_{i,j,k,l} \bra{i}A\otimes \bra{j}\otimes\bra{k}C\otimes\bra{l} \cdot \mathbb{E}_{\theta} \left(   R_P(\theta) \otimes R_P(-\theta) \otimes R_P(\theta)  \otimes R_P(-\theta)\right) \cdot B\ket{j}\otimes\ket{i}\otimes D\ket{l}\otimes\ket{k}\\
      =& \sum_{i,j,k,l} \bra{i}A\otimes \bra{i}\otimes\bra{k} \cdot \frac{1}{4}\sum_{\theta \in \AngleSet} \left(  C\otimes\bra{k} R_P(\theta) \otimes R_P(-\theta) \otimes R_P(\theta)  \otimes R_P(-\theta)\right) \cdot B\ket{j}\otimes\ket{i}\otimes D\ket{l}\otimes\ket{k}\\
      =&\frac{1}{4}\sum_{\theta \in \AngleSet}\left(\sum_{i,j} \bra{i}A R_P(\theta) B \ket{j} \bra{j} R_P(-\theta)\ket{i} \right)
      \left(\sum_{k,l} \bra{k}C R_P(\theta) D\ket{l}\bra{l} R_P(-\theta)\ket{k}\right)\\
      =&\frac{1}{4}\sum_{\theta \in \AngleSet}\tr{A R_P(\theta) B R_P(-\theta)}\tr{C R_P(\theta) D R_P(-\theta)}.
    \end{aligned}
  \end{equation}
\end{proof}

\subsection{Pauli path integral}
A Pauli path is a sequence $\vec{s}=(s_0,\cdots,s_m)\in \bm{P}^{m+1}_n$, where $\bm{P}_n=\{\sfrac{\mathbb{I}}{\sqrt{2}},\sfrac{X}{\sqrt{2}},\sfrac{Y}{\sqrt{2}},\sfrac{Z}{\sqrt{2}}\}^{\otimes n} $ represents the set of all normalized $n$-qubit Pauli words. 
Using the fact that the normalized $n$-qubit Pauli group $\bm{P}_n$ forms a basis of the $2^n$-dimensional Hilbert space, we can express any operator $A$ as a linear combination of elements in $\bm{P}_n$:
\begin{equation}
  A=\sum_{s\in \bm{P}_n} \tr{A s} s,
\end{equation}
Iteratively applying the Pauli operator decomposition, we can express the expectation value of $O$ as the sum of contributions from all Pauli paths:
\begin{equation}\label{eq:Pauli_path_integral_noiseless}
  \begin{aligned}
  \Vexp{O} &= \sum_{s_m} \tr{O s_m} \tr{s_m \mathcal{C}(\bm{\theta})\rho \mathcal{C}(\bm{\theta})^\dagger}\\
  &= \sum_{s_m} \tr{O s_m} \tr{s_m U_m(\theta_m) \cdots U_1(\theta_1) \rho U_1^\dagger(\theta_1) \cdots U_m^\dagger(\theta_m)}\\
  &= \sum_{s_m,s_{m-1}} \tr{O s_m} \tr{s_m U_m(\theta_m) s_{m-1} U_m^\dagger(\theta_m)} \tr{s_{m-1} U_{m-1}(\theta_{m-1}) \cdots U_1(\theta_1) \rho U_1^\dagger(\theta_1) \cdots U_{m-1}^\dagger(\theta_{m-1})}\\
  &\vdots\\
  &= \sum_{s_m,s_{m-1},\cdots,s_0} \tr{O s_m} \tr{s_m U_m(\theta_m) s_{m-1} U_m^\dagger(\theta_m)} \cdots \tr{s_1 U_1(\theta_1) s_0 U_1^\dagger(\theta_1)} \tr{s_0 \rho}\\
  &= \sum_{s_m,s_{m-1},\cdots,s_0} \tr{O s_m} \tr{s_0 \rho} \prod_{i=1}^{m} \tr{s_i U_i(\theta_i) s_{i-1} U_i^\dagger(\theta_i)} \\
  &=\sum_{\vec{s}} f(\vec{s},\bm{\theta},O,\rho),
  \end{aligned}
\end{equation}
where 
\begin{equation}\label{eq:contribution_function_noiseless}
  f(\vec{s},\bm{\theta},O,\rho)\coloneq\tr{O s_m} \tr{s_0 \rho} \prod_{i=1}^{m} \tr{s_i U_i(\theta_i) s_{i-1} U_i^\dagger(\theta_i)}
\end{equation}
is the contribution of a specfic Pauli path $\vec{s}=(s_0,\cdots,s_m)$ to the expectation value $\Vexp{O}$.

For the contribution of Pauli path $f(\vec{s},\bm{\theta},O,\rho)$, we have the following lemma:
\begin{lemma}\label{lem:cross_term}
  For the Pauli path $\vec{s}$ and $\vec{s}\hspace{0.1em}'$, and for arbitrary observable $O_1$ and $O_2$, the contribution $f(\vec{s},\bm{\theta},O_1,\rho)$ and $f(\vec{s}\hspace{0.1em}',\bm{\theta},O_2,\rho)$ satisfy the following equation:
  \begin{equation}
      \mathbb{E}_{\theta} f(\vec{s},\bm{\theta},O_1,\rho)f(\vec{s}\hspace{0.1em}',\bm{\theta},O_2,\rho)= \frac{1}{4^m}\sum_{\bm{\theta} \in \AngleSet^m}f(\vec{s},\bm{\theta},O_1,\rho)f(\vec{s}\hspace{0.1em}',\bm{\theta},O_2,\rho),
  \end{equation}
  where $\bm{\theta}=\{\theta_1,\ldots,\theta_m\}$ is the set of rotation angles and $m$ is number of rotation gates.
\end{lemma}
\begin{proof}
  The proof is straightforward by using Corollary~\ref{cor:two_design}, we have:
  \begin{equation}
    \begin{aligned}
      & \ExpC f(\vec{s},\bm{\theta},O_1,\rho)f(\vec{s}\hspace{0.1em}',\bm{\theta},O_2,\rho)\\
      =& \tr{O_1 s_m}\tr{O_2 s'_m} \tr{s_0 \rho}\tr{s'_0 \rho} \prod_{i=1}^{m}\mathbb{E}_{\theta_i} \tr{s_i U_i(\theta_i) s_{i-1} U_i^\dagger(\theta_i)}\tr{s'_i U_i(\theta_i) s'_{i-1} U_i^\dagger(\theta_i)}
    \end{aligned}
  \end{equation}
  
  For terms $\mathbb{E}_{\theta_i} \tr{s_i U_i(\theta_i) s_{i-1} U_i^\dagger(\theta_i)}\tr{s'_i U_i(\theta_i) s'_{i-1} U_i^\dagger(\theta_i)}$, using Eq.~\eqref{eq:two_design_cor}, we have:
  \begin{equation}
    \mathbb{E}_{\theta_i} \tr{s_i U_i(\theta_i) s_{i-1} U_i^\dagger(\theta_i)}\tr{s'_i U_i(\theta_i) s'_{i-1} U_i^\dagger(\theta_i)}= \frac{1}{4}\sum_{\theta_i \in \AngleSet} \tr{s_i U_i(\theta_i) s_{i-1} U_i^\dagger(\theta_i)}\tr{s'_i U_i(\theta_i) s'_{i-1} U_i^\dagger(\theta_i)}.
  \end{equation}
  Therefore, we have:
  \begin{equation}
    \begin{aligned}
      &\ExpC f(\vec{s},\bm{\theta},O_1,\rho)f(\vec{s}\hspace{0.1em}',\bm{\theta},O_2,\rho)\\
      =& \tr{O_1 s_m}\tr{O_2 s'_m} \tr{s_0 \rho}\tr{s'_0 \rho} \prod_{i=1}^{m}\frac{1}{4}\sum_{\theta_i \in \AngleSet} \tr{s_i U_i(\theta_i) s_{i-1} U_i^\dagger(\theta_i)}\tr{s'_i U_i(\theta_i) s'_{i-1} U_i^\dagger(\theta_i)}\\
      =& \frac{1}{4^m}\sum_{\bm{\theta} \in \AngleSet^m}f(\vec{s},\bm{\theta},O_1,\rho)f(\vec{s}\hspace{0.1em}',\bm{\theta},O_2,\rho).
    \end{aligned}
  \end{equation}
\end{proof}

Next we study the evolution of the Pauli operator $s$ under the operator $U_i(\theta_i)=\exp{-i \frac{\theta_i}{2} P_i}C_i$, which is given by
\begin{equation}
    U_i(\theta_i) s_{i-1} U_i^\dagger(\theta_i) = \exp{-i \frac{\theta_i}{2} P_i} \underbrace{C_i s_{i-1} C_i^\dagger}_{Q_i} \exp{i \frac{\theta_i}{2} P_i},
\end{equation}
where $Q_i=C_i s_{i-1} C_i^\dagger$ is the transformed Pauli operator after applying the Clifford gate $C_i$ to $s_{i-1}$.
The above equation shows that the factor $\tr{s_i U_i(\theta_i) s_{i-1} U_i^\dagger(\theta_i)}$ in $f(\vec{s},\bm{\theta},O,\rho)$ can be expressed as:
\begin{equation}\label{eq:gate_term_in_f}
  \begin{aligned}
    \tr{s_i U_i(\theta_i) s_{i-1} U_i^\dagger(\theta_i)} &= \tr{s_i \exp{-i \frac{\theta_i}{2} P_i} Q_i \exp{i \frac{\theta_i}{2} P_i}}\\
    &= \tr{\exp{i \frac{\theta_i}{2} P_i} s_i \exp{-i \frac{\theta_i}{2} P_i} Q_i }\\
&= \begin{cases}
      \tr{s_i Q}, & [P_i, s_i] = 0, \\
      \cos(\theta_i) \tr{s_i Q_i} - i \sin(\theta_i) \tr{s_i P_i Q_i}, & \{P_i, s_i\} = 0.
      \end{cases}\\
  \end{aligned}
\end{equation}
Because of $Q_i=C_i s_{i-1} C_i^\dagger$, and $C_i$ is Clifford operator, the operator $Q_i$ is also a Pauli operator in $\bm{P}_n$. Then, if $[P_i, s_i] = 0$, we have $Q_i = s_i$, which contributes a term $\tr{s_i Q}$ to the corresponding $f(\vec{s}, \bm{\theta}, O, \rho)$. On the other hand, if $\{P_i, s_i\} = 0$, then $Q_i$ may be either $s_i$ or $s_i P_i$, leading to terms of the form $\cos(\theta_i) \tr{s_i Q_i}$ or $-i \sin(\theta_i) \tr{s_i P_i Q_i}$ in $f(\vec{s}, \bm{\theta}, O, \rho)$, respectively.

Specifically, if the rotation angle $\theta_i$ takes the value in $\AngleSet$, the Pauli rotation $\exp{-i \frac{\theta_i}{2} P_i}$ falls into the set of Clifford gates, and the factor $\tr{s_i U_i(\theta_i) s_{i-1} U_i^\dagger(\theta_i)}$ in Eq.~\eqref{eq:gate_term_in_f} can be expressed as:
\begin{equation}\label{eq:gate_term_in_f_discrete}
  \begin{aligned}
    &\tr{s_i U_i(\theta_i) s_{i-1} U_i^\dagger(\theta_i)} = \begin{cases}
      \tr{s_i Q_i}, & [P_i, s_i] = 0, \\
      \cos(\theta_i) \tr{s_i Q_i} - i \sin(\theta_i) \tr{s_i P_i Q_i}, & \{P_i, s_i\} = 0.
      \end{cases}\\
      &= \begin{cases}
        0, & [P_i, s_i] = 0, Q_i\neq s_i  \\
        1, & [P_i, s_i] = 0, Q_i= s_i\\
        \pm 1 , & \{P_i, s_i\} = 0, Q_i= s_i\\
      \end{cases}
      \mathrm{when} \quad \theta_i\in\{0,\pi\} \quad \mathrm{or} 
      = \begin{cases}
        0, & [P_i, s_i] = 0, s_i\neq Q_i  \\
        1, & [P_i, s_i] = 0, Q_i= s_i\\
        \pm 1 , & \{P_i, s_i\} = 0, Q_i=is_iP_i. \\
      \end{cases}
      \mathrm{when} \quad \theta_i\in\{\frac{\pi}{2},\frac{3\pi}{2}\},
  \end{aligned}
\end{equation}
Here, we ignore the sign $\pm$ in front of the Pauli operator $Q_i$ in the above equation. As shown in~\eqref{eq:gate_term_in_f_discrete}, if $[P_i, s_i] = 0$, then $Q_i$ must be equal to $s_i$. If instead $\{P_i, s_i\} = 0$, then $Q_i = s_i$ when $\theta_i \in \{0, \pi\}$, and $Q_i = i s_i P_i$ when $\theta_i \in \{\frac{\pi}{2}, \frac{3\pi}{2}\}$. This observation will play an important role in the subsequent analysis.

\section{Variance and gradient variance of the loss function of PQCs}\label{app:proof_two_variances_MPQC}
In this section, we express and simpliy the variance of the loss function and the gradient variance of PQCs using the formalisms of the Pauli path integral and quantum rotation 2-design, which form the foundation of our theoretical analysis.

\subsection{Simplified expression via the orthogonality condition of Pauli paths}
For an arbitrary PQC $\mathcal{C}(\bm{\theta})$ and observable $O$, let its loss function be defined as $\loss = \tr{O \mathcal{C}(\bm{\theta})\rho \mathcal{C}(\bm{\theta})^\dagger}$.
According to this definition, the variance of the loss function and that of its gradient can be expressed as follows:
\begin{small}
\begin{equation}\label{eq:twoquantities}
  \begin{aligned}
    \operatorname{Var}_{\bm{\theta}} \left[ \loss \right] &= \mathbb{E}_{\bm{\theta}}\left[\loss^2\right] - \left(\mathbb{E}_{\bm{\theta}}\left[\loss\right]\right)^2 \\
    \operatorname{Var}_{\bm{\theta}} \left[ \frac{\partial \loss}{\partial \theta_j} \right] &= \mathbb{E}_{\bm{\theta}}\left[\left(\frac{\partial \loss}{\partial \theta_j}\right)^2\right] - \left(\mathbb{E}_{\bm{\theta}}\left[\frac{\partial \loss}{\partial \theta_j}\right]\right)^2,
  \end{aligned}
\end{equation}
\end{small}

\noindent where each $\theta_i$ is sampled {uniformly from} $\left[0, 2\pi\right)$.
Writing $P_\alpha$ as the Pauli expansion of the observable $O=\sum_{\alpha} c_\alpha P_\alpha$, the loss function can be expressed in the Pauli path integral formalism according to \cref{eq:Pauli_path_integral_noiseless}:
\begin{equation}
  \begin{aligned}
    &\loss = \langle O \rangle\\
    &=\sum_{\alpha,\vec{s}} c_\alpha  \tr{P_\alpha s_m} \tr{s_0 \rho} \prod_{i=1}^{m} \tr{s_i U_i(\theta_i) s_{i-1} U_i(\theta_i)^\dagger}\\
    &= \sum_{\alpha,\vec{s}} c_\alpha f(\vec{s},\bm{\theta},P_\alpha,\rho),
  \end{aligned}
\end{equation}

\noindent where $\vec{s} = (s_0,s_1,\cdots,s_m)$ is a Pauli path, which is a sequence of normalized Pauli operators $s_i\in \{\frac{\mathbb{I}}{\sqrt{2}},\frac{X}{\sqrt{2}},\frac{Y}{\sqrt{2}},\frac{Z}{\sqrt{2}}\}^{\otimes n}$, and $f(\vec{s},\bm{\theta},P_\alpha,\rho) \coloneqq \tr{P_\alpha s_m} \tr{s_0 \rho} \prod_{i=1}^{m} \tr{s_i U_i(\theta_i) s_{i-1} U_i(\theta_i)^\dagger}$ denotes the contribution of a specific Pauli path $\vec{s}$ to the expectation value $\langle O \rangle$.

In particular, when the rotation angles satisfy $\bm{\theta} \in \AngleSet^m$, each $U_i(\theta_i)$ belongs to the Clifford group. Consequently, for any fixed $s_i$, there exists a unique $s_{i-1}$ such that $\tr{s_i U_i(\theta_i) s_{i-1} U_i(\theta_i)^\dagger} \neq 0$.
Therefore, starting from $s_m \propto P_\alpha$, there exists a unique Pauli path $\vec{s}^{\hspace{0.1em}(\bm{\theta},\alpha)}$ satisfying
$\tr{P_\alpha s_m} \prod_{i=1}^{m} \tr{s_i U_i(\theta_i) s_{i-1} U_i(\theta_i)^\dagger} \neq 0$.

Using the above expression, and assuming the PQC architecture satisfies a mild structural condition (shown in Ref.~\cite{shao2024simulating} to be met by most PQCs and also holding for arbitrary MPQCs which will be proved in \cref{app:variance_MPQC}), we can express the variance of the loss function and that of its gradient in a simplified form.

\begin{lemma}\label{lem:two_variances_MPQC}
Let $O = \sum_{\alpha} c_\alpha P_\alpha$ be an observable, and $\mathcal{C}(\bm{\theta})$ be a PQC with parameters $\bm{\theta} \in [0, 2\pi)^m$. 
Suppose the following orthogonality condition holds:
\begin{equation}\label{eq:orthogonality_condition}
\mathbb{E}_{\bm{\theta}}\left[f(\vec{s},\bm{\theta}, P_\alpha,  \rho) f(\vec{s}\hspace{0.1em}',\bm{\theta}, P_\beta,  \rho)\right] = 0, \quad \forall \alpha \neq \beta,\vec{s}, \vec{s}\hspace{0.1em}', 
\end{equation}
    and each $\langle P_\alpha \rangle$ is not a non-zero constant function of $\bm{\theta}$.
    Then the variance of the loss function and the variance of its gradient can be expressed as:
  \begin{equation}\label{eq:expression_variance}
      \operatorname{Var}_{\bm{\theta}} \left[ \loss \right] = \frac{1}{4^{m}} 
      \sum_{\bm{\theta} \in \AngleSet^m} 
      \sum_\alpha c_\alpha^2 \, f(\vec{s}^{\hspace{0.1em}(\bm{\theta},\alpha)},\bm{\theta},P_\alpha,\rho)^2
  \end{equation}
  \begin{equation}\label{eq:expression_variance_gradient}
    \operatorname{Var}_{\bm{\theta}} \left[ \frac{\partial \loss}{\partial \theta_j} \right] = \frac{1}{4^{m}} 
    \sum_{\substack{
        \bm{\theta} \in \AngleSet^m \\
        \{P_j, s_j^{(\bm{\theta},\alpha)}\} = 0
      }} \sum_\alpha c_\alpha^2 \, f(\vec{s}^{\hspace{0.1em}(\bm{\theta},\alpha)},\bm{\theta},P_\alpha,\rho)^2,
  \end{equation}

\noindent where $P_j$ denotes the Pauli operator in the elementary rotation $e^{-i\frac{\theta_j}{2} P_j}$ of the circuit, and $\vec{s}^{\hspace{0.1em}(\bm{\theta},\alpha)}$ is the unique normalized Pauli operator sequence such that $\tr{P_\alpha s_m} \prod_{i=1}^{m} \tr{s_i U_i(\theta_i) s_{i-1} U_i(\theta_i)^\dagger} \neq 0$.
\end{lemma}

Notably, it can be observed that $\operatorname{Var}_{\bm{\theta}}\left[ \frac{\partial \loss}{\partial \theta_j} \right]$ corresponds to a subset of the terms in $\operatorname{Var}_{\bm{\theta}}\left[ \loss \right]$, which allows us to analyze their scaling using the same techniques. In the following two subsections, we prove \cref{eq:expression_variance} and \cref{eq:expression_variance_gradient}, respectively.

\subsection{Proof of \cref{eq:expression_variance}}\label{sec:proof_eq4}

We begin by expanding the variance of $\loss$ in the language of Pauli path integral:
\begin{equation}\label{eq:variance_PQC}
  \begin{aligned}
    \operatorname{Var}_{\bm{\theta}}[\loss] &= \ExpC [\Vexp{O}^2] - \ExpC [\Vexp{O}]^2 \\
    &= \ExpC \left[\sum_{\alpha,\beta} c_\alpha c_\beta \Vexp{P_\alpha} \Vexp{P_\beta} \right] - \ExpC \left[\sum_{\alpha} c_\alpha\Vexp{P_\alpha}\right]^2 \\
    &= \ExpC \left[\sum_{\alpha,\vec{s}} \sum_{\beta,\vec{s}\hspace{0.1em}'} c_\alpha c_\beta f(\vec{s},\bm{\theta},P_\alpha,\rho) f(\bm{\theta}, P_\beta, \vec{s}\hspace{0.1em}',\rho)\right] - \ExpC \left[\sum_{\alpha,\vec{s}} c_\alpha f(\vec{s},\bm{\theta},P_\alpha,\rho)\right]^2.
  \end{aligned}
\end{equation}
Next, we show that for any $P_\alpha$, the following holds:
\begin{equation}\label{eq:vanish_expectation_term}
 \ExpC \left[\Vexp{P_\alpha}\right] =  \ExpC \left[\sum_{\vec{s}} f(\vec{s},\bm{\theta},P_\alpha,\rho)\right] = 0.
\end{equation}

In the conditions of \cref{lem:two_variances_MPQC}, we require that $\Vexp{P_\alpha}$ is not a non-zero constant, which means that $\Vexp{P_\alpha}$ can either be zero or a non-trivial function of $\bm{\theta}$.
If $\Vexp{P_\alpha} = 0$, then \cref{eq:vanish_expectation_term} holds trivially. Now we suppose that $\Vexp{P_\alpha}$ is not a constant. 
We consider the evolution of the Pauli path in the Heisenberg picture, as described in~\cref{eq:gate_term_in_f}. Initially, starting from the observable, we have $s_m = P_\alpha/\sqrt{2^n}$. If $[P_m, s_m] = 0$, then $Q_m = C_m s_{m-1} C_m^\dagger = s_m$, which implies that the parameter $\theta_m$ has no effect on $\langle P_\alpha \rangle$. 
If this commutation relation persists throughout the circuit, i.e., $[P_i, s_i] = 0$ for all $i$, then each $Q_i$ is uniquely determined, and none of the parameters affects $\langle P_\alpha \rangle$. This contradicts our assumption that $P_\alpha$ is a nontrivial observable with respect to $\mathcal{C}(\bm{\theta})$. 

Therefore, for each non-vanishing term $f(\vec{s},\bm{\theta},P_\alpha,\rho) \neq 0$, there must exist at least one index $i \in [m]$ such that the corresponding contribution contains a term of the form
\[
\cos(\theta_i) \tr{s_i Q_i} \quad \text{or} \quad i \sin(\theta_i) \tr{s_i P_i Q_i}.
\]
Since $\mathbb{E}_{\theta_i}[\cos(\theta_i)] = \mathbb{E}_{\theta_i}[\sin(\theta_i)] = 0$, we obtain
\[
\mathbb{E}_{\theta_i}\left[\cos(\theta_i) \tr{s_i Q}\right] =  \mathbb{E}_{\theta_i}\left[- i \sin(\theta_i) \tr{s_i P_i Q}\right] = 0.
\]
This completes the proof of \cref{eq:vanish_expectation_term}.

Next, we compute $\ExpC [\Vexp{O}^2]$. We first prove that for any fixed $\alpha$, the following orthogonality condition holds:
\begin{equation}\label{eq:orthogonality_same_Pauli_O}
\ExpC \left[f(\vec{s},\bm{\theta}, P_\alpha,  \rho) f(\vec{s}\hspace{0.1em}',\bm{\theta}, P_\alpha,  \rho)\right] = 0, \quad \forall \vec{s} \neq \vec{s}\hspace{0.1em}'.
\end{equation}

Since the observable is the Pauli operator $P_\alpha$, the final Pauli path elements $s_m$ and $s_m'$ must both equal $P_\alpha/\sqrt{2^n}$; otherwise, both $f(\vec{s},\bm{\theta},P_\alpha,\rho)$ and $f(\vec{s}\hspace{0.1em}',\bm{\theta}, P_\alpha,  \rho)$ vanish.

Let $i$ be the largest index such that $s_i \neq s'_i$. According to the analysis following \cref{eq:gate_term_in_f}, we must have $\{P_{i+1}, s_{i+1}(=s'_{i+1})\} = 0$; otherwise, we would have $Q_{i+1} = Q'_{i+1} = s_{i+1}$. Since $Q_{i+1} = C_{i+1} s_i C_{i+1}^\dagger$ and $Q'_{i+1} = C_{i+1} s'_i C_{i+1}^\dagger$, this implies $s_i = s'_i$, contradicting our assumption.

Therefore, $\{P_{i+1}, s_{i+1}\} = 0$, and without loss of generality, we assume that $Q_{i+1} = s_{i+1}$ and $Q'_{i+1} = i s_{i+1} P_{i+1}$. This results in a product of terms in $f(\vec{s},\bm{\theta},P_\alpha,\rho) f(\bm{\theta}, P_\alpha, \vec{s}\hspace{0.1em}', \rho)$ that includes $\cos{\theta_{i+1}} \sin\theta_{i+1}$. However, since $\mathbb{E}_{\theta_{i+1}}[\cos{\theta_{i+1}} \sin\theta_{i+1}] = 0$, the cross term $\ExpC \left[f(\vec{s},\bm{\theta},P_\alpha,\rho) f(\bm{\theta}, P_\alpha, \vec{s}\hspace{0.1em}', \rho)\right]$ vanishes. Hence, we conclude the proof for \cref{eq:orthogonality_same_Pauli_O}.

Combining \cref{eq:orthogonality_same_Pauli_O} with the orthogonality condition:
\begin{equation}\label{eq:orthogonality_condition_app}
\ExpC \left[f(\vec{s},\bm{\theta},P_\alpha,\rho) f(\vec{s}\hspace{0.1em}',\bm{\theta}, P_\beta,  \rho)\right] = 0, \quad \forall \alpha \neq \beta,\vec{s}, \vec{s}\hspace{0.1em}', 
\end{equation}
we obtain
\begin{equation}
  \begin{aligned}
    \operatorname{Var}_{\bm{\theta}}[\loss] 
    &= \ExpC \left[\sum_{\alpha,\vec{s}} \sum_{\beta,\vec{s}\hspace{0.1em}'}c_\alpha c_\beta f(\vec{s},\bm{\theta},P_\alpha,\rho) f(\vec{s}\hspace{0.1em}',\bm{\theta}, P_\beta,\rho)\right] \\
    &= \ExpC \left[\sum_{\alpha, \vec{s}} c_\alpha^2 f(\vec{s},\bm{\theta},P_\alpha,\rho)^2\right] \\
    &= \frac{1}{4^m} \sum_{\bm{\theta} \in \AngleSet^m} \sum_\alpha c_\alpha^2 f(\vec{s}^{\hspace{0.1em}(\bm{\theta},\alpha)},\bm{\theta},P_\alpha,\rho)^2,
  \end{aligned}
\end{equation}
where the last equality uses the property of quantum rotation 2-design, as proven in \cref{lem:cross_term}. \qed

\subsection{Proof of \cref{eq:expression_variance_gradient}}

\cref{eq:expression_variance_gradient} expresses the variance of the gradient with respect to each parameter in terms of the Pauli path integral and quantum rotation 2-design. Similarly, we first express the gradient with respect to a given parameter $\theta_j$ in the form of a Pauli path integral:

\begin{equation}\
\begin{aligned}
  \frac{\partial \Vexp{O}}{\partial{\theta_j}} &= \sum_{s_m} \tr{O s_m} \frac{\partial \tr{s_m \mathcal{C}(\bm{\theta})\rho \mathcal{C}(\bm{\theta})^\dagger}}{\partial{\theta_j}}\\
  &= \sum_{\vec{s}} \frac{\partial f(\vec{s},\bm{\theta},O,\rho)}{\partial{\theta_j}}\\
  &= \sum_{s_m,s_{m-1},\cdots,s_0}  \tr{s_0 \rho} \tr{O s_m}\prod_{i\neq j}^{L} \tr{s_i U_i(\theta_i) s_{i-1} U_i^\dagger(\theta_i)} \frac{\partial }{\partial{\theta_j}}\left(\tr{s_j U_j(\theta_j) s_{j-1} U_j^\dagger(\theta_j)}\right).
\end{aligned}  
\end{equation}

According to the parameter-shift rule~\cite{schuld2019evaluating}, there is $\frac{\partial \Vexp{O}}{\partial{\theta_j}} = \frac{1}{2} (\Vexp{O}_{\theta_j + \frac{\pi}{2}} - \Vexp{O}_{\theta_j - \frac{\pi}{2}})$, where $\Vexp{O}_{\theta_j + \frac{\pi}{2}}$ and $\Vexp{O}_{\theta_j - \frac{\pi}{2}}$ are the expectation values of the observable $O$ when the parameter $\theta_j$ is shifted by $\frac{\pi}{2}$ and $-\frac{\pi}{2}$, respectively.
Therefore, we have $\ExpC \left(\frac{\partial \Vexp{O}}{\partial{\theta_j}}\right) = 0$, and apply the property of quantum rotation 2-design (as in \cref{lem:cross_term}) to $\frac{\partial f(\vec{s},\bm{\theta},O,\rho)}{\partial{\theta_j}}$, we have
\begin{equation}\label{eq:var_general}
\begin{aligned}
  \operatorname{Var}_{\bm{\theta}} \left[ \frac{\partial \loss}{\partial \theta_j} \right] = \ExpC \left[\left(\frac{\partial \Vexp{O}}{\partial{\theta_j}}\right)^2 \right]& = \ExpC  \left[\underset{\vec{s},
    \vec{s}\hspace{0.1em}'}{\sum} \frac{\partial f(\vec{s},\bm{\theta},O,\rho)}{\partial{\theta_j}}\frac{\partial f(\vec{s}\hspace{0.1em}',\bm{\theta},O,\rho)}{\partial{\theta_j}}\right]\\
  &=\frac{1}{4^m}\sum_{\bm{\theta} \in \AngleSet^m} \sum_{\vec{s},\vec{s}\hspace{0.1em}'} \frac{\partial f(\vec{s},\bm{\theta},O,\rho)}{\partial{\theta_j}}\frac{\partial f(\vec{s}\hspace{0.1em}',\bm{\theta},O,\rho)}{\partial{\theta_j}}.
\end{aligned}  
\end{equation}
A detailed proof of \cref{eq:var_general} is also provided in Appendix~G of Ref.~\cite{shao2025diagnosing}. We now evaluate $\left[\frac{\partial }{\partial{\theta}}\left(\tr{s_j U_j(\theta) s_{j-1} U_j^\dagger(\theta_j)}\right)\right]$ when $\theta_j \in \AngleSet$:

\begin{equation}\label{eq:gradient_term_in_f_discrete}
  \begin{aligned}
    &\left[\frac{\partial }{\partial{\theta}}\left(\tr{s_j U_j(\theta) s_{j-1} U_j^\dagger(\theta_j)}\right)\right] = \begin{cases}
      0, & [P_j, s_j] = 0, \\
      -\sin(\theta_j) \tr{s_j Q_j} - i \cos(\theta_j) \tr{s_j P_j Q_j}, & \{P_j, s_j\} = 0.
      \end{cases}\\
      &= \begin{cases}
        \pm1 , & \{P_j, s_j\} = 0, Q_j=is_jP_j, \\
        0, & others.  \\
      \end{cases}  
      \mathrm{when} \quad \theta_j\in\{0,\pi\} \quad \mathrm{or} 
      = \begin{cases}
        \pm1 , & \{P_j, s_j\} = 0, Q_j=s_j, \\
        0, & others.  \\
      \end{cases}    
      \mathrm{when} \quad \theta_j\in\{\frac{\pi}{2},\frac{3\pi}{2}\}.
  \end{aligned}
\end{equation}

It turns out that this term is closely related to the undifferentiated term $\tr{s_j U_j(\theta_j) s_{j-1} U_j^\dagger(\theta_j)}$ when $\theta_j \in \AngleSet$. To formalize this connection, we recall the value of such term
\begin{equation}\label{eq:gate_term_in_f_discrete_repeat}
  \begin{aligned}
    &\tr{s_j U_j(\theta_j) s_{i-1} U_j^\dagger(\theta_j)} = \begin{cases}
      \tr{s_j Q_j}, & [P_j, s_j] = 0, \\
      \cos(\theta_j) \tr{s_j Q_j} - i \sin(\theta_j) \tr{s_j P_j Q_j}, & \{P_j, s_j\} = 0.
      \end{cases}\\
      & = \begin{cases}
        \pm 1 , & \{P_j, s_j\} = 0, Q_j=is_jP_j \\
        1, & [P_j, s_j] = 0, Q_j= s_j\\
        0, & others.  \\
      \end{cases}
      \mathrm{when} \quad \theta_j\in\{\frac{\pi}{2},\frac{3\pi}{2}\}\quad \mathrm{or} 
    = \begin{cases}
        \pm 1 , & \{P_j, s_j\} = 0, Q_j= s_j\\
        1, & [P_j, s_j] = 0, Q_j= s_j\\
        0, & others  \\
      \end{cases}
      \mathrm{when} \quad \theta_j\in\{0,\pi\}.
  \end{aligned}
\end{equation}

It is easy to verify that \cref{eq:gradient_term_in_f_discrete} and \cref{eq:gate_term_in_f_discrete_repeat} become equivalent if we exchange the assignments $\theta_j \in \{0,\pi\}$ and $\theta_j \in \{\frac{\pi}{2}, \frac{3\pi}{2}\}$, while excluding the case where $[P_j, s_j] = 0$ in \cref{eq:gate_term_in_f_discrete_repeat}.
Then we have
\begin{equation}\label{eq:gradient_variance_proof}
\begin{aligned}
  \operatorname{Var}_{\bm{\theta}} \left[ \frac{\partial \loss}{\partial \theta_j} \right] & = \frac{1}{4^m}\sum_{\bm{\theta} \in \AngleSet^m} \sum_{\vec{s},\vec{s}\hspace{0.1em}'} \frac{\partial f(\vec{s},\bm{\theta},O,\rho)}{\partial{\theta_j}}\frac{\partial f(\vec{s}\hspace{0.1em}',\bm{\theta},O,\rho)}{\partial{\theta_j}}\\
  &=\frac{1}{4^m}\sum_{\bm{\theta} \in \AngleSet^m}\underset{\substack{\vec{s}:\{P_j, s_j\} = 0\\
     \vec{s}\hspace{0.1em}':\{P_j, s'_j\} = 0}}{\sum} f(\vec{s},\bm{\theta},O,\rho)f(\vec{s}\hspace{0.1em}',\bm{\theta},O,\rho)\\
  &= \ExpC  \left[\underset{\substack{\vec{s}:\{P_j, s_j\} = 0\\
     \vec{s}\hspace{0.1em}':\{P_j, s'_j\} = 0}}{\sum} f(\vec{s},\bm{\theta},O,\rho)f(\vec{s}\hspace{0.1em}',\bm{\theta},O,\rho)\right]\\\  
  &=\ExpC  \left[\underset{\substack{\vec{s}:\{P_j, s_j\} = 0\\
     \vec{s}\hspace{0.1em}':\{P_j, s'_j\} = 0}}{\sum} \sum_{\alpha,\beta}c_\alpha c_\beta f(\vec{s},\bm{\theta},P_\alpha,\rho)f(\vec{s}\hspace{0.1em}',\bm{\theta},P_\beta,\rho)\right]\\
&=\ExpC  \left[\underset{\substack{\vec{s}:\{P_j, s_j\} = 0}}{\sum}\sum_{\alpha}c_\alpha^2 f(\vec{s},\bm{\theta},P_\alpha,\rho)^2\right]\\
  &=\frac{1}{4^{m}} 
    \sum_{\substack{
        \bm{\theta} \in \AngleSet^m \\
        \{P_j, s_j^{(\bm{\theta},\alpha)}\} = 0
      }} \sum_\alpha c_\alpha^2 \, f(\bm{\theta},P_\alpha,\vec{s}^{\hspace{0.1em}(\bm{\theta},\alpha)},\rho)^2.\\  
\end{aligned}    
\end{equation}

The second-to-last inequality holds due to the orthogonality condition, and the last equality follows from the property of the quantum rotation 2-design, as proven in \cref{lem:cross_term}. \qed

Also, according to the proof of \cref{eq:expression_variance_gradient}, it is easily to derive the upper bound of the variance $\operatorname{Var}_{\bm{\theta}} \left[ \frac{\partial \loss}{\partial \theta_j} \right]$ when the orthogonality condition may not be satisfied:
\begin{corollary}\label{cor:gradient_variance_upper_bound}
For an arbitrary PQC $\mathcal{C}(\bm{\theta})$ and any parameter $\theta_j \in \bm{\theta}$, the variance of the gradient with respect to $\theta_j$ can be upper bounded as
  \begin{equation}
    \operatorname{Var}_{\bm{\theta}} \left[ \frac{\partial \loss}{\partial \theta_j} \right] \leq \left(\frac{\left\|O\right\|_{HS}}{\left\|O\right\|_{\min}} \right)^2\frac{1}{4^{m}} 
    \sum_{\substack{
        \bm{\theta} \in \AngleSet^m \\
        \{P_j, s_j^{(\bm{\theta},\alpha)}\} = 0
      }} \sum_\alpha c_\alpha^2 \, f(\bm{\theta},P_\alpha,\vec{s}^{\hspace{0.1em}(\bm{\theta},\alpha)},\rho)^2,
  \end{equation}
  where $\left\|O\right\|_{HS}:= \sqrt{\frac{\tr{O^2}}{2^n}} = \sqrt{\sum_\alpha c_\alpha^2}$ denotes as the Hilbert-Schmidt norm of $O$ and $\left\|O\right\|_{\min}:=\min\{\left|c_\alpha\right|>0\}$.
  Here, the orthogonality condition in \cref{eq:orthogonality_condition_app} is not required to hold.
\end{corollary}
\begin{proof}
According to \cref{eq:gradient_variance_proof}, for arbitrary PQC $\mathcal{C}(\bm{\theta})$, when the orthogonality condition may not hold, we have
\begin{equation}
\operatorname{Var}_{\bm{\theta}} \left[ \frac{\partial \loss}{\partial \theta_j} \right]= 
\ExpC \left[\sum_\alpha c_\alpha\sum_{\vec{s}:\{P_j, s_j\} = 0} f(\vec{s},\bm{\theta},P_\alpha,\rho) \right] ^2.
\end{equation}
Applying the Cauchy-Schwarz inequality to the summation, the variance of the gradient can be upper bounded as follows:
\begin{equation}\label{eq:gradient_variance_lower_bound_temp}
\begin{aligned}
  &\operatorname{Var}_{\bm{\theta}} \left[ \frac{\partial \loss}{\partial \theta_j} \right]\\ 
&= 
\ExpC \left[\sum_\alpha c_\alpha\sum_{\vec{s}:\{P_j, s_j\} = 0} f(\vec{s},\bm{\theta},P_\alpha,\rho) \right] ^2\\
&\leq 
\ExpC \left[\left(\sum_\alpha c_\alpha^2\right) \sum_\alpha \left(\sum_{\vec{s}:\{P_j, s_j\} = 0} f(\vec{s},\bm{\theta},P_\alpha,\rho)\right)^2\right] \\
&\leq \left\|O\right\|_{HS}^2 
 \ExpC  \left[\sum_\alpha \frac{c_\alpha^2}{\min\{c_\alpha^2\}} \left(\sum_{\vec{s}:\{P_j, s_j\} = 0} f(\vec{s},\bm{\theta},P_\alpha,\rho)\right)^2\right] \\
&=\left(\frac{\left\|O\right\|_{HS}}{\left\|O\right\|_{\min}} \right)^2
 \ExpC  \left[\sum_\alpha c_\alpha^2 \left(\sum_{\vec{s}:\{P_j, s_j\} = 0}f(\vec{s},\bm{\theta},P_\alpha,\rho)\right)^2\right]. \\
\end{aligned}    
\end{equation}
Then according to \cref{eq:orthogonality_same_Pauli_O}, the cross terms in \cref{eq:gradient_variance_lower_bound_temp} vanishes, then we have
\begin{equation}
\begin{aligned}
  &\operatorname{Var}_{\bm{\theta}} \left[ \frac{\partial \loss}{\partial \theta_j} \right]\\ 
&\leq\left(\frac{\left\|O\right\|_{HS}}{\left\|O\right\|_{\min}} \right)^2
 \ExpC  \left[\sum_\alpha \sum_{\vec{s}:\{P_j, s_j\} = 0}c_\alpha^2f(\vec{s},\bm{\theta},P_\alpha,\rho)^2 \right]\\
&=\left(\frac{\left\|O\right\|_{HS}}{\left\|O\right\|_{\min}} \right)^2\frac{1}{4^{m}} 
    \sum_{\substack{
        \bm{\theta} \in \AngleSet^m \\
        \{P_j, s_j^{(\bm{\theta},\alpha)}\} = 0
      }} \sum_\alpha c_\alpha^2 \, f(\bm{\theta},P_\alpha,\vec{s}^{\hspace{0.1em}(\bm{\theta},\alpha)},\rho)^2.\\ 
&=\order{\mathrm{poly}(n)}\frac{1}{4^{m}} 
    \sum_{\substack{
        \bm{\theta} \in \AngleSet^m \\
        \{P_j, s_j^{(\bm{\theta},\alpha)}\} = 0
      }} \sum_\alpha c_\alpha^2 \, f(\bm{\theta},P_\alpha,\vec{s}^{\hspace{0.1em}(\bm{\theta},\alpha)},\rho)^2,\\ 
\end{aligned}    
\end{equation}
where the last equality follows from \cref{eq:upper_bound_HS}.
\end{proof}

\section{Variance and gradient variance of the loss function of MPQCs}\label{app:variance_MPQC}

In this section, we leverage \cref{lem:two_variances_MPQC} to derive analytical expressions for both the variance of the loss function of MPQCs and that of its gradient.
To apply this lemma, it is necessary to prove that the Pauli path of MPQC satisfies the orthogonality condition, and that for any $P_\alpha$, the quantity $\tr{\channelM{\bm{\theta},\bm{\theta}_\mathcal{G}}(\rho) P_\alpha}$ is not a non-zero constant function of $\left(\bm{\theta},\bm{\theta}_\mathcal{G}\right)$.

We first express the variance of the MPQC in terms of the Pauli path integral. Suppose $\unitaryM{\bm{\theta},\bm{\theta}_\mathcal{G}}$ denotes the unitary representation of $\channelM{\bm{\theta},\bm{\theta}_\mathcal{G}}$ that includes the ancilla qubits but excludes all $op$. Instead, the operations $op$ are treated explicitly as acting on the initial state of the ancilla qubits. Then, the loss function reads
\begin{equation}
  \begin{aligned}
\lossM&=\tr{\channelM{\bm{\theta},\bm{\theta}_\mathcal{G}}\left(\rho\right)O}\\
    &= \tr{ \left[\unitaryM{\bm{\theta},\bm{\theta}_\mathcal{G}}
    \left( op\left(\ket{0}\bra{0}\right)^{\otimes n}
    \otimes \rho \right) \left(\unitaryM{\bm{\theta},\bm{\theta}_\mathcal{G}}\right)^\dagger\right]
    \cdot \left[ I \otimes O \right] },
  \end{aligned}
\end{equation}
Here, the first $n$ qubits are the ancilla qubits, and the last $n$ qubits correspond to the original PQC, which we will refer to as the \emph{system qubits} in the following discussion. The observable operator acting on the ancilla qubits is fixed to be $I$, according to the definition of the quantum channel.

Next we express $\unitaryM{\bm{\theta},\bm{\theta}_\mathcal{G}}$ as the form in \cref{eq:parameterized_circuit}:
\begin{equation}
  \unitaryM{\bm{\theta},\bm{\theta}_\mathcal{G}} = \mathbf{U}_m(\theta_m)\cdots \mathbf{U}_{L+1}(\theta_{L+1}) \prod_{i=1}^n\left(R_{Z_iZ_{i+n}}(\bm{\theta}_{\mathcal{G}_{i,1}})
R_{Y_iY_{i+n}}(\bm{\theta}_{\mathcal{G}_{i,2}})
R_{X_iX_{i+n}}(\bm{\theta}_{\mathcal{G}_{i,3}})\right)\mathbf{U}_L(\theta_L) 
\cdots\mathbf{U}_1(\theta_1),
\end{equation}
where $\mathbf{U}_i(\theta_i)$ denote the unitary operator corresponding to the original circuit acting on $2n$ qubits, i.e. $\mathbf{U}_i(\theta_i) = I\otimes U_i(\theta_i)$. For convenience in the subsequent proof, we denote $R_{i,j}(\bm{\theta}_{\mathcal{G}_{i,j}})$ as the 2-qubit rotation gate $R_{P(j)_iP(j)_{i+n}}(\bm{\theta}_{\mathcal{G}_{i,j}})$, where $i\in[n], j\in[3]$, and $P(1) = Z$, $P(2) = Y$, $P(3) = X$.

Following the procedure in \cref{eq:Pauli_path_integral_noiseless},
we expand the loss function $\lossM$ using the Pauli path integral formalism:
\begin{small}
\begin{equation}\label{eq:MPQC_Pauli_path_integral_noiseless}
  \begin{aligned}
&\lossM=\tr{\unitaryM{\bm{\theta},\bm{\theta}_\mathcal{G}}op\left(\ket{0}\bra{0}\right)^{\otimes n}
    \otimes \rho\left(\unitaryM{\bm{\theta},\bm{\theta}_\mathcal{G}}\right)^\dagger I \otimes O}\\
    &= \sum_{\alpha,\mathbf{s}_m} c_\alpha\tr{I\otimes P_\alpha \mathbf{s}_m} \tr{\mathbf{s}_m \unitaryM{\bm{\theta},\bm{\theta}_\mathcal{G}}op\left(\ket{0}\bra{0}\right)^{\otimes n}
\otimes\rho\left(\unitaryM{\bm{\theta},\bm{\theta}_\mathcal{G}}\right)^\dagger}\\
  &=\sum_{\substack{\alpha,\mathbf{s}_m,\mathbf{s}_{m-1},\cdots,\mathbf{s}_0\\  \mathbf{s}_{\mathcal{G}_{1,1}},\mathbf{s}_{\mathcal{G}_{1,2}},\cdots,\mathbf{s}_{\mathcal{G}_{n,3}}
  }} \tr{I\otimes O \mathbf{s}_m} 
  \tr{\mathbf{s}_m \mathbf{U}_m(\theta_m) \mathbf{s}_{m-1} \mathbf{U}_m(\theta_m)^\dagger} \cdots \tr{\mathbf{s}_{L+1} \mathbf{U}_{L+1}(\theta_{L+1}) \mathbf{s}_{\mathcal{G}_{1,1}} \mathbf{U}_{L+1}(\theta_{L+1})^\dagger}\cdot\\ 
  &\qquad
  \cdot \tr{\mathbf{s}_{\mathcal{G}_{1,1}} R_{11}(\bm{\theta}_{\mathcal{G}_{11}}) \mathbf{s}_{\mathcal{G}_{1,2}} R_{11}(-\bm{\theta}_{\mathcal{G}_{11}})} \tr{\mathbf{s}_{\mathcal{G}_{1,2}}R_{12}(\bm{\theta}_{\mathcal{G}_{12}}) \mathbf{s}_{\mathcal{G}_{1,3}} R_{12}(-\bm{\theta}_{\mathcal{G}_{12}})}\cdots\tr{\mathbf{s}_{\mathcal{G}_{n,3}}R_{n3}(\bm{\theta}_{\mathcal{G}_{n3}}) \mathbf{s}_{L} R_{n3}(-\bm{\theta}_{\mathcal{G}_{n3}})}\cdot\\
  &\qquad \cdot \tr{\mathbf{s}_{L} \mathbf{U}_L(\theta_L) \mathbf{s}_{L-1} \mathbf{U}_L(\theta_L)^\dagger}
  \cdots\tr{\mathbf{s}_{1} \mathbf{U}_1(\theta_1) \mathbf{s}_0 \mathbf{U}_1(\theta_1)^\dagger} \tr{\mathbf{s}_0 op\left(\ket{0}\bra{0}\right)^{\otimes n}\otimes\rho}\\
  & =\sum_{\alpha,\vec{\mathbf{s}}} c_\alpha f\left(\vec{\mathbf{s}},\left(\bm{\theta},\bm{\theta}_\mathcal{G}\right),I\otimes P_\alpha,op\left(\ket{0}\bra{0}\right)^{\otimes n}
\otimes\rho\right),
  \end{aligned}
\end{equation}  
\end{small}
where we define $\vec{\mathbf{s}}=(\mathbf{s}_0,\cdots,\mathbf{s}_{m}, \mathbf{s}_{\mathcal{G}_{n,3}},\mathbf{s}_{\mathcal{G}_{n,2}},\cdots,\mathbf{s}_{\mathcal{G}_{1,1}})$ with each element a normalized 2$n$-qubit Pauli operator and $f\left(\vec{\mathbf{s}},\left(\bm{\theta},\bm{\theta}_\mathcal{G}\right),I\otimes P_\alpha,op\left(\ket{0}\bra{0}\right)^{\otimes n}
\otimes\rho\right)$ as the contribution of Pauli path $\vec{\mathbf{s}}$ to the expectation value. To prove that MPQC satisfies the conditions demanded in \cref{lem:two_variances_MPQC}, we need the following \cref{lem:orthogonality} and \cref{lem:split} proved in Ref.~\cite{shao2024simulating}.

\begin{lemma}\label{lem:orthogonality}
Consider a PQC $\mathcal{C}(\bm{\theta})=U_m(\theta_m)  \cdots {U}_1(\theta_1)$ measured with observable $ O = \sum_{\alpha} c_\alpha P_\alpha$.
Let $\overline{P_i}$ denote the Pauli operator $P_i$ after conjugation by a sequence of Clifford gates, i.e., $\overline{P_i} = C_m \cdots C_i P_i C_i^\dagger \cdots C_m^\dagger$. Then the orthogonality condition \cref{eq:orthogonality_condition_app} holds for $\mathcal{C}(\bm{\theta})$ if the set of Pauli operators $\{\overline{P_i}\}$ can split the Pauli operator set $\{P_\alpha\}$ of $O$. We say that  Pauli set $A$ can split Pauli set $B$ if there exist no two distinct elements in $B$ that exhibit identical anti-commute/commute relation with each element in $A$.
\end{lemma}

\begin{lemma}\label{lem:split}
$\{\overline{P_i}\}$ can split the entire $n$-qubit Pauli $\{\mathbb{I}, X, Y, Z\}^{\otimes n}$ is equivalent to the condition that 
\begin{equation}\label{eq:generate}
  \langle \{\overline{P_i}\}\rangle/\left(\langle \{\overline{P_i}\}\rangle\cap\langle i\mathbb{I}^{\otimes n}\rangle\right)=\{\mathbb{I},X,Y,Z\}^{\otimes n},
\end{equation}
here $\langle \{\overline{P_i}\} \rangle$ denotes to the Pauli subgroup that is generated by set $\{\overline{P_i}\}$, meaning every element in $\langle \{\overline{P_i}\}\rangle$ can be expressed as the finite product of elements in $\{\overline{P_i}\}$.
\end{lemma}

Next, we prove that two conditions of \cref{lem:two_variances_MPQC} are both satisfied for arbitrary MPQC, which are concluded in the following two lemmas.

\begin{lemma}\label{lem:orthogonality_MPQC}
 Consider a MPQC $\channelM{\bm{\theta},\bm{\theta}_\mathcal{G}}$ taking in parameters $\left(\bm{\theta},\bm{\theta}_\mathcal{G}\right)\in \left[0, 2\pi \right)^{m+3n}$ measured with observable $O = \sum_{\alpha} c_\alpha P_\alpha$. Then, the orthogonality condition for the Pauli paths in the expansion form of \cref{eq:MPQC_Pauli_path_integral_noiseless} always holds.
\end{lemma}
\begin{proof}
  We employ \cref{lem:orthogonality} and \cref{lem:split} to prove \cref{lem:orthogonality_MPQC}. To facilitate the analysis, we first express the Pauli operators generated by the MPQC in \cref{lem:orthogonality} to act on the full $2n$-qubit system.
Specifically, these operators can be written as
\[
\overline{\mathbf{P}_i} = \mathbf{C}_m \cdots \mathbf{C}_i \mathbf{P}_i \mathbf{C}_i^\dagger \cdots \mathbf{C}_m^\dagger,
\]
where each $\mathbf{C}_i$ denotes a Clifford operator in the original PQC, extended to act on $2n$ qubits.

Recall that each gadget $\mathcal{G}(\boldsymbol{\theta})$ employs three two-qubit rotation gates: $R_{XX}$, $R_{YY}$, and $R_{ZZ}$, acting between a system qubit and an ancilla qubit. For each system qubit, at least one such gadget is applied. Then, the Pauli operator set generated in the gadget layer contains at least the following:
\[
\{X^{j_1}_{i_1}Y^{j_2}_{i_2}Z^{j_3}_{i_3} \otimes \tilde{\mathbf{C}}_{L+1}X^{j_1}_{i_1+n}Y^{j_2}_{i_2+n}Z^{j_3}_{i_3+n}\tilde{\mathbf{C}}_{L+1}^\dagger\}_{\tiny\substack{i_1,i_2,i_3 \\ j_1,j_2,j_3}} := F,
\]
where $i_1, i_2, i_3 \in [n]$, $j_1, j_2, j_3 \in \{0,1\}$ satisfying $j_1 + j_2 + j_3 = 1$,  and $\tilde{\mathbf{C}}_{L+1} \coloneq \mathbf{C}_m \cdots \mathbf{C}_{L+1}$. Since the Pauli operator set of the observable of MPQCs takes the form $\{I \otimes P_\alpha\}$, the (anti)commutation relation between any element 
\[
X^{j_1}_{i_1}Y^{j_2}_{i_2}Z^{j_3}_{i_3} \otimes \tilde{\mathbf{C}}_{L+1}X^{i_1}_{j_1+n}Y^{j_2}_{i_2+n}Z^{j_3}_{i_3+n}\tilde{\mathbf{C}}_{L+1}^\dagger \in F
\]
and $I \otimes P_\alpha$ is determined by the (anti)commutation relation between $\tilde{\mathbf{C}}_{L+1}X^{j_1}_{i_1+n}Y^{j_2}_{i_2+n}Z^{j_3}_{i_3+n}\tilde{\mathbf{C}}_{L+1}^\dagger$ and $P_\alpha$.

This implies that $F$ can split the Pauli operator set $\{I \otimes P_\alpha\}$ if and only if the set $\{\tilde{\mathbf{C}}_{L+1}X^{j_1}_{i_1+n}Y^{j_2}_{i_2+n}Z^{j_3}_{i_3+n}\tilde{\mathbf{C}}_{L+1}^\dagger\}$ can split $\{P_\alpha\}$. 

It is straightforward to verify that $\{\tilde{\mathbf{C}}_{L+1}X^{j_1}_{i_1+n}Y^{j_2}_{i_2+n}Z^{j_3}_{i_3+n}\tilde{\mathbf{C}}_{L+1}^\dagger\}$ generates the entire $n$-qubit Pauli group. By \cref{lem:split}, we conclude that the operator $F$ already suffices to split the Pauli operator set of any observable $O$. Consequently, the whole Pauli operator set $\{\overline{\mathbf{P}_i}\}$ of MPQC can split the Pauli operator set of arbitrary $O$. According to \cref{lem:orthogonality}, we thus conclude that the orthogonality condition holds for all MPQCs.
\end{proof}

\begin{lemma}\label{lem:non_constant_MPQC}
For any MPQC and any nontrivial $n$-qubit Pauli word $P$, the expectation value $\Vexp{P}$ is not a non-zero constant function of the parameters in MPQC.
\end{lemma}
\begin{proof}
Suppose there exists an MPQC $\channelM{\bm{\theta},\bm{\theta}_\mathcal{G}}$ and a nontrivial Pauli operator $P \neq I$ such that $\Vexp{P} = c \neq 0$. Then we have
\begin{equation}
\ExpMC \left[\Vexp{P}\right] = \ExpMC \left[\sum_{\vec{\mathbf{s}}} f\left(\vec{\mathbf{s}},\left(\bm{\theta},\bm{\theta}_\mathcal{G}\right),I\otimes P,op\left(\ket{0}\bra{0}\right)^{\otimes n}
\otimes\rho\right)\right] = c \neq 0.
\end{equation}
For arbitrary $\vec{\mathbf{s}}$, if there exists some $\{\mathbf{s}_i,\mathbf{P}_i\}=0$ or some $\{\mathbf{s}_{\mathcal{G}_{i,j}}, P(j)_iP(j)_{i+n}\}= 0$, then the corresponding term $\ExpC f\left(\vec{\mathbf{s}},\left(\bm{\theta},\bm{\theta}_\mathcal{G}\right),I\otimes P,op\left(\ket{0}\bra{0}\right)^{\otimes n}
\otimes\rho\right)$ vanishes.
This is because, according to~\eqref{eq:gate_term_in_f}, this term must contain one of the following components:
\begin{equation}
  \begin{aligned}
&\mathbb{E}_{\theta_i} \left[ \cos(\theta_i) \tr{ \mathbf{s}_i \mathbf{C}_i \mathbf{s}_{i-1} \mathbf{C}_i^\dagger }\right] \quad \text{or} \quad 
\mathbb{E}_{\theta_i}\left[\sin(\theta_i) \tr{ \mathbf{s}_i P_i \mathbf{C}_i \mathbf{s}_{i-1} \mathbf{C}_i^\dagger }, \right], \quad i\neq L+1\\      
\text{or} \quad& \mathbb{E}_{\theta_{L+1}} \left[ \cos(\theta_{L+1}) \tr{ \mathbf{s}_{L+1} \mathbf{C}_{L+1} \mathbf{s}_{\mathcal{G}_{1,1}} \mathbf{C}_{L+1}^\dagger }\right] \quad \text{or} \quad 
\mathbb{E}_{\theta_{L+1}}\left[\sin(\theta_{L+1}) \tr{ \mathbf{s}_{L+1} P_{L+1} \mathbf{C}_{L+1} \mathbf{s}_{\mathcal{G}_{1,1}} \mathbf{C}_{L+1}^\dagger }, \right]\\
\text{or} \quad& \mathbb{E}_{\theta_{\mathcal{G}_{i,j}}} \left[ \cos(\theta_{\mathcal{G}_{i,j}}) \tr{ \mathbf{s}_{\mathcal{G}_{i,j}} \mathbf{s}_{\mathcal{G}_{i,j+1}} }\right] \quad \text{or} \quad 
\mathbb{E}_{\theta_{\mathcal{G}_{i,j}}}\left[\sin(\theta_{\mathcal{G}_{i,j}}) \tr{ \mathbf{s}_{\mathcal{G}_{i,j}}
P(j)_iP(j)_{i+n}
\mathbf{s}_{\mathcal{G}_{i,j+1}}}\right], \quad i\in[n],j\leq 2\\
\text{or} \quad& 
\mathbb{E}_{\theta_{\mathcal{G}_{i,3}}}\left[\cos(\theta_{\mathcal{G}_{i,3}}) \tr{ \mathbf{s}_{\mathcal{G}_{i,3}}
\mathbf{s}_{\mathcal{G}_{i+1,1}} }\right]\quad 
\text{or} \quad 
\mathbb{E}_{\theta_{\mathcal{G}_{i,3}}}\left[\sin(\theta_{\mathcal{G}_{i,3}}) \tr{ \mathbf{s}_{\mathcal{G}_{i,1}}
P(3)_iP(3)_{i+n}
\mathbf{s}_{\mathcal{G}_{i+1,1}} }\right], \quad i\leq n-1\\
\text{or} \quad& \mathbb{E}_{\mathcal{G}_{n,3}} \left[ \cos(\theta_{\mathcal{G}_{n,3}}) \tr{ \mathbf{s}_{\mathcal{G}_{n,3}} \mathbf{s}_L }\right] \quad \text{or} \quad 
\mathbb{E}_{\theta_{\mathcal{G}_{n,3}}}\left[\sin(\theta_{\mathcal{G}_{n,3}}) \tr{ \mathbf{s}_{\mathcal{G}_{n,3}} P(3)_nP(3)_{2n} \mathbf{C}_{L+1} \mathbf{s}_L }, \right],
  \end{aligned}
\end{equation}
while all of them equal 0.

Based on this observation, if 
\[
\ExpMC \left[\sum_{\vec{\mathbf{s}}} f\left(\vec{\mathbf{s}},\left(\bm{\theta},\bm{\theta}_\mathcal{G}\right),I\otimes P,op\left(\ket{0}\bra{0}\right)^{\otimes n}
\otimes\rho\right)\right] = c \neq 0,
\]
Then, there must exist a Pauli path $\vec{\mathbf{s}}$ such that each element commutes with the corresponding generator of its associated rotation gate. According to~\eqref{eq:gate_term_in_f}, the following equation must hold:
\begin{equation}\label{eq:commute_relation}
\mathbf{C}_i \mathbf{s}_{i-1} \mathbf{C}_i^\dagger = \mathbf{s}_i, \quad i>L+2. 
\end{equation}
Here we again we ignore the sign $\pm$ in front of the Pauli operator as we only concern the commutation relation between Pauli operators. By recursively applying~\eqref{eq:commute_relation} and $\mathbf{C}_{L+1} \mathbf{s}_{\mathcal{G}_{1,1}} \mathbf{C}_{L+1}^\dagger = \mathbf{s}_{L+1}$, we obtain
\begin{equation}
\mathbf{s}_{\mathcal{G}_{1,1}}  = \mathbf{C}_{L+1}^\dagger \cdots \mathbf{C}_m^\dagger \, \mathbf{s}_m \, \mathbf{C}_m \cdots \mathbf{C}_{L+1}.
\end{equation}
By applying the same procedure to the Pauli operators of Pauli path that pass through the gadget layers, we obtain that
$\mathbf{s}_{\mathcal{G}_{i,j}} = \mathbf{s}_{\mathcal{G}_{1,1}}$. Due to the commutation condition, this implies that $[\mathbf{s}_{\mathcal{G}_{1,1}}, P(j)_iP(j)_{i+n}] = 0$ for all $i,j$.

Since $\mathbf{s}_{\mathcal{G}_{1,1}}$ commutes with all $P(j)_iP(j)_{i+n}$, it follows that
\[
\mathbf{C}_{L+1}^\dagger \cdots \mathbf{C}_m^\dagger \, \mathbf{s}_m \, \mathbf{C}_m \cdots \mathbf{C}_{L+1} \, P(j)_iP(j)_{i+n}
= P(j)_iP(j)_{i+n} \, \mathbf{C}_{L+1}^\dagger \cdots \mathbf{C}_m^\dagger \, \mathbf{s}_m \, \mathbf{C}_m \cdots \mathbf{C}_{L+1}.
\]
This implies
\[
\mathbf{s}_m \, \mathbf{C}_m \cdots \mathbf{C}_{L+1} \, P(j)_iP(j)_{i+n} \, \mathbf{C}_{L+1}^\dagger \cdots \mathbf{C}_m^\dagger 
= \mathbf{C}_m \cdots \mathbf{C}_{L+1} \, P(j)_iP(j)_{i+n} \, \mathbf{C}_{L+1}^\dagger \cdots \mathbf{C}_m^\dagger \mathbf{s}_m,
\]
and hence $\mathbf{s}_m$ commutes with each $\overline{\mathbf{P}}_{i,j} := \mathbf{C}_m \cdots \mathbf{C}_{L+1} \, P(j)_iP(j)_{i+n} \, \mathbf{C}_{L+1}^\dagger \cdots \mathbf{C}_m^\dagger$. Through the same analysis in \cref{lem:orthogonality_MPQC}, $[\mathbf{s}_m,\overline{\mathbf{P}}_{i,j}] = 0$ if and only if $[P,{C}_m \cdots {C}_{L+1}P(j)_{i+n}{C}_{L+1}^\dagger \cdots{C}_m^\dagger ] = 0$. Also since $\{{C}_m \cdots {C}_{L+1}P(j)_{i+n}{C}_{L+1}^\dagger \cdots{C}_m^\dagger\}$ can generate the entire $n$-qubit Pauli group, by \cref{lem:split}, the Pauli word can commute with all ${C}_m \cdots {C}_{L+1}P(j)_{i+n}{C}_{L+1}^\dagger \cdots{C}_m^\dagger$ must be $I$. This leads to a contradiction with the assumption that $P \neq I$.
\end{proof}

\cref{lem:orthogonality_MPQC} and \cref{lem:non_constant_MPQC} guarantee that the circuit architecture of MPQC always satisfies the conditions required in \cref{lem:two_variances_MPQC}. Therefore, according to \cref{lem:two_variances_MPQC} both the variance and the gradient variance of the loss function can be explicitly expressed using the Pauli path integral formulation:

\begin{lemma}\label{lem:two_variance_MPQC}
  Consider a MPQC $\channelM{\bm{\theta},\bm{\theta}_\mathcal{G}}$ taking in parameters $\left(\bm{\theta},\bm{\theta}_\mathcal{G}\right)\in \left[0, 2\pi \right)^{m+3n}$ measured with observable $O = \sum_{\alpha} c_\alpha P_\alpha$. The variance and the gradient variance of the loss function $\lossM$ under the Pauli path integral formulation can be expressed as
  \begin{equation}\label{eq:variance_MPQC}
      \operatorname{Var}_{(\bm{\theta},\bm{\theta}_\mathcal{G})}\left[\lossM\right] = \frac{1}{4^{m+3n}} 
      \sum_{\substack{\bm{\theta} \in \AngleSet^m \\
      \bm{\theta}_\mathcal{G} \in \AngleSet^{3n}}} 
      \sum_{\alpha} c_\alpha^2 f\left(\uniquePathM,\left(\bm{\theta},\bm{\theta}_\mathcal{G}\right),I\otimes P_\alpha,op\left(\ket{0}\bra{0}\right)^{\otimes n}
\otimes\rho\right)^2.      
  \end{equation}
\begin{equation}\label{eq:gradient_variance_MPQC}
    \operatorname{Var}_{(\bm{\theta},\bm{\theta}_\mathcal{G})}\left[ \frac{\partial \lossM}{\partial \theta_j}\right] = \frac{1}{4^{m+3n}} 
    \sum_{\substack{\bm{\theta} \in \AngleSet^m \\
    \bm{\theta}_\mathcal{G} \in \AngleSet^{3n}\\
        \{\mathbf{P}_j, \mathbf{s}_j^{(\bm{\theta},\alpha)}\} = 0}} 
    \sum_{\alpha} c_\alpha^2 f\left(\uniquePathM,\left(\bm{\theta},\bm{\theta}_\mathcal{G}\right),I\otimes P_\alpha,op\left(\ket{0}\bra{0}\right)^{\otimes n}
\otimes\rho\right)^2,
\end{equation}
where $\uniquePathM$ is the unique Pauli path such that $f\left(\uniquePathM,\left(\bm{\theta},\bm{\theta}_\mathcal{G}\right),I\otimes P_\alpha,op\left(\ket{0}\bra{0}\right)^{\otimes n}
\otimes\rho\right) \neq 0$, if $\tr{\mathbf{s}_0op\left(\ket{0}\bra{0}\right)^{\otimes n}
\otimes\rho}\neq 0$.
\end{lemma}

\section{Proof of \cref{thm:absence_BP}}\label{app:proof_absence_BP}
\subsection{Impact of the gadget $\mathcal{G}(\boldsymbol{\theta})$ on pauli paths}\label{app:impact_gadget}
Here, we discuss the impact of $\mathcal{G}(\boldsymbol{\theta})$ on the Pauli path in the Heisenberg picture. Suppose that the Pauli operator at the output of the gadget $\mathcal{G}(\boldsymbol{\theta})$ in \cref{fig:gadget_analysis} is $I \otimes P$. Then, for certain subsets of angle choices $\theta_1, \theta_2, \theta_3 \in \AngleSet$, we can determine the corresponding Pauli operators $P_1$, $P_2$, and $P_3$, where $\theta_1$, $\theta_2$, and $\theta_3$ are the rotation angles for the $R_{XX}$, $R_{YY}$, and $R_{ZZ}$ gates, respectively.

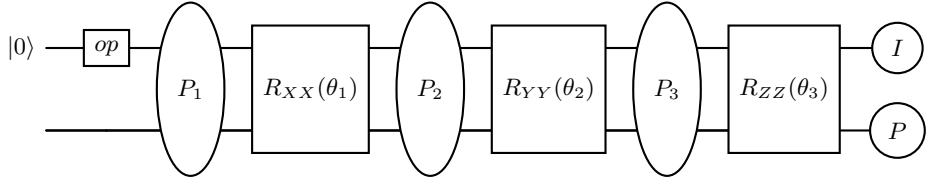
\begin{figure}[h]
  \centering
  \begin{quantikz}
    \lstick{$\ket{0}$}  & \gate[1]{op} &\gate[2, style={draw, shape=ellipse, minimum width=0.4cm, minimum height=0.6cm}]{\raisebox{0.2ex}{$P_1$}}& \gate[2]{R_{XX}(\theta_1)} & \gate[2, style={draw, shape=ellipse, minimum width=0.4cm, minimum height=0.6cm}]{\raisebox{0.2ex}{$P_2$}}& \gate[2]{R_{YY}(\theta_2)} & \gate[2, style={draw, shape=ellipse, minimum width=0.4cm, minimum height=0.6cm}]{\raisebox{0.2ex}{$P_3$}}& \gate[2]{R_{ZZ}(\theta_3)}& \gate[1, style={draw, shape=circle}]{I} \\
   & \qw &\qw &\qw &\qw & \qw & \qw & \qw& \gate[1, style={draw, shape=circle}]{P}\\
  \end{quantikz}
  \caption{Effect of the gadget $\mathcal{G}(\boldsymbol{\theta})$ on Pauli paths in the Heisenberg picture.}\label{fig:gadget_analysis}
\end{figure}
Analyzing the backward propagation of Pauli path, direct calculation based on~\eqref{eq:gate_term_in_f_discrete} yields that
\begin{itemize}
  \item $P = I$, $\theta_1, \theta_2, \theta_3 \in \AngleSet$, $P_1 = P_2 = P_3 = II$.
  \item $P = X$, $\{IX, ZZ\} = 0$, $\theta_3 \in \{\frac{\pi}{2},\frac{3\pi}{2}\}\rightarrow P_3 = ZY$; $\{ZY, YY\} = 0$, $\theta_2 \in \{\frac{\pi}{2},\frac{3\pi}{2}\} \rightarrow P_2 = XI$; $[XI, XX] = 0$, $\theta_1 \in \AngleSet \rightarrow P_1 = XI$.
  \item $P = Y$, $\{IY, ZZ\} = 0$, $\theta_3 \in \{\frac{\pi}{2},\frac{3\pi}{2}\}\rightarrow P_3 = ZX$; $[ZX, YY] = 0$, $\theta_2 \in \AngleSet \rightarrow P_2 = ZX$; $\{ZX, XX\} = 0$, $\theta_1 \in \{\frac{\pi}{2},\frac{3\pi}{2}\} \rightarrow P_1 = YI$.
  \item $P = Z$, $[IZ, ZZ] = 0$, $\theta_3 \in \AngleSet\rightarrow P_3 = IZ$; $\{IZ, YY\} = 0$, $\theta_2 \in \{\frac{\pi}{2},\frac{3\pi}{2}\} \rightarrow P_2 = YX$; $\{YX, XX\} = 0$, $\theta_1 \in \{\frac{\pi}{2},\frac{3\pi}{2}\} \rightarrow P_1 = ZI$.
\end{itemize}
Here we also we ignore the sign $\pm$ in front of the Pauli operator. The above result indicates that, from the perspective of the Heisenberg picture, among the $64$ possible combinations of $\theta_1, \theta_2, \theta_3 \in \AngleSet$, there exist at least $2 \times 2 \times 4 = 16$ configurations that lead to $P_1 = P \otimes I$.

Moreover, for any given single-qubit Pauli operator $P$, there exist $\theta_1, \theta_2, \theta_3 \in \AngleSet$ such that $P_1 = I \otimes P$. Specifically, we have:

\begin{itemize}
  \item $P = I$, $\theta_1, \theta_2, \theta_3 \in \AngleSet$, $P_1 = P_2 = P_3 = II$.
  \item $P = X$, $\{IX, ZZ\} = 0$, $\theta_3 \in \{0,\pi\}\rightarrow P_3 = IX$; $\{IX, YY\} = 0$, $\theta_2 \in \{0,\pi\} \rightarrow P_2 = IX$; $[IX, XX] = 0$, $\theta_1 \in \AngleSet \rightarrow P_1 = IX$.
  \item $P = Y$, $\{IY, ZZ\} = 0$, $\theta_3 \in \{0,\pi\}\rightarrow P_3 = IY$; $[IY, YY] = 0$, $\theta_2 \in \AngleSet \rightarrow P_2 = IY$; $\{IY, XX\} = 0$, $\theta_1 \in \{0,\pi\} \rightarrow P_1 = IY$.
  \item $P = Z$, $[IZ, ZZ] = 0$, $\theta_3 \in \AngleSet\rightarrow P_3 = IZ$; $\{IZ, YY\} = 0$, $\theta_2 \in \{0,\pi\} \rightarrow P_2 = IZ$; $\{IZ, XX\} = 0$, $\theta_1 \in \{0,\pi\} \rightarrow P_1 = IZ$.
\end{itemize}
Therefore, there also exist 16 choices of $\theta_1, \theta_2, \theta_3 \in \AngleSet$ that leave the Pauli path unchanged; that is, the resulting Pauli operator $P_1$ remains $I \otimes P$.

Next, we analyze the remaining 32 configurations of  $\theta_1, \theta_2, \theta_3 \in \AngleSet$ when the Pauli operator $P$ is non-trivial. Consider the case $P = X$ as an example. The analysis proceeds as follows:
\begin{itemize}
  \item $P = X$, $\{IX, ZZ\} = 0$, $\theta_3 \in \{0,\pi\}\rightarrow P_3 = IX$; $\{IX, YY\} = 0$, $\theta_2 \in \{\frac{\pi}{2},\frac{3\pi}{2}\} \rightarrow P_2 = YZ$; $[YZ, XX] = 0$, $\theta_1 \in \AngleSet \rightarrow P_1 = YZ$.
  \item $P = X$, $\{IX, ZZ\} = 0$, $\theta_3 \in \{\frac{\pi}{2},\frac{3\pi}{2}\}\rightarrow P_3 = ZY$; $\{ZY, YY\} = 0$, $\theta_2 \in \{0,\pi\} \rightarrow P_2 = ZY$; $[ZY, XX] = 0$, $\theta_1 \in \AngleSet \rightarrow P_1 = ZY$.
\end{itemize}
Thus, among these 32 configurations, 16 of them transform $IX$ to $YZ$, while the other 16 transform $IX$ to $ZY$. Following similar calculations, we find that: 
\begin{itemize}
  \item When $P = Y$, 16 configurations of $\theta_1, \theta_2, \theta_3 \in \AngleSet$  map $IY$ to $XZ$, and 16 to $ZX$.
  \item When $P = Z$ 16 configurations of $\theta_1, \theta_2, \theta_3 \in \AngleSet$ map $IZ$ to $XY$, and 16 to $YX$.
\end{itemize}

To summarize, among all $64$ possible angle combinations with $\theta_1, \theta_2, \theta_3 \in \AngleSet$, the operator $P_1$ has the following possibilities:
\begin{itemize}
  \item  If $P = I$, then $P_1 = II$ for all 64 configurations of $\theta_1, \theta_2, \theta_3 \in \AngleSet$.
  \item $P \neq I$, $P_1 = \begin{cases}
      PI, & \text{for 16 configurations of } \theta_1, \theta_2, \theta_3 \in \AngleSet; \\
      IP, & \text{for 16 configurations of } \theta_1, \theta_2, \theta_3 \in \AngleSet;\\
      Q_1Q_2, & \text{for 16 configurations of } \theta_1, \theta_2, \theta_3 \in \AngleSet;\\
      Q_2Q_1, & \text{for 16 configurations of } \theta_1, \theta_2, \theta_3 \in \AngleSet,
      \end{cases}$\\
  where $\{Q_1,Q_2,P\} = \{X,Y,Z\}$.
\end{itemize}
\begin{remark}
We analyze the effect of the gadget $\mathcal{G}(\boldsymbol{\theta})$ on the backward propagation of Pauli paths from an operational perspective. When the three parameters of the gadget are chosen from the discrete set $\AngleSet$, and $P \neq I$, we find that in $16$ out of the $64$ possible angle combinations—that is, in a proportion of $1/4$—the backward-propagated operator $IP$ is transformed via a ``swap'' operation. In another $1/4$ of the combinations, the Pauli operator remains unchanged during the backward propagation.

Furthermore, on the system qubit, for any given Pauli operator $P'$, there exists a proportion of $1/4$ among the total angle combinations for which that $P'$ appears after backward propagation when $P\neq I$. This reflects the uniformity of Pauli operator appearances under the action of the gadget when the angles are sampled from the discrete set.
\end{remark}
\subsection{Lower bound of the variance of the loss function of MPQC}\label{app:absence_BP_MPQC}
In this subsection, we derive a lower bound on the variance of the loss function for well-constructed MPQCs. According to Ref.~\cite{arrasmith2022equivalence}, a non-vanishing variance implies the absence of barren plateaus. Hence, our result confirms that MPQCs do not suffer from barren plateau.

The lower bound of the variance can be described by the following theorem:
\begin{theorem}\label{thm:absence_BP_MPQC}[\cref{thm:absence_BP}, formal version]
Consider a $k$-local observable $O = \sum_\alpha c_\alpha P_\alpha$~(i.e., each Pauli word $P_\alpha$ acts non-trivially on at most $k$ qubits) and an MPQC $\channelM{\bm{\theta},\bm{\theta}_\mathcal{G}}$ which is achieved by inserting a layer of the gadgets after the $l$-th layer (also, $U_L(\theta_L)$) of the PQC. Suppose for each Pauli word $P_\alpha$, the support size of its backward light cone at the gadget layer is upper bounded by $K=\order{\log n}$. Then the variance of the loss function $\lossM = \tr{\channelM{\bm{\theta},\bm{\theta}_\mathcal{G}}\left(\rho\right)O}$ is lower bounded by
  \[
\operatorname{Var}_{(\bm{\theta},\bm{\theta}_\mathcal{G})}\left[\lossM\right] \geq \left(\frac{\tau}{4}\right)^{K}\left\|O\right\|_{HS}^2 =\Omega\left(\frac{1}{\mathrm{poly}(n)}\right),
  \]
  where $\left\|O\right\|_{HS}= \sqrt{\frac{\tr{O^2}}{2^n}} = \sqrt{\sum_\alpha c_\alpha^2}$.
\end{theorem}
\begin{proof}
According to \cref{eq:variance_MPQC}, the variance of the loss function for the MPQC $\channelM{\bm{\theta},\bm{\theta}_\mathcal{G}}$ can be written and lower bounded as follows:
\begin{equation}
  \begin{aligned}
    \operatorname{Var}_{(\bm{\theta},\bm{\theta}_\mathcal{G})}\left[\lossM\right]&= \frac{1}{4^{m+3n}} 
      \sum_{\substack{\bm{\theta} \in \AngleSet^m \\
      \bm{\theta}_\mathcal{G} \in \AngleSet^{3n}}} 
      \sum_{\alpha} c_\alpha^2 f\left(\uniquePathM,\left(\bm{\theta},\bm{\theta}_\mathcal{G}\right),I\otimes P_\alpha,op\left(\ket{0}\bra{0}\right)^{\otimes n}
\otimes\rho\right)^2\\
    &\geq \frac{1}{4^{m+3n}} 
      \sum_{\substack{\bm{\theta} \in \AngleSet^m \\
      \bm{\theta}_\mathcal{G} \in M_{\mathrm{swap}}(\bm{\theta})}} 
      \sum_{\alpha} c_\alpha^2 f\left(\uniquePathM,\left(\bm{\theta},\bm{\theta}_\mathcal{G}\right),I\otimes P_\alpha,op\left(\ket{0}\bra{0}\right)^{\otimes n}
\otimes\rho\right)^2.\\
  \end{aligned}
\end{equation}  
Here, we consider a specific subset $M_{\mathrm{swap}}(\bm{\theta}) \subseteq \AngleSet^{3n}$, defined as the collection of discrete angle configurations such that, for each $\bm{\theta}_\mathcal{G} \in M_{\mathrm{swap}}(\bm{\theta})$, all the gadgets transform the backward-propagated operator $IP$ into $PI$. Here, the input $\bm{\theta}\in \AngleSet^m$ determines the Pauli operators that are backward propagated to the gadget layer.
When the backward-propagated operator is nontrivial (i.e., $P \neq I$) on the $i$-th qubit, which occurs on at most $K$ qubits, we choose the angle combination of $\bm{\theta}_{\mathcal{G}_i}$ according to the first case in \cref{app:impact_gadget}, which provides a construction of 16 configurations of $\bm{\theta}_{\mathcal{G}_i}\in \AngleSet^3$. On the otherhand, when the backward-propagated operator is $I$, which holds for at least $n - K$ qubits, any angle combination $\bm{\theta}_{\mathcal{G}} \in \AngleSet^3$ satisfies the required condition. It implies that for arbitrary $\bm{\theta} \in \AngleSet^m$,
\begin{equation}
\left|M_{\mathrm{swap}}(\bm{\theta})\right|\geq 4^{3(n-K)}16^{K} = 4^{3n}\left(\frac{1}{4}\right)^K.  
\end{equation}
The effect of choosing $\bm{\theta}_\mathcal{G} \in M_{\mathrm{swap}}(\bm{\theta})$ on the Pauli path is illustrated in \cref{fig:Pauli_path_MPQC}.
\begin{figure}[H]
    \centering
      \includegraphics[width=0.5\textwidth]{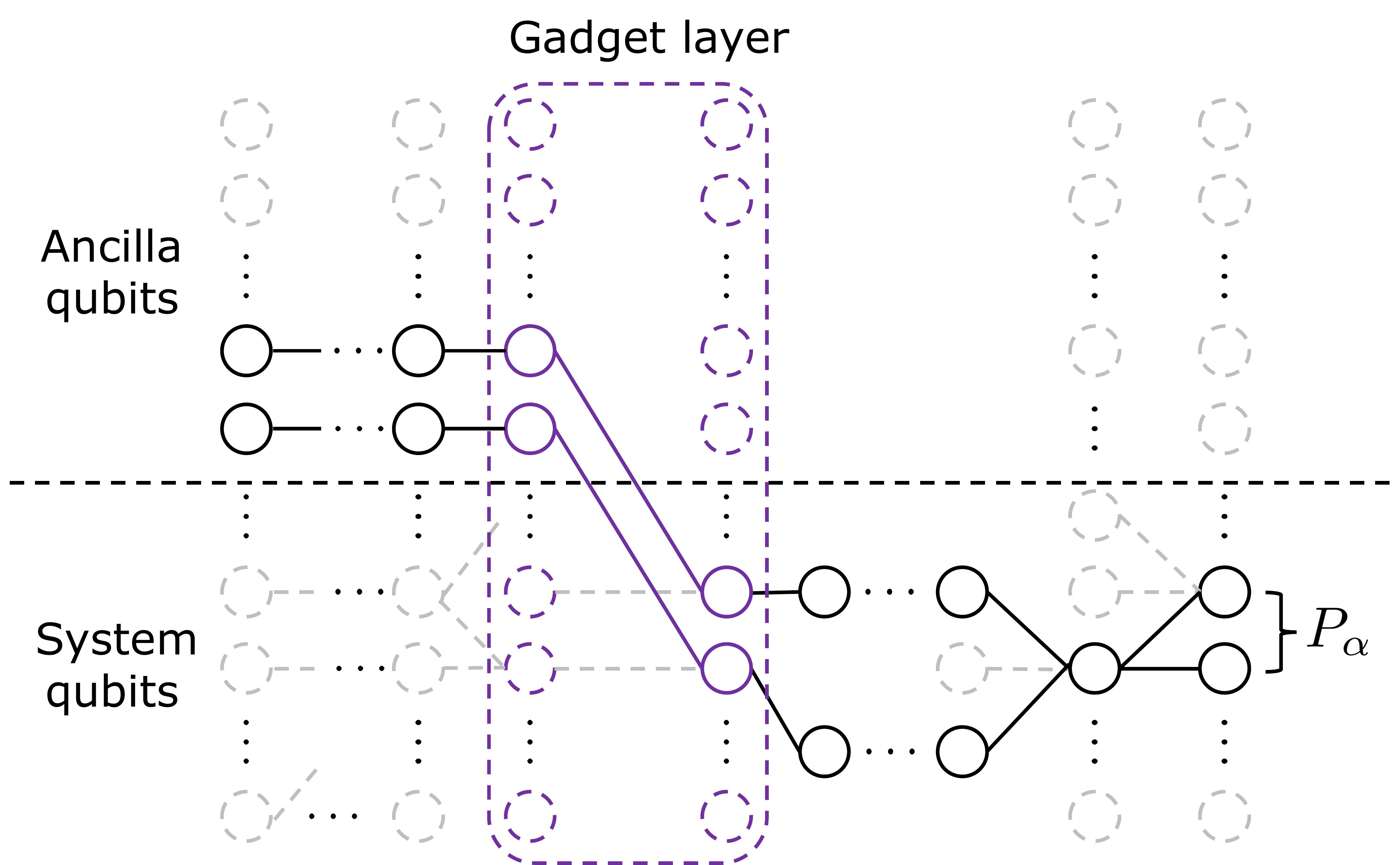}
\caption{Pauli path of MPQC propagated from the observable $P_\alpha$. Each column corresponds to a Pauli operator in the Pauli path, and each circle in the column represents a Pauli operator acting on one specfic qubit. Solid circles denote nontrivial Pauli operators (i.e., not equal to $I$), while dashed circles indicate identity operators. Lines between Pauli operators at adjacent layers represent quantum gates acting on the corresponding qubits. All gates within the gadget layer are grouped into a single layer, as indicated by the purple dashed box. Purple circles represent Pauli operators immediately before and after the gadget layer in the backward propagation. In this example, we choose $\bm{\theta} \in \AngleSet^m$ and $\bm{\theta}_\mathcal{G} \in M_{\mathrm{swap}}(\bm{\theta})$, so that the backward-propagated Pauli path is uniquely determined. The configuration $\bm{\theta}_\mathcal{G} \in M_{\mathrm{swap}}(\bm{\theta})$ ensures that nontrivial Pauli operators originally acting on system qubits are swapped to the corresponding ancilla qubits.}\label{fig:Pauli_path_MPQC}
\end{figure} 
Since the Pauli operators remaining on the system qubits before the gadget layer are all identities when we choose $\bm{\theta}_\mathcal{G} \in M_{\mathrm{swap}}(\bm{\theta})$, the variance of the loss function is lower bounded by
\begin{equation}
  \begin{aligned}
    \operatorname{Var}_{(\bm{\theta},\bm{\theta}_\mathcal{G})}\left[\lossM\right] & \geq \frac{1}{4^{m+3n}} 
      \sum_{\substack{\bm{\theta} \in \AngleSet^m \\
      \bm{\theta}_\mathcal{G} \in M_{\mathrm{swap}}(\bm{\theta})}} 
      \sum_{\alpha} c_\alpha^2 f\left(\uniquePathM,\left(\bm{\theta},\bm{\theta}_\mathcal{G}\right),I\otimes P_\alpha,op\left(\ket{0}\bra{0}\right)^{\otimes n}
\otimes\rho\right)^2\\
    & \geq \frac{1}{4^{m+3n}} \underset{\substack{\bm{\theta} \in \AngleSet^m\\
    \bm{\theta}_\mathcal{G} \in M_{\mathrm{swap}}(\bm{\theta})}}{\sum}\sum_{\alpha} c_\alpha^2 \tr{ \left(I\otimes P_\alpha\right)^2/2^n}^2 \tr{\mathbf{s}_L^{\hspace{0.1em}\left((\bm{\theta},\bm{\theta}_\mathcal{G}),\alpha\right)} op \left(\ket{0}\bra{0}\right)^{\otimes n}\otimes \rho}^2\\
& = \frac{1}{4^{m+3n}} \underset{\substack{\bm{\theta} \in \AngleSet^m\\
    \bm{\theta}_\mathcal{G} \in M_{\mathrm{swap}}(\bm{\theta})}}{\sum}\sum_{\alpha} c_\alpha^2 \tr{\mathbf{s}_L^{\hspace{0.1em}\left((\bm{\theta},\bm{\theta}_\mathcal{G}),\alpha\right)}|_{\leq n} op \left(\ket{0}\bra{0}\right)^{\otimes n}}^2\tr{I\rho}^2\\
   & \geq \frac{1}{4^{m+3n}} \underset{\substack{\bm{\theta} \in \AngleSet^m}}{\sum} \left|M_{\mathrm{swap}}(\bm{\theta})\right|\sum_{\alpha} c_\alpha^2 \tau^{K}\\
    &\geq \frac{1}{4^{m+3n}} 4^{m} 4^{3n}\left(\frac{1}{4}\right)^K\sum_{\alpha} c_\alpha^2 \tau^{K}\\
    & =\left(\frac{\tau}{4}\right)^{K}\sum_{\alpha} c_\alpha^2 = \left(\frac{\tau}{4}\right)^{K}\left\|O\right\|_{HS}^2=\Omega\left(\frac{1}{\mathrm{poly}(n)}\right).
  \end{aligned}
\end{equation} 
Here, the first equality holds because $\mathbf{U}_i(\theta_i)$ for $i \leq L$ has no effect on the Pauli path, as all Pauli operators acting on the system qubits are identities. The notation $\mathbf{s}_L^{\hspace{0.1em}\left((\bm{\theta},\bm{\theta}_\mathcal{G}),\alpha\right)}|_{\leq n}$ denotes the Pauli operator supported on the first $n$ qubits of $\mathbf{s}_L^{\hspace{0.1em}\left((\bm{\theta},\bm{\theta}_\mathcal{G}),\alpha\right)}$.
The third inequality holds since $\mathbf{s}_L^{((\bm{\theta},\bm{\theta}_\mathcal{G}),\alpha)}|_{\le n}$ contains at most $K$ nontrivial single-qubit Pauli operators, each contributing at least a factor $\tau$.


\end{proof}

\section{Proof of \cref{thm:absence_BP_MPQC_parameters}}\label{app:proof_variance_scaling_parameters}

In this section, we prove that introducing a gadget layer consistently improves the trainability of a PQC, by establishing a lower bound on the gradient of the variance $\operatorname{Var}_{\bm{\theta}} \left[ \frac{\partial \loss}{\partial \theta_j} \right]$ for every $\theta_j \in \bm{\theta}$. 

\subsection{Feedforward parameters number of PQCs}
Recall that for each $\theta_j$, we prove that for PQC satisfying the conditions of \cref{lem:two_variances_MPQC}, $\operatorname{Var}_{\bm{\theta}} \left[ \frac{\partial \loss}{\partial \theta_j}\right]$ can be express as 
\begin{equation}
    \operatorname{Var}_{\bm{\theta}} \left[ \frac{\partial \loss}{\partial \theta_j}\right]
    =\frac{1}{4^{m}} 
    \sum_{\substack{
        \bm{\theta} \in \AngleSet^m \\
        \{P_j, s_j^{(\bm{\theta},\alpha)}\} = 0
      }} \sum_\alpha c_\alpha^2 \, f(\bm{\theta},P_\alpha,\vec{s}^{\hspace{0.1em}(\bm{\theta},\alpha)},\rho)^2.
\end{equation}
This indicates that, in order to analyze the scaling of this quantity, we need to characterize the number of angle combinations $\bm{\theta} \in \AngleSet^m$ satisfying $
\{P_j, s_j^{(\bm{\theta},\alpha)}\} = 0$. To this end, we introduce the concept of the \emph{feedforward parameter number} $\Fnum $ associated with the parameter $\theta_j$ and the observable $O$:
\begin{definition}[feedforward parameters number]
For a PQC $\mathcal{C}(\bm{\theta})$ and an observable $O = \sum_{\alpha} c_\alpha P_\alpha$, we consider its parameter $\theta_j$ appearing in a rotation gate $R_{P_j}(\theta_j)$.
Denote by $\{J_\alpha\}_\alpha$ the collection of backward light cones of $\{P_\alpha\}$ that include $R_{P_j}(\theta_j)$, and let $\{\bar{J}_\alpha\}$ represent the portions of these cones that appear after the layer containing $R_{P_j}(\theta_j)$, as illustrated in \cref{fig:feedforward_parameters_number}. For each region $\bar{J}_\alpha$, count the number of rotation gates that contain parameters, resulting in a set $\{\#_R\bar{J}_\alpha\}$. The quantity $\Fnum $ is defined as the maximum value in this set. More precisely,

\begin{equation}
\Fnum  =
\begin{cases*}
0, \text{if } \{\bar{J}_\alpha\} = \emptyset \\
\max_{\alpha}\{\#_R\bar{J}_\alpha\},  \text{otherwise}
\end{cases*}
\end{equation}
\end{definition}

\begin{figure}[H]
    \centering
      \includegraphics[width=0.5\textwidth]{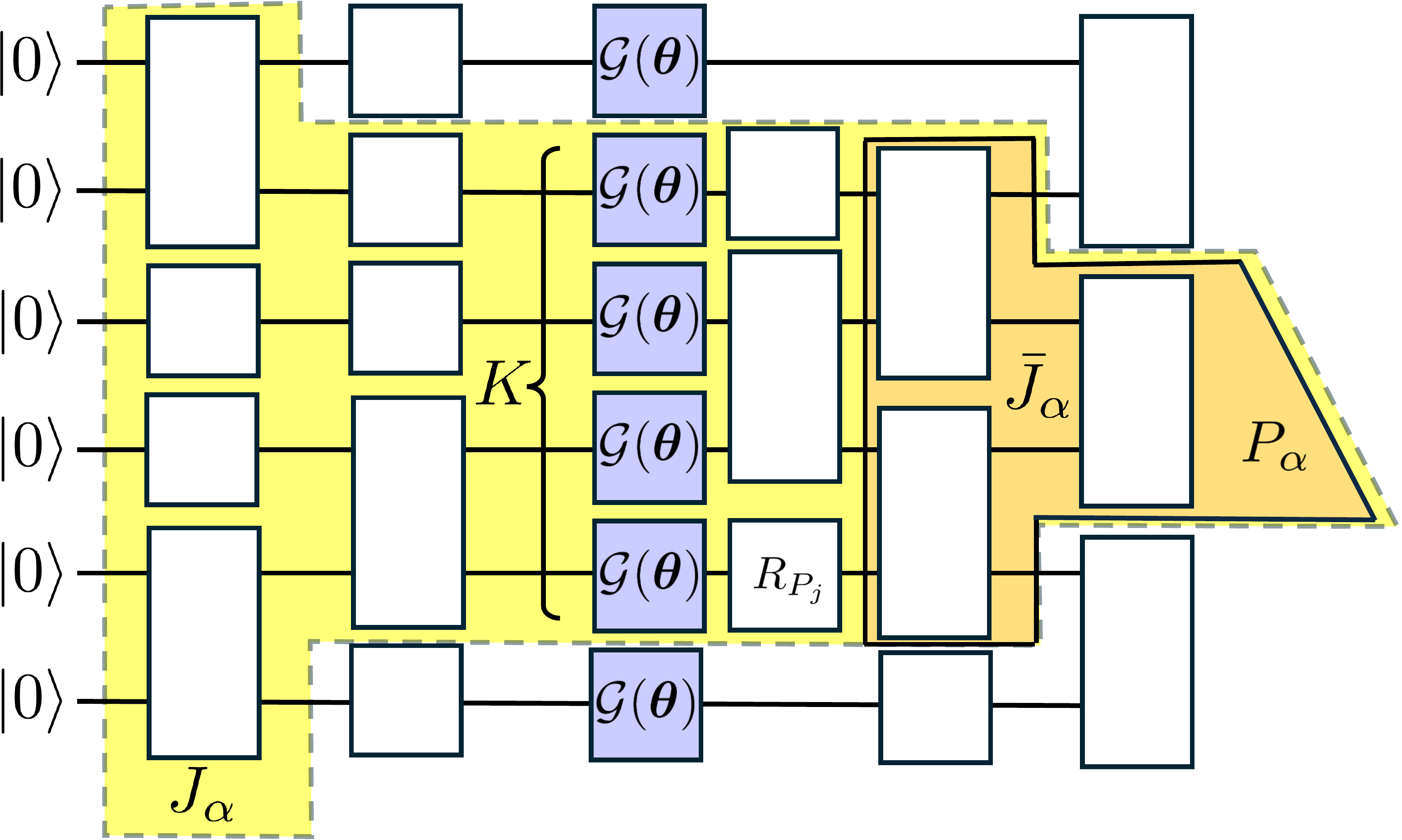}
\caption{Illustration of the feedforward parameters number $\Fnum $. The yellow region represents the backward light cone $J_\alpha$ of a specific Pauli word $P_\alpha$ in the observable $O$. The orange region $\bar{J}_\alpha$ denotes the portion of $J_\alpha$ that appears after the layer containing the rotation gate $R_{P_j}(\theta_j)$. The quantity $\#_R\bar{J}_\alpha$ counts the number of rotation gates with parameters within the region $\bar{J}_\alpha$.
The support size of the backward light cone at layer $l$ (i.e., the layer where the gadget is inserted in MPQCs) is upper bounded by $K$.
}\label{fig:feedforward_parameters_number}
\end{figure}

It is straightforward to observe that in the Heisenberg picture, examining the rotation angles in $\{\bar{J}_\alpha\}$ suffices to determine whether $\{P_j, s_j^{(\bm{\theta},\alpha)}\} = 0$. This implies that, at most $\Fnum$ parameters in $\bm{\theta}$ need to be considered. Then, to characterize the total number of parameters that need to be considered in each Pauli path after the gadget layer, we introduce the following definition.

\begin{definition}[Total number of feedforward parameters after the gadget layer]
For a PQC $\mathcal{C}(\bm{\theta})$ and an observable $O = \sum_{\alpha} c_\alpha P_\alpha$, 
consider its corresponding MPQC $\channelM{\bm{\theta},\bm{\theta}_\mathcal{G}}$ obtained by inserting a gadget layer.
We define the \emph{total number of feedforward parameters after the gadget layer}, denoted by $\Fnumall$, 
as the maximum number of parameters contained in the backward light cones of all Pauli terms $P_\alpha$ that located after the gadget layer.
\end{definition}

It is straightforward to verify that for any parameter $\theta_j$ lie after the gadget layer, we have $\Fnum \le \Fnumall$.

\subsection{Lower bound of gradient variance of the loss function of MPQCs}\label{subapp:lower_bound_gradient}
We are now ready to present the following theorem, which provides the formal version of \cref{thm:absence_BP_MPQC_parameters}:
\begin{theorem}[\cref{thm:absence_BP_MPQC_parameters}, formal version]\label{thm:absence_BP_MPQC_parameters_app}
  Consider an MPQC $\channelM{\bm{\theta},\bm{\theta}_\mathcal{G}}$ and a $k$-local observable $O = \sum_{\alpha} c_\alpha P_\alpha$. Suppose the support size of the backward light cone of each $P_\alpha$ at the gadget layer is upper bounded by $K = \order{\log n}$ and $\Fnumall =\order{\log n}$.
Then, the variance of the gradient with respect to the parameters $\bm{\theta} \in [0, 2\pi)^m$ in the original PQC satisfies the following properties:
\begin{itemize}
\item  For parameter $\theta_j$ located after the gadget layer, if $ \operatorname{Var}_{\bm{\theta}}\left[ \frac{\partial \loss}{\partial{\theta_j}}\right]\neq 0$, then $\operatorname{Var}_{(\bm{\theta},\bm{\theta}_\mathcal{G})}\left[ \frac{\partial \lossM}{\partial{\theta_j}}\right]$ is lower bounded by
\begin{equation}\label{eq:lower_bound_gradient_variance_MPQC}
  \operatorname{Var}_{(\bm{\theta},\bm{\theta}_\mathcal{G})}\left[ \frac{\partial \lossM}{\partial{\theta_j}}\right]\geq \left(\frac{1}{2}\right)^{\Fnum }\left(\frac{\tau}{4}\right)^{K}\left\|O\right\|_{min}^2= \Omega\left(\frac{1}{\mathrm{poly}(n)}\right),
\end{equation}
where $\left\|O\right\|_{\min}:=\min\{\left|c_\alpha\right|>0\}$.

\item For parameter $\theta_j$ located before the gadget layer, $\theta_j$ remains trainable if it is already trainable in the original PQC, which is ensured by the following lower bound on the gradient variance $\operatorname{Var}_{(\bm{\theta},\bm{\theta}_\mathcal{G})}\left[ \frac{\partial \lossM}{\partial{\theta_j}}\right]$:
\begin{equation}\label{eq:lower_bound_gradient_variance_MPQC_same_order}
  \operatorname{Var}_{(\bm{\theta},\bm{\theta}_\mathcal{G})}\left[ \frac{\partial \lossM}{\partial{\theta_j}}\right]\geq\left(\frac{1}{4}\right)^{K} \left(\frac{\left\|O\right\|_{\min}}{\left\|O\right\|_{HS}} \right)^2\operatorname{Var}_{\bm{\theta}}\left[ \frac{\partial \loss}{\partial{\theta_j}}\right] = \Omega\left(\frac{1}{\mathrm{poly}(n)}\right) \operatorname{Var}_{\bm{\theta}}\left[ \frac{\partial \loss}{\partial{\theta_j}}\right].
\end{equation}
  
\end{itemize}
\end{theorem}
\begin{proof}
We first suppose that the gate $R_{P_j}(\theta_j)$ contains parameter $\theta_j$ is located after the gadget layer. According to \cref{eq:gradient_variance_MPQC}, the variance of its gradient
$\operatorname{Var}_{(\bm{\theta},\bm{\theta}_\mathcal{G})}\left[ \frac{\partial \lossM}{\partial \theta_j}\right]$ can be expressed as
\begin{equation}\label{eq:variance_paras_for_proof}
\begin{aligned}
\operatorname{Var}_{(\bm{\theta},\bm{\theta}_\mathcal{G})}\left[ \frac{\partial \lossM}{\partial \theta_j}\right] = \frac{1}{4^{m+3n}} 
    \sum_{\substack{\bm{\theta} \in \AngleSet^m \\
    \bm{\theta}_\mathcal{G} \in \AngleSet^{3n}\\
        \{\mathbf{P}_j, \mathbf{s}_j^{((\bm{\theta},\bm{\theta}_\mathcal{G}),\alpha)}\} = 0}} 
    \sum_{\alpha} c_\alpha^2 f\left(\uniquePathM,\left(\bm{\theta},\bm{\theta}_\mathcal{G}\right),I\otimes P_\alpha,op\left(\ket{0}\bra{0}\right)^{\otimes n}
\otimes\rho\right)^2.\\
\end{aligned}  
\end{equation}
If $\operatorname{Var}_{\bm{\theta}}\left[ \frac{\partial \loss}{\partial{\theta_j}}\right]$ is nonzero, 
then there at least exist one Pauli word $P_\beta$ in $O$ and $\bm{\theta} \in \AngleSet^m$ such that $\{P_j, \vec{s}^{\hspace{0.1em}(\bm{\theta},\beta)}_j\} = 0$. Since the gadget layer does not affect the Pauli path after $l$-th layer (in the Heisenberg picture), we have that in the MPQC setting, $\mathbf{P}_j = I \otimes P_j$ and $\mathbf{s}_j^{((\bm{\theta},\bm{\theta}_\mathcal{G}),\beta)} = I \otimes \vec{s}^{\hspace{0.1em}(\bm{\theta},\beta)}_j$ for arbitrary $\bm{\theta}_\mathcal{G}$. It also implies that
\[
\{\mathbf{P}_j, \mathbf{s}_j^{((\bm{\theta},\bm{\theta}_\mathcal{G}),\beta)}\} = \{I \otimes P_j, I \otimes \vec{s}^{\hspace{0.1em}(\bm{\theta},\beta)}_j\} = \{P_j, \vec{s}^{\hspace{0.1em}(\bm{\theta},\beta)}_j\} = 0.
\]
This implies that if there exists a parameter configuration $\bm{\theta} \in \AngleSet^m$ and a Pauli word $P_\beta$ in $O$ such that $\{P_j, \vec{s}^{\hspace{0.1em}(\bm{\theta},\beta)}_j\} = 0$ holds in the original PQC, then one can construct a group of angle combinations $(\bm{\theta},\bm{\theta}_\mathcal{G})$ such that $\{\mathbf{P}_j, \mathbf{s}_j^{((\bm{\theta},\bm{\theta}_\mathcal{G}),\beta)}\}=0$ holds in the corresponding MPQC.

Employing this property, we now count the number of discrete angle configurations in \cref{eq:variance_paras_for_proof} for which the corresponding term does not vanish. Let $M_j$ denote the set of angle configurations of $\bm{\theta}\in \AngleSet^m$ that maximize the number of angle configurations satisfying $\{P_j, \vec{s}^{\hspace{0.1em}(\bm{\theta},\beta)}_j\} = 0$. 
Since $\operatorname{Var}_{\bm{\theta}}\left[ \frac{\partial \loss}{\partial \theta_j} \right] \neq 0$, there exists at least one Pauli word $P_\beta$ in $O$ and one angle configuration $\bm{\theta} \in \AngleSet^m$ such that $\{P_j, \vec{s}^{\hspace{0.1em}(\bm{\theta},\beta)}_j\} = 0$. 
On the other hand, the angle values $\{0, \pi\}$ and $\{\pi/2, 3\pi/2\}$ yield the same effect on the backward propagation of the Pauli path, up to an overall sign.
Therefore, for parameters located in the region $\bar{J}_\beta$, they can be replaced by their corresponding pairs without affecting the commutation relation between $\vec{s}^{\hspace{0.1em}(\bm{\theta},\beta)}_j$ and $P_j$. Consequently, there exist at least $2^{\#_R\bar{J}_\beta}$ angle configurations such that $\{P_j, \vec{s}^{\hspace{0.1em}(\bm{\theta},\beta)}_j\} = 0$ holds. 

While for parameters outside $\bar{J}_\beta$, their values do not affect the commutation relation between $\vec{s}^{\hspace{0.1em}(\bm{\theta},\beta)}_j$ and $P_j$, and hence can be chosen arbitrarily, yielding $4^{m - \#_R\bar{J}_\beta}$ possible angle configurations.

Therefore, we obtain:
\begin{equation}
  |M_j| \geq 2^{\#_R\bar{J}_\beta}4^{m-\#_R\bar{J}_\beta} = 4^m \left(\frac{1}{2}\right)^{\#_R\bar{J}_\beta}.
\end{equation}

Next, we fix the choice of $\bm{\theta}_\mathcal{G}$. We pick $\bm{\theta}_\mathcal{G} \in \AngleSet^{3n}$ the same as the construction in the proof of \cref{thm:absence_BP_MPQC} (corresponding to the first case in \cref{app:impact_gadget}). It swaps the non-trivial Pauli operator in the system qubits to the ancillas. We also denote the set of such configurations of $\bm{\theta}_\mathcal{G}$ as $M_{\mathrm{swap}}(\bm{\theta})$. Since the support size of the backward-propagated Pauli operator at the gadget layer is upper bounded by $K$, following the same counting argument as in the proof of \cref{thm:absence_BP_MPQC}, we obtain that for any $\bm{\theta}\in \AngleSet^m$,
\begin{equation}
\left|M_{\mathrm{swap}}(\bm{\theta})\right| \geq 4^{3n}\left(\frac{1}{4}\right)^K.
\end{equation}

Based on the above constructions of $\bm{\theta}$ and $\bm{\theta}_\mathcal{G}$, we obtain the following lower bound:
\begin{equation}
\begin{aligned}
\operatorname{Var}_{(\bm{\theta},\bm{\theta}_\mathcal{G})}\left[ \frac{\partial \lossM}{\partial \theta_j}\right] =& \frac{1}{4^{m+3n}} 
    \sum_{\substack{\bm{\theta} \in \AngleSet^m \\
    \bm{\theta}_\mathcal{G} \in \AngleSet^{3n}\\
        \{\mathbf{P}_j, \mathbf{s}_j^{((\bm{\theta},\bm{\theta}_\mathcal{G}),\alpha)}\} = 0}} 
    \sum_{\alpha} c_\alpha^2 f\left(\uniquePathM,\left(\bm{\theta},\bm{\theta}_\mathcal{G}\right),I\otimes P_\alpha,op\left(\ket{0}\bra{0}\right)^{\otimes n}
\otimes\rho\right)^2 \\
\geq& \frac{1}{4^{m+3n}} 
    \sum_{\substack{\bm{\theta} \in \AngleSet^m \\
    \bm{\theta}_\mathcal{G} \in \AngleSet^{3n}\\
        \{\mathbf{P}_j, \mathbf{s}_j^{((\bm{\theta},\bm{\theta}_\mathcal{G}),\beta)}\} = 0}} 
     c_\beta^2 f\left(\uniquePathMb,\left(\bm{\theta},\bm{\theta}_\mathcal{G}\right),I\otimes P_\beta,op\left(\ket{0}\bra{0}\right)^{\otimes n}
\otimes\rho\right)^2 \\
\geq& \frac{1}{4^{m+3n}} 
    \sum_{\substack{\bm{\theta} \in M_j \\
    \bm{\theta}_\mathcal{G} \in M_{\mathrm{swap}}(\bm{\theta})} }
     c_\beta^2 f\left(\uniquePathMb,\left(\bm{\theta},\bm{\theta}_\mathcal{G}\right),I\otimes P_\beta,op\left(\ket{0}\bra{0}\right)^{\otimes n}
\otimes\rho\right)^2\\
\geq& \frac{1}{4^{m+3n}} 
    \sum_{\substack{\bm{\theta} \in M_j} } \left|M_{\mathrm{swap}}(\bm{\theta})\right|
     c_\beta^2 
\tr{\mathbf{s}_L|_{\leq n} op \left(\ket{0}\bra{0}\right)^{\otimes n}}^2\tr{I\rho}^2\\
\geq& \frac{|M_j|}{4^{m+3n}} 4^{3n}\left(\frac{1}{4}\right)^K
     c_\beta^2 \tau^K\\
\geq&  \frac{c^2_\beta}{4^{m+3n}}4^m 4^{3n}\left(\frac{1}{2}\right)^{\#_R\bar{J}_\beta}\left(\frac{1}{4}\right)^K \tau^{K} \\
=&  \left(\frac{1}{2}\right)^{\#_R\bar{J}_\beta}\left(\frac{\tau}{4}\right)^{K}c^2_\beta\\
\geq& \left(\frac{1}{2}\right)^{\Fnum}\left(\frac{\tau}{4}\right)^{K}\left\|O\right\|_{min}^2= \Omega\left(
    \frac{1}{\mathrm{poly}(n)}\right),
\end{aligned}  
\end{equation}
where the last equation holds because $\Fnum \leq \Fnumall = \order{\log n}$.
This completes the proof of \cref{eq:lower_bound_gradient_variance_MPQC}.

If the parameter $\theta_j$ is located before the gadget layer, we again consider a specific construction of $\bm{\theta}_\mathcal{G} \in \AngleSet^{3n}$. In particular, we choose $\bm{\theta}_\mathcal{G}$ such that it does not affect the backward propagation of the Pauli path; we denote this the angle configuration as $M_{\mathrm{same}}(\bm{\theta})$.

Note that for arbitrary $\bm{\theta} \in \AngleSet^m$ and Pauli word $P_\alpha$ in $O$, at most $K$ nontrivial Pauli operators are propagated backward to the gadget layer. We then select $\bm{\theta}_\mathcal{G}$ corresponding to the second case described in \cref{app:impact_gadget} that does not change the Pauli operators on these qubits.
There is at least $4^{3(n - K)} 16^K$ distinct angle configurations of $\bm{\theta}_\mathcal{G}$ that satisfy this requirement. This implies that for any $\bm{\theta} \in \AngleSet^m$ and $P_\alpha$, we have
\[
 \left|M_{\mathrm{same}}(\bm{\theta})\right| \geq 4^{3(n - K)} 16^K.
\]

By restricting our attention to these configurations, we obtain the following lower bound on the gradient variance with respect to $\theta_j$:
\begin{equation}
\begin{aligned}
&\operatorname{Var}_{(\bm{\theta},\bm{\theta}_\mathcal{G})}\left[ \frac{\partial \lossM}{\partial \theta_j}\right] \\
\geq& \frac{1}{4^{m+3n}} 
    \sum_{\substack{\bm{\theta} \in \AngleSet^m \\
    \bm{\theta}_\mathcal{G} \in M_{\mathrm{same}}(\bm{\theta})\\
        \{\mathbf{P}_j, \mathbf{s}_j^{(\bm{\theta},\alpha)}\} = 0}} 
    \sum_{\alpha} c_\alpha^2 f\left(\uniquePathM,\left(\bm{\theta},\bm{\theta}_\mathcal{G}\right),I\otimes P_\alpha,op\left(\ket{0}\bra{0}\right)^{\otimes n}
\otimes\rho\right)^2\\
= &  \frac{1}{4^{m+3n}} \sum_{\substack{\bm{\theta} \in \AngleSet^m}} \left|M_{\mathrm{same}}(\bm{\theta})\right|
    \sum_{\alpha} c_\alpha^2 \sum_{\substack{\mathbf{s}_m,\mathbf{s}_{m-1},\cdots,\mathbf{s}_0\\
        \{\mathbf{P}_j, \mathbf{s}_j\} = 0
  }} \tr{I\otimes P_\alpha \mathbf{s}_m}^2 
  \prod_{i=i}^{m}\tr{\mathbf{s}_{i} \mathbf{U}_{i}(\theta_i) \mathbf{s}_{i-1} \mathbf{U}_i(\theta_i)^\dagger}^2 \tr{\mathbf{s}_{0} op\left(\ket{0}\bra{0}\right)^{\otimes n}
\otimes\rho}^2\\ 
\geq &  \frac{4^{3(n-K)}16^K}{4^{m+3n}} \sum_{\alpha} \underset{\substack{\bm{\theta} \in \AngleSet^m}}{\sum} \sum_{\substack{\mathbf{s}_m,\mathbf{s}_{m-1},\cdots,\mathbf{s}_0\\
        \{\mathbf{P}_j, \mathbf{s}_j\} = 0
  }} c_\alpha^2 \tr{I\otimes P_\alpha \mathbf{s}_m}^2 
  \prod_{i=i}^{m}\tr{\mathbf{s}_{i} \mathbf{U}_{i}(\theta_i) \mathbf{s}_{i-1} \mathbf{U}_i(\theta_i)^\dagger}^2 \tr{\mathbf{s}_{0} op\left(\ket{0}\bra{0}\right)^{\otimes n}
\otimes\rho}^2\\
= &\left(\frac{1}{4}\right)^{K} \left(\frac{1}{4}\right)^{m} \underset{\substack{\bm{\theta} \in \AngleSet^m}}{\sum} \sum_{\vec{{s}}:\{{P}_j, {s}_j\} = 0} \sum_{\alpha} c_\alpha^2 \tr{P_\alpha {s}_m}^2 \prod_{i=i}^{m}\tr{{s}_{i} {U}_{i}(\theta_i) {s}_{i-1} {U}_i(\theta_i)^\dagger}^2\tr{{s}_0 \ket{0^n}\bra{0^n}}^2\\
\geq& \left(\frac{1}{4}\right)^{K} \left(\frac{\left\|O\right\|_{\min}}{\left\|O\right\|_{HS}} \right)^2\operatorname{Var}_{\bm{\theta}}\left[ \frac{\partial \loss}{\partial{\theta_j}}\right] = \Omega\left(\frac{1}{\mathrm{poly}(n)}\right) \operatorname{Var}_{\bm{\theta}}\left[ \frac{\partial \loss}{\partial{\theta_j}}\right].
\end{aligned}  
\end{equation}
Here the first equality holds due to the choice $\bm{\theta}_\mathcal{G} \in M_{\mathrm{same}}(\bm{\theta})$, under which all the Pauli paths in the gadget layer remain unchanged and equal to $\mathbf{s}_L$, i.e., $\mathbf{s}_{L+1}=\mathbf{s}_{\mathcal{G}_{1,1}}= \mathbf{s}_{\mathcal{G}_{1,2}}= \cdots =\mathbf{s}_{\mathcal{G}_{n,3}} = \mathbf{s}_L$. The last inequality employs the conclusion in \cref{cor:gradient_variance_upper_bound}.
\end{proof}

\begin{remark}
As we can see, the proof of this theorem is rather loose, as three cases in \cref{app:impact_gadget} were entirely omitted. We believe that introducing the gadget layer enriches the diversity of Pauli paths contributing to the gradient, which can substantially increase the overall gradient variance.
\end{remark}
\section{Locating the Gadget Layer via Circuit Geometry}\label{app:K_f_D-l_geometry}

In this section, we demonstrate how to determine the placement of the gadget layer based on the geometric structure of the circuit.
Our goal is to determine the appropriate position of the gadget layer—specifically, the value of $D - l$ (where $D$ denotes the depth of the original PQC)—such that both $K$ and $\Fnumall$ are of order $\order{\log n}$, thereby fulfilling the assumptions required by the theorem.

As an example, we consider a class of PQCs defined on (hyper)cubic lattices. These circuits are composed of two-qubit gates, or blocks of gates that effectively act on two qubits, applied along the edges of a lattice such that each qubit participates in exactly one two-qubit gate (or gate block) per layer. We consider circuits embedded in a $d$-dimensional (hyper)cubic lattice with $d \geq 1$. For simplicity, in the following discussion we assume that each two-qubit gate block consists of a single two-qubit gate. This simplification only affects constant prefactors in the scaling of gate-related quantities and does not alter the asymptotic analysis.

In such uniform architectures, it is natural to characterize the size of backward light cones using the concept of \emph{operator spreading velocity} $v \in [0,1]$~\cite{PhysRevX.8.021014}. According to the analysis in Ref.~\cite{aharonov2025importance}, for a 1-local observable, the number of qubits involved after $D-l$ layers in the Heisenberg picture is given by
\[
n_{D-l} = \left( \frac{2v}{d} \left(D-l\right) \right)^d.
\]
Therefore, for arbitrary $k$-local Pauli word $P_\alpha$ in $O$, the number of qubits influenced after $D-l$ layers  is at most 
\[
k n_{D-l} = k \left( \frac{2v}{d} \left(D-l\right) \right)^d,
\]
which is a upper bound of $K$.

For this type of circuit, the feedforward parameter number can also be tightly upper bounded. Based on the calculation in Ref.~\cite{aharonov2025importance}, the total number of gates involved in the backward light cone of $k$-local observable after $D-l$ layers is upper bounded by
\[
\frac{k n_{D-l} \left(D-l\right)}{2(d+1)}=\frac{k2^{d-1}v^d}{(d+1)d^d}\left(D-l\right)^{d+1},
\]
which serves as an upper bound on $\Fnumall$.

To conclude, for PQCs defined on a $d$-dimensional cubic lattice and measured with a $k$-local observable, if the corresponding MPQC is constructed by inserting a gadget layer after the $l$-th layer of the original circuit, then the following result holds:
\begin{equation}\label{eq:upper_bound_K}
  K \leq k n_{D-l}=k \left( \frac{2v}{d} \left(D-l\right) \right)^d
\end{equation}
\begin{equation}\label{eq:upper_bound_Fnum}
  \Fnumall \leq \frac{k n_{D-l} \left(D-l\right)}{2(d+1)}=\frac{k2^{d-1}v^d}{(d+1)d^d}\left(D-l\right)^{d+1}.
\end{equation}
Then, by restricting $D - l = \order{\left(\log n\right)^{\frac{1}{d+1}}}$ and treating $v$ and $d$ as constants, we can apply \cref{eq:upper_bound_K} and \cref{eq:upper_bound_Fnum} to obtain the following results:
\begin{equation}
  K = \order{\left(\log n\right)^{\frac{d}{d+1}}}, \Fnumall = \order{\log n}.
\end{equation}
This implies that the conditions in \cref{thm:absence_BP_MPQC_parameters_app} are naturally satisfied when $D - l = \order{\left(\log n\right)^{\frac{1}{d+1}}}$. We thus obtain the following corollary:
\begin{corollary}
Let $\mathcal{C}(\bm{\theta})$ be a PQC defined on a $d$-dimensional (hyper)cubic lattice, and let its circuit depth be denoted by $D$. Suppose the corresponding MPQC $\channelM{\bm{\theta}, \bm{\theta}_\mathcal{G}}$ is constructed by inserting a layer of gadgets after the $l$-th layer of $\mathcal{C}(\bm{\theta})$. Then, the lower bounds on the gradient variance $\operatorname{Var}_{(\bm{\theta},\bm{\theta}_\mathcal{G})}\left[ \frac{\partial \lossM}{\partial \theta_j} \right]$ established in \cref{thm:absence_BP_MPQC_parameters_app} hold, provided that $D - l = \order{\left(\log n\right)^{\frac{1}{d+1}}}$.
\end{corollary}


\section{Strategy for activating single parameter}\label{app:activate_single_parameter}
In this section, we provide additional details on the activation of a single parameter, including the construction of the enlarged gadget and the proof of \cref{thm:activate_parameters}.

\subsection{Selection of the enlarged gadget}
Suppose we aim to activate a single-qubit rotation gate $T = R_{P_T}(\theta_T)$,
which acts nontrivially on the $t$-th system qubit and is located before the gadget layer. To achieve this, we insert one extra gadget immediately before $T$ and enlarge one gadget $\mathcal{G}(\bm{\theta})$ in the gadget layer to obtain an new type of gadget $\gadgetT{\bm{\theta}}$, in which three additional two-qubit rotation gates are inserted. The only restriction we impose on the enlarged gadget is that if we choose the $i$-th gadget $\mathcal{G}_i(\bm{\theta})$ in the gadget layer, there must exist some $P_\beta$ in $O$ and $\bm{\theta}\in \AngleSet^m$ such that $s^{\left(\bm{\theta}, \beta\right)}_{L}|_i \neq I$. In other words, we require that the $i$-th Pauli word in the operator arriving at the gadget layer, backward propagated from $P_\beta$ for some angle configuration $\bm{\theta}\in \AngleSet^m$, be nontrivial.

Next, we show that such a gadget $\mathcal{G}(\bm{\theta})$ satisfying the above condition can be efficiently identified. We first randomly select a Pauli word $P_\beta$ from $O$, and to determine the Pauli operator that is backward propagated to the gadget layer, it suffices to scan over the angles within the backward light cone of $P_\beta$, i.e., at most $\Fnumall = \order{\log n}$ parameters. We assign these angles random values from $\{0,\pi/2,\pi,3\pi/2\}$ and then compute $s^{\left(\bm{\theta}, \beta\right)}_{L}$. Since the resulting circuit is Clifford, evaluating $s^{\left(\bm{\theta}, \beta\right)}_{L}$ can be done efficiently. We then arbitrarily choose one position where $s^{\left(\bm{\theta}, \beta\right)}_{L}$ acts nontrivially to construct $\gadgetT{\bm{\theta}}$.
Moreover, since each angle can take four possible values, a large number of distinct $s^{\left(\bm{\theta}, \beta\right)}_{L}$ can be generated, implying that almost any gadget within the support of the backward light cone of $P_\beta$ has a high probability of satisfying the required condition.

Without loss of generality and for the convenience of proof, we make the following reasonable assumption: there exists a Pauli word $P_\beta$ in observable $O$ and some $\bm{\theta} \in \AngleSet^m$ such that the backward-propagated Pauli operator $s^{\left(\bm{\theta}, \beta\right)}_{L}$ reaches the gadget layer and satisfies $s^{\left(\bm{\theta}, \beta\right)}_{L}|_t \neq I$, where the subscript $t$ denotes the $t$-th qubit with the target gate $T$ applied.
If this condition is not satisfied, one can instead modify another gadget $\mathcal{G}_i(\bm{\theta})$ into $\gadgetT{\bm{\theta}}$ with a nontrivial input, thereby activating the gate $T$—the only difference being that $\mathcal{G}_i(\bm{\theta})$ and $T$ act on different system qubits, which is depicted in \cref{fig:G_T'_different_sys}. This modification does not affect the validity of the subsequent analysis. 

\begin{figure}[htbp!]
    \centering
      \includegraphics[width=0.5\textwidth]{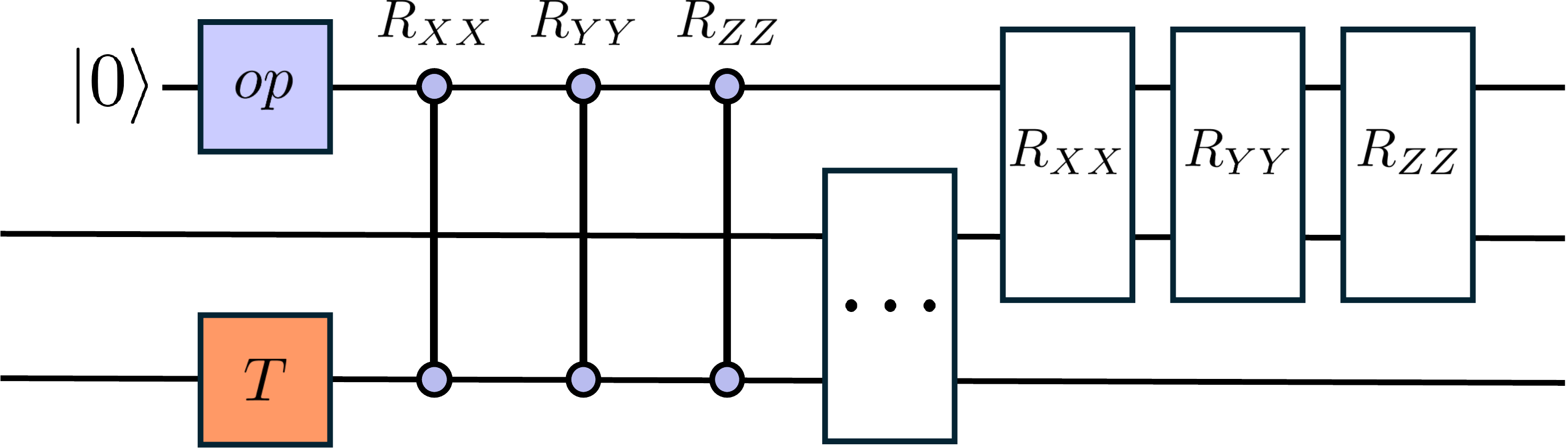}
\caption{Construction of $\gadgetT{\bm{\theta}}$ when $\mathcal{G}_i(\bm{\theta})$ and $T$ act on different system qubits. The first three two-qubit parameterized gates act on the ancilla qubit and the system qubit on which the target gate $T$ is applied.}\label{fig:G_T'_different_sys}
\end{figure} 



\subsection{Proof of \cref{thm:activate_parameters}}
Now we are ready to prove \cref{thm:activate_parameters}. For clarity, we provide a detailed lower bound on the variance of the partial derivative of the loss function with respect to $\theta_T$, following the notation introduced in the manuscript:
\begin{theorem}\label{thm:activate_parameters_app}
  Consider a $T$-activating MPQC $\channelMT{\bm{\theta}, \bm{\theta}_{\mathcal{G}}, \bm{\theta}_{\mathcal{G}'_T}}$ and a $k$-local observable $O = \sum_{\alpha} c_\alpha P_\alpha$. Suppose the conditions in \cref{thm:absence_BP_MPQC_parameters_app} still hold. Then, we have
\begin{equation}\label{eq:activate_parameter_lower_bound}
  \operatorname{Var}_{(\bm{\theta},\bm{\theta}_{\mathcal{G}},\bm{\theta}_{\mathcal{G}'_T})}\left[\frac{\partial \lossMT}{\partial{\theta_T}}\right]\geq \left(\frac{1}{2}\right)^{\Fnumall+8}\left(\frac{\tau}{4}\right)^{K+1}\left\|O\right\|_{min}^2= \Omega\left(\frac{1}{\mathrm{poly}(n)}\right),
\end{equation}
for the loss function of the $T$-activating MPQC, defined as
$\lossMT \coloneq \tr{\channelMT{\bm{\theta},\bm{\theta}_{\mathcal{G}},\bm{\theta}_{\mathcal{G}'_T}}(\rho)O}$.
\end{theorem}
\begin{proof}
We begin by expressing the unitary representation of $\channelMT{\bm{\theta}, \bm{\theta}_{\mathcal{G}}, \bm{\theta}_{\mathcal{G}'_T}}$ when ancilla qubits are included. 
We denote it as $\unitaryMT{\bm{\theta}, \bm{\theta}_{\mathcal{G}}, \bm{\theta}_{\mathcal{G}'_T}}$.
Note that an additional $\gadget{\bm{\theta}}$ is inserted before the gate $T$, so the unitary $\unitaryMT{\bm{\theta}, \bm{\theta}_{\mathcal{G}}, \bm{\theta}_{\mathcal{G}'_T}}$ acts on a $(2n + 1)$-qubit Hilbert space. The loss function of this MPQC can thus be written as
\begin{equation}
  \begin{aligned}
\lossMT &= \tr{\channelMT{\bm{\theta},\bm{\theta}_{\mathcal{G}},\bm{\theta}_{\mathcal{G}'_T}}(\rho)O}\\    
&= \sum_\alpha c_\alpha\tr{\unitaryMT{\bm{\theta},\bm{\theta}_{\mathcal{G}},\bm{\theta}_{\mathcal{G}'_T}}(op(\ket{0}\bra{0})^{\otimes (n+1)}\otimes\rho)\unitaryMT{\bm{\theta},\bm{\theta}_{\mathcal{G}},\bm{\theta}_{\mathcal{G}'_T}}^\dagger I\otimes P_\alpha}.\\    
  \end{aligned}
\end{equation}
 
Then we rewrite $\operatorname{Var}_{(\bm{\theta},\bm{\theta}_{\mathcal{G}},\bm{\theta}_{\mathcal{G}'_T})}\left[\frac{\partial \lossMT}{\partial{\theta_T}}\right]$ in the language of Pauli path integral and quantum rotation 2-design according to \cref{eq:gradient_variance_MPQC}:
\begin{equation}\label{eq:gradient_variance_expressionMT}
\begin{aligned}
&\operatorname{Var}_{(\bm{\theta},\bm{\theta}_{\mathcal{G}},\bm{\theta}_{\mathcal{G}'_T})}\left[\frac{\partial \lossMT}{\partial{\theta_T}}\right] \\
=& \frac{1}{4^{m+3n+6}} 
    \sum_{\substack{\bm{\theta} \in \AngleSet^m \\
    \bm{\theta}_\mathcal{G} \in \AngleSet^{3n}\\
    \bm{\theta}_{\mathcal{G}'_T} \in \AngleSet^{6}\\
        \{\mathbf{P}_T, \mathbf{s}_T^{\left((\bm{\theta},\bm{\theta}_\mathcal{G},\bm{\theta}_{\mathcal{G}'_T}),\alpha\right)}\} = 0}} 
    \sum_{\alpha} c_\alpha^2 f\left(\uniquePathMT,\left(\bm{\theta},\bm{\theta}_\mathcal{G},\bm{\theta}_{\mathcal{G}'_T}\right),I\otimes P_\alpha,op\left(\ket{0}\bra{0}\right)^{\otimes (n+1)}
\otimes\rho\right)^2\\
\geq& \frac{1}{4^{m+3n+6}} 
    \sum_{\substack{\bm{\theta} \in \AngleSet^m \\
    \bm{\theta}_\mathcal{G} \in \AngleSet^{3n}\\
    \bm{\theta}_{\mathcal{G}'_T} \in \AngleSet^{6}\\
        \{\mathbf{P}_T, \mathbf{s}_T^{\left((\bm{\theta},\bm{\theta}_\mathcal{G},\bm{\theta}_{\mathcal{G}'_T}),\beta\right)}\} = 0}} 
    c_\beta^2 f\left(\uniquePathMTb,\left(\bm{\theta},\bm{\theta}_\mathcal{G},\bm{\theta}_{\mathcal{G}'_T}\right),I\otimes P_\beta,op\left(\ket{0}\bra{0}\right)^{\otimes (n+1)}
\otimes\rho\right)^2,
\end{aligned}  
\end{equation}
where we fix a specific Pauli word $P_\beta$ in $O$ such that its backward-propagated Pauli operator $s^{\left(\bm{\theta}, \beta\right)}_{L}$ satisfies $s^{\left(\bm{\theta}, \beta \right)}_{L}|_t \neq I$ for some $\bm{\theta} \in \AngleSet^m$.

We then again derive a lower bound for \cref{eq:gradient_variance_expressionMT} by constructing explicit angle configurations of $\bm{\theta}$, $\bm{\theta}_\mathcal{G}$, and $\bm{\theta}_{\mathcal{G}'_T}$, where all angles take values in $\AngleSet$.
For $\bm{\theta}$, we select configurations such that the Pauli operator backward propagated from $P_\beta$ acts nontrivially on the $t$-th system qubit when reaching the gadget layer. Let $M_t \subseteq \AngleSet^m$ denote the set of such configurations with the maximal cardinality.
From the perspective of backward Pauli propagation, only the gates in the backward light cone of $P_\beta$ following the gadget layer affect the Pauli path $s^{\left(\bm{\theta}, \beta\right)}_{L}$.
Therefore, it suffices to fix at most $\Fnumall$ angles in $\bm{\theta}$ to ensure $s^{\left(\bm{\theta}, \beta\right)}_{L}|_t \neq I$.
This implies that
\begin{equation}
\left|M_t\right|\geq 4^{m - \Fnumall}2^{\Fnumall}=4^m \left(\frac{1}{2}\right)^{\Fnumall},
\end{equation}
where the factor $2^{\Fnumall}$ arises from the observation discussed in \cref{subapp:lower_bound_gradient}, namely that the angle values $\{0, \pi\}$ and $\{\pi/2, 3\pi/2\}$ produce identical effects on the backward propagation of the Pauli path.

Below, we illustrate the choice of $\bm{\theta}_{\mathcal{G}'_T} \in \AngleSet^{6}$ based on $\bm{\theta}$ with the aid of the following figure.
\begin{figure}[H]
    \centering
  \begin{quantikz}
    \lstick{$\ket{0}$}& \gate[1]{op} &\gate[1, style={draw, shape=circle}]{Q_1} & \gate[2]{R_{XX}(\theta_1)}& \gate[2]{R_{YY}(\theta_2)} & \gate[2]{R_{ZZ}(\theta_3)}& \gate[1, style={draw, shape=circle}]{P_1}&\qw &\gate[2]{R_{XX}(\theta_4)}& \gate[2]{R_{YY}(\theta_5)} & \gate[2]{R_{ZZ}(\theta_6)}&\gate[1, style={draw, shape=circle}]{I} \\
    \lstick{$\ket{\psi}$} & \gate[1]{T} &\gate[1, style={draw, shape=circle}]{Q_2} &\qw &\qw & \qw & \gate[1]{\cdots}& \gate[1, style={draw, shape=circle}]{P_2} &\qw & \qw & \qw& \gate[1, style={draw, shape=circle}]{P}\\
  \end{quantikz}
\caption{Expansion of $\gadgetT{\bm{\theta}}$ in terms of Pauli operators for analyzing its effect on Pauli paths. Both $Q_i$ and $P_i$ represent Pauli operators.}\label{fig:for_proof_activation}
  \end{figure}
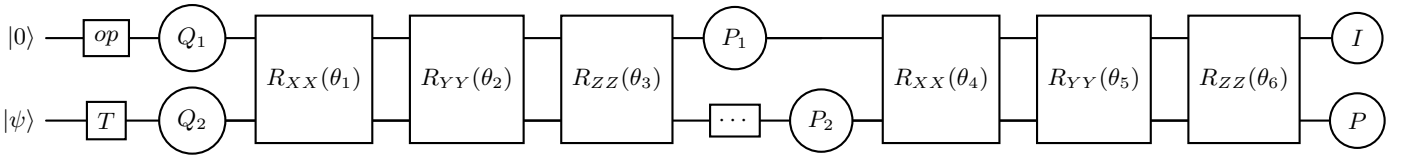
The choice of $\bm{\theta} \in M_t$ ensures that the backward-propagated Pauli operator $P$ is nontrivial, i.e., $P \neq I$. Then, we set the parameters of $\gadgetT{\bm{\theta}}$
to satisfy the condition $\{\mathbf{P}_T, \mathbf{s}_T^{\left((\bm{\theta},\bm{\theta}_\mathcal{G},\bm{\theta}_{\mathcal{G}'_T}),\beta\right)}\} = 0$. This can be achieved according to the following rules:

\begin{itemize}
  \item Choose $\theta_4,\theta_5,\theta_6 \in \AngleSet ^3$ such that $P_1 = P$, $P_2 = I$. 
  \item Choose $\theta_1,\theta_2,\theta_3 \in \AngleSet ^3$ such that $\{P_T, Q_2\} = 0$.
\end{itemize}

The above requirements can always be fulfilled as follows: we choose $\theta_4, \theta_5, \theta_6$ according to the first case in \cref{app:impact_gadget}, which swaps the operator onto the ancilla qubit and yields 16 possible angle configurations, corresponding to $P_1 = P$ and $P_2 = I$. Then, we pick $\theta_1, \theta_2, \theta_3$ according to the first, third, and fourth cases in \cref{app:impact_gadget}, which allow the resulting operator $Q_2$ to be any nontrivial Pauli operator. We then select one such configuration to ensure $\{P_T, Q_2\} = 0$, which also yields at least 16 angle combinations. Denote by $M_{\mathrm{anti}}(\bm{\theta})$ the set of parameter configurations $\theta_1,\ldots,\theta_6$ satisfying these two conditions. Then, for any given $\bm{\theta}\AngleSet^m$ (which determines the Pauli operator $P$), we have
\begin{equation}\label{eq:anti_count}
  \left|M_{\mathrm{anti}}(\bm{\theta})\right|\geq 16^2 = 4^4.
\end{equation}

For $\bm{\theta}_\mathcal{G}$, we adopt the same configuration as in the proof of \cref{thm:absence_BP_MPQC}, which transforms the operator $IP$ into $PI$. We denote this set of configurations as $M{_\mathrm{swap}}(\bm{\theta},\bm{\theta}_{\mathcal{G}'_T})$. 
Here, $\bm{\theta}_{\mathcal{G}'_T}$ is treated as an input, since it determines the angle configuration of the gadget $\gadget{\bm{\theta}}$ placed before $T$. Following a similar argument to that in the proof of \cref{thm:absence_BP_MPQC}, we obtain that for any $\bm{\theta}\AngleSet^m$ and $\bm{\theta}_{\mathcal{G}'_T}\AngleSet^6$,
\begin{equation}
  \left|M_{\mathrm{swap}}(\bm{\theta},\bm{\theta}_{\mathcal{G}'_T})\right|\geq 4^{3(n-K-1)}16^{K} = 4^{3n}\left(\frac{1}{4}\right)^{K+3}.
\end{equation}

From a geometric perspective, when choosing $\bm{\theta}\in M_t$, $\bm{\theta}_{\mathcal{G}'_T} \in M_{\mathrm{anti}}(\bm{\theta})$ and $\bm{\theta}_\mathcal{G} \in M_{\mathrm{swap}}(\bm{\theta},\bm{\theta}_{\mathcal{G}'_T})$, the corresponding Pauli path takes the form illustrated in \cref{fig:Pauli_path_MPQCT}. The configuration $\bm{\theta}_{\mathcal{G}'_T} \in M_{\mathrm{anti}}(\bm{\theta})$ acts as a ``bridge'' that transports the Pauli operator $Q_2$ to the location of gate $T$, while simultaneously ensuring that $\{P_T, Q_2\} = 0$.
\begin{figure}[H]
    \centering
      \includegraphics[width=0.5\textwidth]{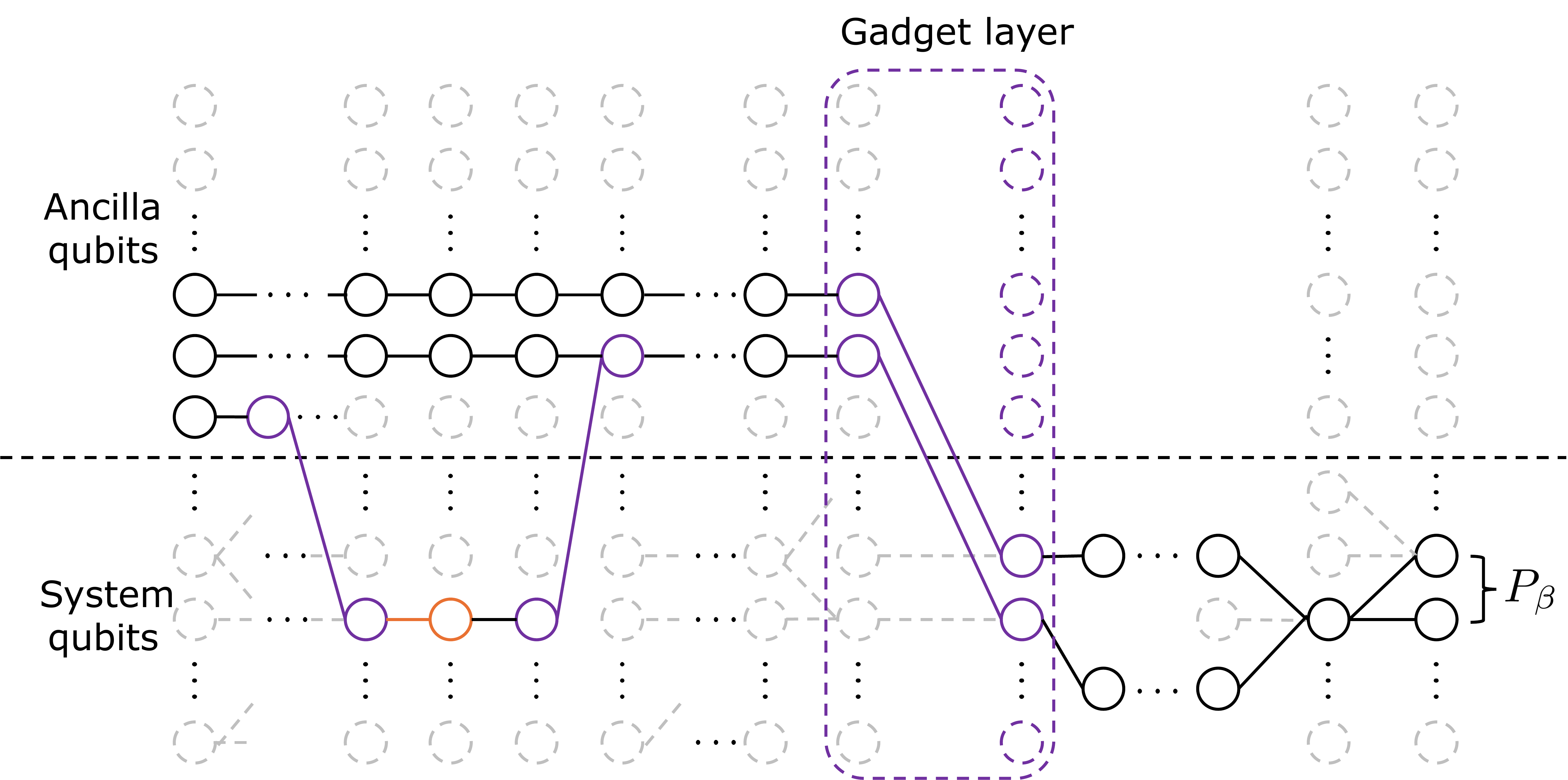}
\caption{Pauli path of the ${T}$-activating MPQC propagated from the observable $P_\beta$. The orange line marks the target single-qubit rotation gate $T$. The choice of parameters $\bm{\theta} \in M_t$ and $\bm{\theta}_{\mathcal{G}'_T} \in M_{\mathrm{anti}}(\bm{\theta})$ ensures that a nontrivial Pauli operator is transported along the backward-propagated path to the location of $T$. The additional $\gadget{\bm{\theta}}$ inserted before $T$ then swaps the Pauli operator onto the corresponding ancilla qubit, thereby preserving a non-vanishing Pauli path.}\label{fig:Pauli_path_MPQCT}
\end{figure}

Therefore $\operatorname{Var}_{(\bm{\theta},\bm{\theta}_{\mathcal{G}},\bm{\theta}_{\mathcal{G}'_T})}\left[\frac{\partial \lossMT}{\partial{\theta_T}}\right]$ can be lower bounded as
\begin{equation}\label{eq:gradient_variance_lower_bound_MT}
\begin{aligned}
&\operatorname{Var}_{(\bm{\theta},\bm{\theta}_{\mathcal{G}},\bm{\theta}_{\mathcal{G}'_T})}\left[\frac{\partial \lossMT}{\partial{\theta_T}}\right] \\
\geq& \frac{1}{4^{m+3n+6}} 
    \sum_{\substack{\bm{\theta} \in M_t \\
    \bm{\theta}_\mathcal{G} \in M_{\mathrm{swap}}(\bm{\theta},\bm{\theta}_{\mathcal{G}'_T})\\
    \bm{\theta}_{\mathcal{G}'_T} \in M_{\mathrm{anti}}(\bm{\theta})}} 
     c_\beta^2 f\left(\uniquePathMTb,\left(\bm{\theta},\bm{\theta}_\mathcal{G},\bm{\theta}_{\mathcal{G}'_T}\right),I\otimes P_\beta,op\left(\ket{0}\bra{0}\right)^{\otimes n}
\otimes\rho\right)^2\\
\geq& \frac{1}{4^{m+3n+6}} 
    \sum_{\substack{\bm{\theta} \in M_t \\
    \bm{\theta}_\mathcal{G} \in M_{\mathrm{swap}}(\bm{\theta},\bm{\theta}_{\mathcal{G}'_T})\\
    \bm{\theta}_{\mathcal{G}'_T} \in M_{\mathrm{anti}}(\bm{\theta})}} 
     c_\beta^2 \tau^{K+1}\\    
\geq& \frac{c_\beta^2\tau^{K+1}}{4^{m+3n+6}} 4^m \left(\frac{1}{2}\right)^{\Fnumall} 4^4 4^{3n}\left(\frac{1}{4}\right)^{K+3}\\
=& c_\beta^2\left(\frac{1}{2}\right)^{\Fnumall+8}\left(\frac{\tau}{4}\right)^{K+1}\\
\geq& \left\|O\right\|_{min}^2\left(\frac{1}{2}\right)^{\order{\log n}}\left(\frac{\tau}{4}\right)^{\order{\log n}}= \Omega\left(\frac{1}{\mathrm{poly}(n)}\right).
\end{aligned}  
\end{equation}
Here, the second inequality holds because the Pauli operator ${\mathbf{s}}^{\hspace{0.1em}\left((\bm{\theta},\bm{\theta}_\mathcal{G},\bm{\theta}_{\mathcal{G}'_T}),\beta\right)}_0$ acts trivially (i.e., as the identity $I$) on all system qubits, while its support on the ancilla qubits has weight at most $K+1$.



\end{proof}

\section{Strategy for activating multiple parameters}\label{app:activate_multiple_paras}
In this section, we present a strategy to activate multiple parameters in PQCs.
Suppose we aim to activate a set of parameters contained in the gate set $\{T_1, T_2, \ldots\}$.
A straightforward approach is to directly extend the method in \cref{app:activate_single_parameter}: specifically, we modify multiple gadgets $\gadget{\bm{\theta}}$ into $\mathcal{G}'_{T_i}(\bm{\theta})$ and insert an additional $\gadget{\bm{\theta}}$ before each $T_i$.
According to the proof technique in \cref{thm:activate_parameters_app},
$\order{\log n}$ parameters can be activated simultaneously.

We next propose a nontrivial approach to activate parameters that are located in close proximity to each other.
Specifically, we first identify the parameters placed nearest to the measurement layer and record the qubits they act on as $t_1, \ldots, t_S$.
{We then consider a backward light cone of these $S$ qubits in the original circuit, which defines a region that contains all parameters to be activated.}
We refer to this region as the \emph{activation zone}, highlighted by the red dashed line in \cref{fig:activate_multi_paras}. 
To activate the parameters within the activation zone, we modify $S$ $\gadget{\bm{\theta}}$ in the gadget layer into 
$\mathcal{G}'_{T_i}(\bm{\theta})$, each acting on qubits $t_1, \ldots, t_S$, respectively. 
Finally, we insert a layer of $\gadget{\bm{\theta}}$ gates within the support of the activation zone.
The resulting circuit is referred to as the \emph{$\{T_1, T_2, \ldots\}$-activating MPQC}, and the entire construction procedure is illustrated in \cref{fig:activate_multi_paras}.

Next, we prove that, under certain conditions, the parameters within the activation zone are trainable.
To establish this result, we introduce the following notations.
Let the unitary blocks in the activation zone be denoted by ${U_i(\theta_i)}$ for $i \in act$ and denote the support size of the activation zone by $K_{act}$.
We represent the corresponding quantum channel as $\channelMTs{\bm{\theta}, \bm{\theta}_{\mathcal{G}}, \bm{\theta}_{\mathcal{G}'_T}}$, and its unitary representation (including the ancilla qubits) as $\unitaryMTs{\bm{\theta}, \bm{\theta}_{\mathcal{G}}, \bm{\theta}_{\mathcal{G}'_T}}$, where, as before, $\bm{\theta}_{\mathcal{G}'_T}$ denotes the parameters in all enlarged gadgets $\mathcal{G}'_{T_i}(\bm{\theta})$, and $\bm{\theta}_{\mathcal{G}}$ collects the parameters in all gadgets $\gadget{\bm{\theta}}$.

\begin{figure}[H]
    \centering
      \includegraphics[width=0.45\textwidth]{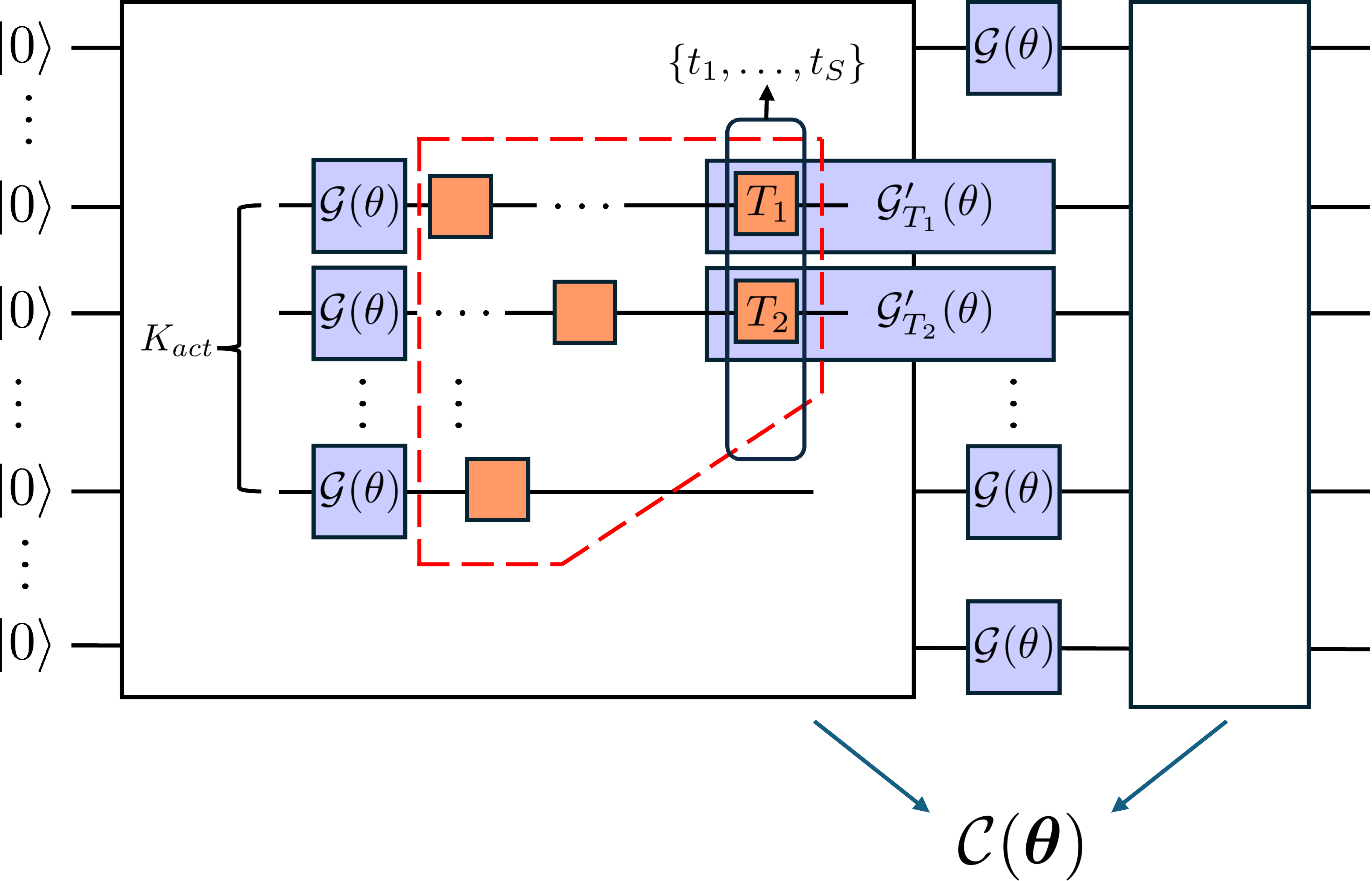}
\caption{
Modified MPQC to activate multiple parameters. The region enclosed by the red dashed line is referred to as the {activation zone}, where the orange boxes indicate parameterized rotation gates. $\{t_1, \ldots, t_S\}$ denotes the set of qubit indices on which the gates in the last layer of the activation zone act. $K_{act}$ denotes the support size of this region. All parameters within this zone can be simultaneously activated by this circuit.
}\label{fig:activate_multi_paras}
\end{figure} 

We are now ready to prove the following theorem, which guarantees that parameters in $\{\theta_i\}_{i \in act}$ are trainable:

\begin{theorem}\label{thm:activate_multi_parameters}
  Consider a $\{T_1,T_2,...\}$-activating MPQC $\channelMTs{\bm{\theta}, \bm{\theta}_{\mathcal{G}}, \bm{\theta}_{\mathcal{G}'_T}}$ measured a $k$-local observable $O = \sum_{\alpha} c_\alpha P_\alpha$ and a parameter $\theta_j$ in the activation zone. Suppose that the following conditions are satisfied:
\begin{itemize}
  \item There exists a Pauli word $P_\beta$ in the observable $O$ and a configuration $\bm{\theta} \in \AngleSet^m$ such that the backward-propagated Pauli operator $s_L^{\left(\bm{\theta}, \beta\right)}$ reaches the gadget layer and satisfies $s_L^{\left(\bm{\theta}, \beta\right)}|_{\{t_1, \ldots, t_S\}} \neq I$.
  \item For arbitrary Pauli word $P$ whose support lies in $\{t_1,...t_S\}$, we backward propagate the unitaries in the activation zone, i.e, $\{\mathbf{U}_{i}(\theta_i)\}$, $i\in act$ 
from arbitrary Pauli word $\mathbf{P}$ whose support lies in $\{t_1,...t_S\}$, achieve another 2$n$-qubit Pauli path $\vec{\mathbf{s}}_{act}^{\hspace{0.1em}\left(\{\theta_i\}_{i\in act},\mathbf{P}\right)}$. Suppose there exist some $\mathbf{P}_{act}$ and $\{\theta_i\}_{i\in act}$ for all  $\theta_i \in \AngleSet$ such that
\begin{equation}
\{\mathbf{P}_j,\mathbf{s}^{\hspace{0.1em}\left(\{\theta_i\}_{i\in act},\mathbf{P}_{act}\right)}_{act|\mathbf{P}_j}\} = 0,
\end{equation}
where $\mathbf{s}^{\hspace{0.1em}\left(\{\theta_i\}_{i\in act},\mathbf{P}_{act}\right)}_{act|\mathbf{P}_j}$ denotes the Pauli operator associated with the segment following $\mathbf{U}_j(\theta_j)$ in $\vec{\mathbf{s}}_{act}^{\hspace{0.1em}\left(\{\theta_i\}_{i\in act},P\right)}$.
\item $K$, $K_{act}$, $\Fnumall$, and $\Fnumact$ (defined as the number of parameters within the activation zone) are all of order $\order{\log n}$.
\end{itemize}
Then, we have that for the loss function of the $\{T_1, T_2, \ldots\}$-activating MPQC:
\begin{equation}
  \begin{aligned}
  &\lossMTs = \tr{\channelMTs{\bm{\theta}, \bm{\theta}_{\mathcal{G}}, \bm{\theta}_{\mathcal{G}'_T}}(\rho)O}\\
&= \tr{\unitaryMTs{\bm{\theta}, \bm{\theta}_{\mathcal{G}}, \bm{\theta}_{\mathcal{G}'_T}}\left(op\left(\ket{0}\bra{0}\right)^{\otimes (n+K_{act})}\otimes\rho\right)\unitaryMTs{\bm{\theta}, \bm{\theta}_{\mathcal{G}}, \bm{\theta}_{\mathcal{G}'_T}}^\dagger I\otimes O},
  \end{aligned}
\end{equation}
the gradient variance with respect to a parameter $\theta_j$ for $j \in \mathrm{act}$ can be lower bounded as
\begin{equation}\label{eq:activate_multi_parameter_lower_bound}
  \operatorname{Var}_{(\bm{\theta},\bm{\theta}_{\mathcal{G}},\bm{\theta}_{\mathcal{G}'_T})}\left[\frac{\partial \lossMTs}{\partial{\theta_j}}\right]\geq \left\|O\right\|_{min}^2\left(\frac{1}{2}\right)^{\Fnumall+\Fnumact+8S}\left(\frac{\tau}{4}\right)^{K+K_{act}}= \Omega\left(\frac{1}{\mathrm{poly}(n)}\right).
\end{equation}
\end{theorem}
\begin{proof}
The proof technique is similar to that of \cref{thm:activate_parameters_app}.
The main difference lies in the need to handle the backward propagation of the Pauli path throughout the entire activation zone.
We again express $\operatorname{Var}_{(\bm{\theta},\bm{\theta}_{\mathcal{G}},\bm{\theta}_{\mathcal{G}'_T})}\left[\frac{\partial \lossMTs}{\partial{\theta_j}}\right]$ in the form of Pauli path integral combined with the quantum rotation 2-design, and derive a lower bound by focusing on a specific $P_\beta$ in $O$:
\begin{equation}\label{eq:gradient_variance_expressionMT_for multi_pars}
\begin{aligned}
&\operatorname{Var}_{(\bm{\theta},\bm{\theta}_{\mathcal{G}},\bm{\theta}_{\mathcal{G}'_T})}\left[\frac{\partial \lossMTs}{\partial{\theta_j}}\right] \\
=& \frac{1}{4^{m+3(n-S+K_{act})+6S}} 
    \sum_{\substack{\bm{\theta} \in \AngleSet^m \\
    \bm{\theta}_\mathcal{G} \in \AngleSet^{3(n-S+K_{act})}\\
    \bm{\theta}_{\mathcal{G}'_T} \in \AngleSet^{6S}\\
        \{\mathbf{P}_j, \mathbf{s}_j^{\left((\bm{\theta},\bm{\theta}_\mathcal{G},\bm{\theta}_{\mathcal{G}'_T}),\beta\right)}\} = 0}} 
    \sum_{\alpha} c_\alpha^2 f\left(\uniquePathMT,\left(\bm{\theta},\bm{\theta}_\mathcal{G},\bm{\theta}_{\mathcal{G}'_T}\right),I\otimes P_\alpha,op\left(\ket{0}\bra{0}\right)^{\otimes (n+K_{act})}
\otimes\rho\right)^2\\
\geq& \frac{1}{4^{m+3(n-S+K_{act})+6S}} 
    \sum_{\substack{\bm{\theta} \in \AngleSet^m \\
    \bm{\theta}_\mathcal{G} \in \AngleSet^{3(n-S+K_{act})}\\
    \bm{\theta}_{\mathcal{G}'_T} \in \AngleSet^{6S}\\
        \{\mathbf{P}_j, \mathbf{s}_j^{\left((\bm{\theta},\bm{\theta}_\mathcal{G},\bm{\theta}_{\mathcal{G}'_T}),\beta\right)}\} = 0}} 
    c_\beta^2 f\left(\uniquePathMTb,\left(\bm{\theta},\bm{\theta}_\mathcal{G},\bm{\theta}_{\mathcal{G}'_T}\right),I\otimes P_\beta,op\left(\ket{0}\bra{0}\right)^{\otimes (n+K_{act})}
\otimes\rho\right)^2.
\end{aligned}  
\end{equation}

We again consider a specific set of angle configurations in the circuit, where all angles belong to $\AngleSet$. To formalize our construction, we partition the angles in $\bm{\theta}$ into three parts:
\begin{itemize}
  \item $\bm{\theta}_{af} \in [0, 2\pi)^{\#af}$ denotes the set of angles after the gadget layer, where $\#af$ is the number of such parameterized gates;
  \item $\bm{\theta}_{act} \in [0, 2\pi)^{\#act}$ denotes the set of angles within the activation zone, where $\#act$ is the number of such parameterized gates;
  \item $\bar{\bm{\theta}} \in [0, 2\pi)^{m - \#af - \#act}$ denotes the remaining angles in $\bm{\theta}$, excluding $\bm{\theta}_{af}$ and $\bm{\theta}_{act}$.
\end{itemize}
In the following, we demonstrate how to choose $\bm{\theta}_{af}$, $\bm{\theta}_{\mathcal{G}'_T}$, $\bm{\theta}_{act}$, $\bm{\theta}_\mathcal{G}$ and $\bar{\bm{\theta}}$ in the discrete angle set to derive a lower bound of order $\Omega\left(\frac{1}{\mathrm{poly}(n)}\right)$ for \cref{eq:gradient_variance_expressionMT_for multi_pars}. 

We first select configurations of $\bm{\theta}_{af}$ such that the backward-propagated Pauli operator $s^{\left(\bm{\theta}, \beta\right)}_{L}$ (i.e., the operator reaching the gadget layer) satisfies $s^{\left(\bm{\theta}, \beta\right)}_{L}|_{\{t_1,...t_S\}} \neq I$. Since the backward light cone of $P_\beta$ before the gadget layer contains at most $\Fnumall$ gates, we only need to fix at most $\Fnumall$ angles in $\bm{\theta}_{af}$ to make this requirement hold. Let $M_{af}$ denote the maximal set of such angle configurations of $\bm{\theta}_{af}$, we have
\begin{equation}
\left|M_{af}\right|\geq 4^{\#af - \Fnumall}2^{\Fnumall}=4^{\#af}\left(
  \frac{1}{2}\right)^{\Fnumall}.
\end{equation}

We then illustrate the choice of $\bm{\theta}_{\mathcal{G}'_T}$. Specifically, we select $\bm{\theta}_{\mathcal{G}'_T}$ such that the Pauli operator propagated to the activation zone becomes $\mathbf{P}_{act}$, ensuring that $\{\mathbf{P}_j, \mathbf{s}^{\hspace{0.1em}\left(\{\theta_i\}_{i\in act}, \mathbf{P}_{act}\right)}_{act|\mathbf{P}_j}\} = 0$ for some $\bm{\theta}_{act}$. 
Here, we employ the same discrete angle construction of $\gadgetT{\bm{\theta}}$ as used in the proof of \cref{thm:activate_parameters_app}, which first swaps the nontrivial Pauli operator to the ancilla and then uses another three angles to generate the desired Pauli operator $\mathbf{P}_{act}$. 
Again, we denote by $M_{\mathrm{anti}}(\bm{\theta}_{af})$ the set of parameter configurations $\bm{\theta}_{\mathcal{G}'_T} \in \AngleSet^{6S}$ that satisfy the required condition. 
Similar to the counting argument in \cref{eq:anti_count}, we obtain
\begin{equation}
  \left|M_{\mathrm{anti}}(\bm{\theta}_{af})\right| \geq 4^{4S},
\end{equation}
as the weight of $\mathbf{P}_{act}$ is at most $S$.

We now move on to the choice of $\bm{\theta}_{act}$. 
We choose $\bm{\theta}_{act}$ such that 
$\{\mathbf{s}^{\hspace{0.1em}\left(\{\theta_i\}_{i \in Act}, \mathbf{P}_{act}\right)}_{act|\mathbf{P}_j}, \mathbf{P}_j\} = 0$. 
Since the number of parameterized gates in the activation zone 
is at most $\Fnumact$, we only need to fix $\Fnumact$ angles.   
Let $M_{act}$ denote the maximal set of such angle configurations of $\bm{\theta}_{act}$. 
Then, we have
\begin{equation}
\left|M_{act}\right|\geq 4^{\#act-\Fnumact}2^{\Fnumact}=4^{\#act}\left(
  \frac{1}{2}\right)^{\Fnumact}.
\end{equation}

For $\bm{\theta}_\mathcal{G}$, we adopt the same configuration as in the proof of \cref{thm:absence_BP_MPQC}, 
which transforms the operator $IP$ into $PI$. 
Here, $\bm{\theta}_\mathcal{G}$ consists of two parts: one located in the gadget layer and the other positioned 
before the activation zone (as shown on the leftmost side of \cref{fig:activate_multi_paras}). 
For the $\bm{\theta}_\mathcal{G}$ in the gadget layer, since the support size of the Pauli operator propagated from $P_\beta$ 
is at most $K$, and we replace $S$ of them with $\gadgetT{\bm{\theta}}$, we only need to fix the angle configurations of $K-S$ gadgets, 
while the remaining $n-K-S$ gadgets can take arbitrary angle configurations in $\AngleSet^{3}$, 
as their inputs are $I\otimes I$. 
For the $\gadget{\bm{\theta}}$ gates located at the left boundary of the activation zone, 
since there are at most $K_{act}$ such gadgets, we need to fix at most $K_{act}$ of them. 
We denote this set of configurations as $M_{\mathrm{swap}}(\bm{\theta}_{af},\bm{\theta}_{act})$. 
Following a similar argument to that in the proof of \cref{thm:absence_BP_MPQC}, we obtain that for any 
$\bm{\theta}_{af}$ and $\bm{\theta}_{act}$,
\begin{equation}
  \left|M_{\mathrm{swap}}(\bm{\theta}_{af},\bm{\theta}_{act})\right|
  \geq 4^{3(n-K-S)}16^{K-S}16^{K_{act}} 
  = 4^{3(n-S+K_{act})}\left(\frac{1}{4}\right)^{K_{act}+K+2S}.
\end{equation}

For $\bar{\bm{\theta}}$, it can take arbitrary angle configurations within $\AngleSet^{m - \#af - \#act}$.

From a geometric perspective, when choosing $\bm{\theta}_{af}\in M_{af}$, $\bm{\theta}_{\mathcal{G}'_T} \in M_{\mathrm{anti}}(\bm{\theta}_{af})$,
$\bm{\theta}_{act}\in M_{act}$
and $\bm{\theta}_\mathcal{G} \in M_{\mathrm{swap}}(\bm{\theta}_{af},\bm{\theta}_{act})$, the corresponding Pauli path takes the form illustrated in \cref{fig:Pauli_path_MPQCT}, which extend the case in \cref{fig:Pauli_path_MPQCT} to multiple parameters. 

\begin{figure}[H]
    \centering
\includegraphics[width=0.6\textwidth]{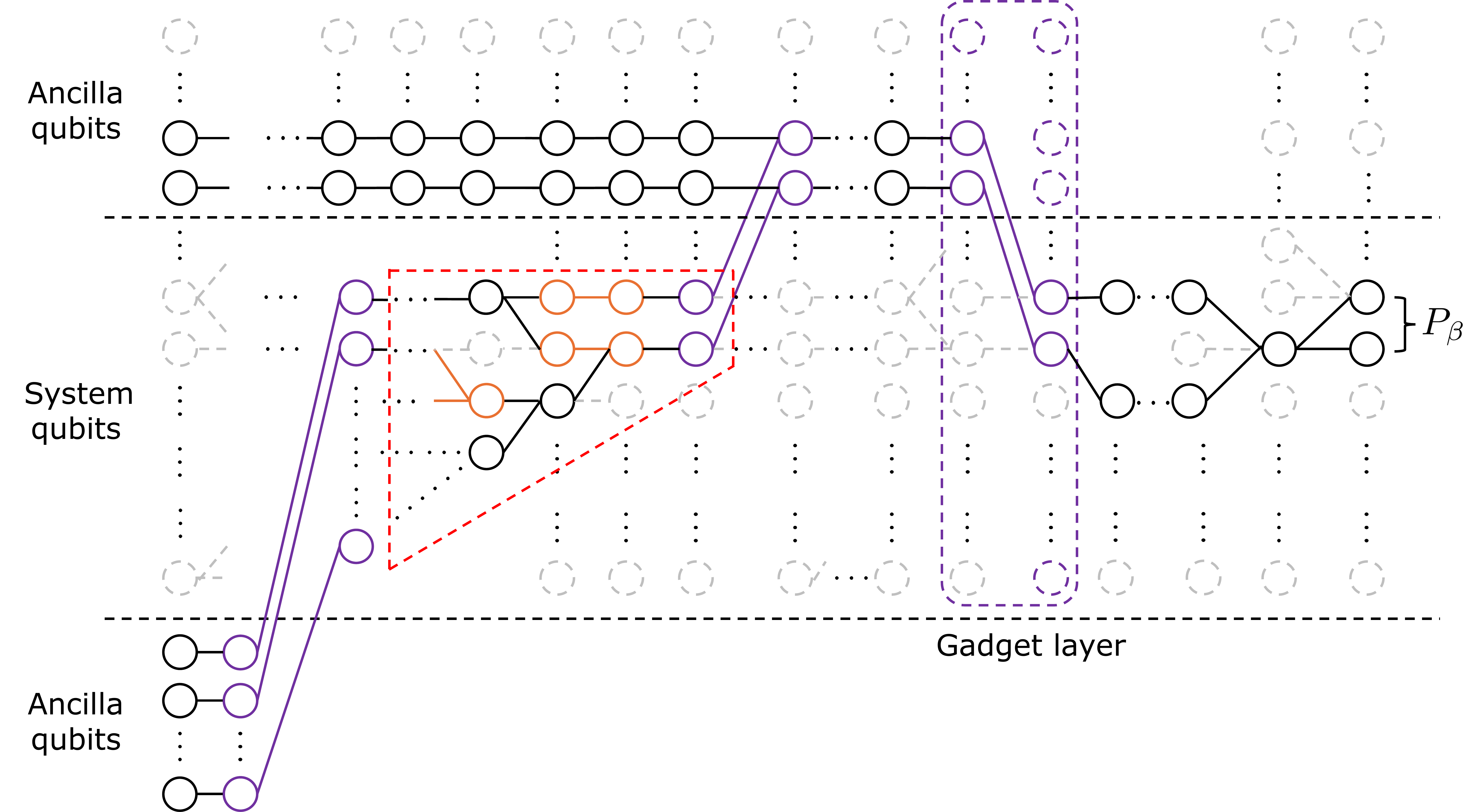}
\caption{Pauli path propagated from the observable $P_\beta$ of the $\{T_1, T_2, \ldots\}$-activating MPQC. 
The region enclosed by the red dashed line denotes the activation zone, while the ancilla qubits introduced by the insertion 
of $\gadget{\bm{\theta}}$ gates before the activation zone are located below the black dashed line. 
The choices of $\bm{\theta}_{af} \in M_{\mathrm{af}}$ and $\bm{\theta}_{\mathcal{G}'_T} \in M_{\mathrm{anti}}(\bm{\theta}_{af})$ 
ensure that $\mathbf{P}_{act}$ is backward propagated into the activation zone. 
Then, by choosing $\bm{\theta}_{act} \in M_{act}$, we guarantee that 
$\{\mathbf{P}_j, \mathbf{s}^{\hspace{0.1em}\left(\{\theta_i\}_{i\in act},\mathbf{P}_{act}\right)}_{act|\mathbf{P}_j}\} = 0$. 
Finally, the configuration $\bm{\theta}_{\mathcal{G}} \in M_{\mathrm{swap}}(\bm{\theta}_{af},\bm{\theta}_{act})$ 
swaps these operators onto the corresponding ancilla qubits, thereby ensuring nonvanishing Pauli paths 
and maintaining finite gradient variance for the parameters $\theta_j$ within the activation zone.
}\label{fig:Pauli_path_MPQC_activating_pras}
\end{figure}

Then $\operatorname{Var}_{(\bm{\theta},\bm{\theta}_{\mathcal{G}},\bm{\theta}_{\mathcal{G}'_T})}\left[\frac{\partial \lossMTs}{\partial{\theta_j}}\right]$ can be lower bounded as
\begin{equation}\label{eq:gradient_variance_lower_bound_MTs}
\begin{aligned}
&\operatorname{Var}_{(\bm{\theta},\bm{\theta}_{\mathcal{G}},\bm{\theta}_{\mathcal{G}'_T})}\left[\frac{\partial \lossMTs}{\partial{\theta_j}}\right] \\
\geq& \frac{1}{4^{m+3(n-S+K_{act})+6S}} 
    \sum_{\substack{\bm{\theta}_{af} \in M_{af},\bm{\theta}_{act} \in M_{act},\bar{\bm{\theta}} \\
    \bm{\theta}_{\mathcal{G}'_T} \in M_{\mathrm{anti}}(\bm{\theta}_{af})\\
    \bm{\theta}_\mathcal{G} \in M_{\mathrm{swap}}(\bm{\theta}_{af},\bm{\theta}_{act})}} 
     c_\beta^2 f\left(\uniquePathMTb,\left(\bm{\theta},\bm{\theta}_\mathcal{G},\bm{\theta}_{\mathcal{G}'_T}\right),I\otimes P_\beta,op\left(\ket{0}\bra{0}\right)^{\otimes (n+K_{act})}
\otimes\rho\right)^2\\
\geq& \frac{1}{4^{m+3(n-S+K_{act})+6S}} 
    \sum_{\substack{\bm{\theta}_{af} \in M_{af},\bm{\theta}_{act} \in M_{act},\bar{\bm{\theta}} \\
    \bm{\theta}_{\mathcal{G}'_T} \in M_{\mathrm{anti}}(\bm{\theta}_{af})\\
    \bm{\theta}_\mathcal{G} \in M_{\mathrm{swap}}(\bm{\theta}_{af},\bm{\theta}_{act})}} 
     c_\beta^2 \tau^{K+K_{act}}\\
\geq& \frac{1}{4^{m+3(n-S+K_{act})+6S}} 
    \sum_{\substack{\bm{\theta}_{af} \in M_{af},\bm{\theta}_{act} \in M_{act},\bar{\bm{\theta}}}} 
    \left|M_{\mathrm{anti}}(\bm{\theta}_{af})\right| \left|M_{\mathrm{swap}}(\bm{\theta}_{af},\bm{\theta}_{act})\right| c_\beta^2 \tau^{K+K_{act}}\\
\geq& \frac{c_\beta^2\tau^{K+K_{act}}}{4^{m+3(n-S+K_{act})+6S}} \underbrace{4^{\#af}\left(
  \frac{1}{2}\right)^{\Fnumall}}_{\leq |M_{af}|}\underbrace{4^{\#act}\left(
  \frac{1}{2}\right)^{\Fnumact}}_{\leq |M_{act}|} \underbrace{4^{m - \#af - \#act}}_{ \text{all possible }\bar{\bm{\theta}}} \underbrace{4^{4S}}_{\leq \left|M_{\mathrm{anti}}(\bm{\theta}_{af})\right|}\underbrace{4^{3(n-S+K_{act})}\left(\frac{1}{4}\right)^{K_{act}+K+2S}}_{\leq \left|M_{\mathrm{swap}}(\bm{\theta}_{af},\bm{\theta}_{act})\right|} \\
=& c_\beta^2\left(\frac{1}{2}\right)^{\Fnumall+\Fnumact+8S} \left(\frac{\tau}{4}\right)^{K_{act}+K}\\
\geq& \left\|O\right\|_{min}^2\left(\frac{1}{2}\right)^{\order{\log n}}\left(\frac{\tau}{4}\right)^{\order{\log n}}= \Omega\left(\frac{1}{\mathrm{poly}(n)}\right).
\end{aligned}  
\end{equation}
Similarly, the second inequality holds because the Pauli operator ${\mathbf{s}}^{\hspace{0.1em}\left((\bm{\theta},\bm{\theta}_\mathcal{G},\bm{\theta}_{\mathcal{G}'_T}),\beta\right)}_0$ acts trivially (i.e., as the identity $I$) on all system qubits, while its support on the ancilla qubits has weight at most $K+K_{act}$.
\end{proof}
\begin{remark}
In \cref{thm:activate_multi_parameters}, we assumed the existence of a Pauli word $P_\beta$ in the observable $O$ and a configuration $\bm{\theta} \in \AngleSet^m$ such that the Pauli path $s_L^{\left(\bm{\theta}, \beta\right)}|_{\{t_1, \ldots, t_S\}} \neq I$. In fact, this assumption can be weakened to only require the existence of a Pauli word $P_\beta$ and a configuration $\bm{\theta} \in \AngleSet^m$ such that at least one Pauli path $s_L^{(\bm{\theta}, \beta)}$ has weight at least $S$. This relaxed condition can be handled using a construction similar to that in \cref{fig:G_T'_different_sys}. If no such path exists, we note that $S\leq K_{act} = \order{\log n}$ according to \cref{thm:activate_multi_parameters}, and we can always shift the gadget layer earlier in the circuit to increase the weight of the Pauli operator reached the gadget layer, thereby ensuring that activation is still possible.
\end{remark}

\section{Proof of \cref{thm:absence_BP_robustness}}\label{app:noise_robustness}
In this section, we prove that BP can also be eliminated in MPQCs even in the presence of noise. We begin by introducing the noise model and explaining how it affects the Pauli path. Finally, we present the noisy counterparts of \cref{thm:absence_BP_MPQC_parameters_app}, \cref{thm:absence_BP_MPQC_parameters_app}, and \cref{thm:activate_parameters_app}, thereby completing the proof of \cref{thm:absence_BP_robustness}.

\subsection{Noise model and Pauli path integral with noise}
We consider the case of Pauli type noises, which is a common type of noise in quantum circuits and can be described by the following quantum channel $\mathcal{N}$:
\begin{equation}
  \mathcal{N}(\rho)= (1-\sum_i p_i) \rho + \sum_{i} p_i \sigma_i \rho \sigma_i^\dagger,
\end{equation}
where $\sigma_i$ denotes a non-identity Pauli operator, $p_i$ is the corresponding probability, and the total probability $\sum_i p_i < 1$ characterizes the noise strength, which we denote by $\gamma_\mathcal{N}$. In our discussion, we assume that the Pauli noises appear in the quantum circuit. The gates are followed by Pauli noise channels, as shown in Fig~\ref{fig:noisy_gate}.

\begin{figure}[htbp]
  \begin{quantikz}
    \lstick{} & \gate[2]{U} & \gate[wires=2,style={starburst,draw=red,line width=1pt,inner xsep=-4pt,inner ysep=-5pt}]{\mathcal{N}} & \qw \\
    \lstick{} & \qw & \qw & \qw \\
  \end{quantikz}
  \caption{The noisy channel $\widetilde{U}$: ideal gate $U$ followed by Pauli noise channel $\mathcal{N}$ acting on the output.}\label{fig:noisy_gate}
\end{figure}
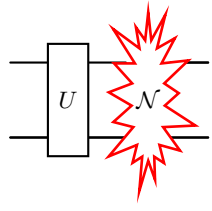

Because of the anti-commuting property of the Pauli operator, the Pauli noise channel $\mathcal{N}$ ($\mathcal{N}^\dagger$) acting on the normalized Pauli operator $s$ can be expressed as:
\begin{equation}\label{eq:Pauli_noise_channel}
  \mathcal{N}(s) = \mathcal{N}^\dagger(s) = (1-\sum_i p_i)s + \sum_{i} p_i \sigma_i s \sigma_i^\dagger=\left(1-2\sum_{i} \mathbf{1}_{ac}(s,\sigma_i) p_i\right) s,
\end{equation}
where $\mathbf{1}_{ac}(s,\sigma_i)$ is the indicator function that equals to $1$ if $s$ and $\sigma_i$ anti-commute, otherwise it equals to $0$.
Thus there is $\mathcal{N}(s)= cs$ for some constant $c$, and because of $\sum_i p_i=\gamma_\mathcal{N}$, we have 
\begin{equation}\label{eq:Pauli_noise_coefficients}
  c =\tr{s\mathcal{N}(s)}= 1-2\sum_{i} \mathbf{1}_{ac}(s,\sigma_i) p_i \geq 1-2\gamma_\mathcal{N}.
\end{equation}

We assume that there is a Pauli noise channel $\mathcal{N}_i$ is following the $i$-block $\mathbf{U}_i(\theta_i)$ in the MPQC. Or in other words, the ideal gate $\mathbf{U}_i(\theta_i)$ is replaced by the noisy channel $\widetilde{\mathbf{U}}_j(\theta_j)(\cdot)=\mathcal{N}_i\circ \mathbf{U}_i(\theta_i)(\cdot)\mathbf{U}_i(\theta_i)^\dagger$ in the noisy MPQC. Moreover, we assume that each $\mathcal{N}_i$ takes the form
\begin{equation}
\mathcal{N}_i = \mathcal{I} \otimes \mathcal{N}_i',
\end{equation}
where $\mathcal{N}_i'$ is a Pauli noise channel acting on the same qubits as $\mathbf{U}_i(\theta_i)$, which is a reasonable assumption for current quantum devices.

Similarly, for the two-qubit gates in the gadget layer, we assume that a Pauli noise channel is applied after each layer. Specifically, the ideal sequence of ideal gates 
\[
\prod_{i=1}^n\left(R_{Z_iZ_{i+n}}(\bm{\theta}_{\mathcal{G}_{i,1}})
R_{Y_iY_{i+n}}(\bm{\theta}_{\mathcal{G}_{i,2}})
R_{X_iX_{i+n}}(\bm{\theta}_{\mathcal{G}_{i,3}})\right)
\] 
is transformed into 
\begin{equation}
  \mathcal{N}_{\mathcal{G}_{1}}\circ \prod_{i=1}^n\mathcal{R}_{Z_iZ_{i+n}}(\bm{\theta}_{\mathcal{G}_{i,1}})\circ  \mathcal{N}_{\mathcal{G}_{2}}\circ \prod_{i=1}^n\mathcal{R}_{Y_iY_{i+n}}(\bm{\theta}_{\mathcal{G}_{i,2}})  \circ  \mathcal{N}_{\mathcal{G}_{3}}\circ \prod_{i=1}^n\mathcal{R}_{X_iX_{i+n}}(\bm{\theta}_{\mathcal{G}_{i,3}}),
\end{equation}
where we write $\mathcal{R}_P(\theta)$ as the channel representation for rotation gate $R_P(\theta)$. This assumption is reasonable since the gates $\prod_{i=1}^nR_{Z_iZ_{i+n}}(\bm{\theta}_{\mathcal{G}_{i,1}})$ ($\prod_{i=1}^n R_{Y_iY_{i+n}}(\bm{\theta}_{\mathcal{G}_{i,2}})$ or $\prod_{i=1}^n R_{X_iX_{i+n}}(\bm{\theta}_{\mathcal{G}_{i,3}})$) can be applied in parallel within a single layer. Finally, we define a Pauli noise channel $\mathcal{N}_{op}$ that follows the application of $n$ copies of $op$, i.e., $\widetilde{op}^{\otimes n}(\cdot)=\mathcal{N}_{op}\circ op^{\otimes n}(\cdot)$.

As a result, the noisy circuit $\unitaryMNoi{\bm{\theta}, \bm{\theta}_\mathcal{G}}$ can be expressed as:
\begin{equation}\label{eq:noisy_MPQC}
  \begin{aligned}
\unitaryMNoi{\bm{\theta},\bm{\theta}_\mathcal{G}}
  =& \mathcal{N}_m \circ \mathbf{U}_m(\theta_m)  \cdots \mathcal{N}_{L+1}\circ \mathbf{U}_{L+1}(\theta_{L+1}) \circ \mathcal{N}_{\mathcal{G}_{1}}\circ \prod_{i=1}^n\mathcal{R}_{Z_iZ_{i+n}}(\bm{\theta}_{\mathcal{G}_{i,1}})\circ  \\
  &\mathcal{N}_{\mathcal{G}_{2}}\circ \prod_{i=1}^n\mathcal{R}_{Y_iY_{i+n}}(\bm{\theta}_{\mathcal{G}_{i,2}})  \circ  \mathcal{N}_{\mathcal{G}_{3}}\circ \prod_{i=1}^n\mathcal{R}_{X_iX_{i+n}}(\bm{\theta}_{\mathcal{G}_{i,3}}) \circ \mathcal{N}_L \circ \mathbf{U}_L(\theta_L) \cdots \mathcal{N}_1 \circ \mathbf{U}_1(\theta_1).
  \end{aligned}
\end{equation}  

Under this condition, the noisy loss function $\lossMNoi$, corresponding to \cref{eq:MPQC_Pauli_path_integral_noiseless}, can be expressed as:
\begin{small}
\begin{equation}\label{eq:MPQC_Pauli_path_integral_noisy}
  \begin{aligned}
&\lossMNoi =\tr{\unitaryMNoi{\bm{\theta},\bm{\theta}_\mathcal{G}}
\left(\widetilde{op}\left(\ket{0}\bra{0}\right)^{\otimes n}
    \otimes \rho \right)  I \otimes O}\\
    &= \sum_{\alpha,\mathbf{s}_m} c_\alpha\tr{I\otimes P_\alpha \mathbf{s}_m} \tr{\unitaryMNoi{\bm{\theta},\bm{\theta}_\mathcal{G}} \left( \widetilde{op}\left(\ket{0}\bra{0}\right)^{\otimes n}
    \otimes \rho \right) \mathbf{s}_m}\\
&=\sum_{\substack{\alpha,\mathbf{s}_m,\mathbf{s}_{m-1},\cdots,\mathbf{s}_0\\  \mathbf{s}_{\mathcal{G}_{1,1}},\mathbf{s}_{\mathcal{G}_{1,2}},\cdots,\mathbf{s}_{\mathcal{G}_{n,3}}
  }} c_\alpha\tr{I\otimes P_\alpha \mathbf{s}_m} 
  \tr{\mathcal{N}_m^\dagger(\mathbf{s}_m){\mathbf{U}}_m(\theta_m) \mathbf{s}_{m-1} {\mathbf{U}}_m(\theta_m)^\dagger} \cdots \tr{\mathcal{N}_{L+1}^\dagger(\mathbf{s}_{L+1}) {\mathbf{U}}_{L+1}(\theta_{L+1}) \mathbf{s}_{\mathcal{G}_{1,1}} {\mathbf{U}}_{L+1}(\theta_{L+1})^\dagger}\cdot\\ 
  &\qquad
  \cdot \tr{\mathcal{N}_{\mathcal{G}_{1}}^\dagger(\mathbf{s}_{\mathcal{G}_{1,1}}) {R}_{11}(\bm{\theta}_{\mathcal{G}_{11}}) \mathbf{s}_{\mathcal{G}_{1,2}} {R}_{11}(-\bm{\theta}_{\mathcal{G}_{11}})} \tr{\mathbf{s}_{\mathcal{G}_{1,2}}{R}_{12}(\bm{\theta}_{\mathcal{G}_{12}}) \mathbf{s}_{\mathcal{G}_{1,3}} {R}_{12}(-\bm{\theta}_{\mathcal{G}_{12}})}\cdots\tr{\mathbf{s}_{\mathcal{G}_{n,3}}{R}_{n3}(\bm{\theta}_{\mathcal{G}_{n3}}) \mathbf{s}_{L} {R}_{n3}(-\bm{\theta}_{\mathcal{G}_{n3}})}\cdot\\
  &\qquad \cdot \tr{\mathcal{N}_{L}^\dagger(\mathbf{s}_{L}) {\mathbf{U}}_L(\theta_L) \mathbf{s}_{L-1} {\mathbf{U}}_L(\theta_L)^\dagger}
  \cdots\tr{\mathcal{N}_{1}^\dagger(\mathbf{s}_{1}) {\mathbf{U}}_1(\theta_1) \mathbf{s}_0 {\mathbf{U}}_1(\theta_1)^\dagger} \tr{\mathcal{N}_{op}^\dagger(\mathbf{s}_{0}){op}\left(\ket{0}\bra{0}\right)^{\otimes n}\otimes\rho}\\
  & =\sum_{\alpha,\vec{\mathbf{s}}} c_\alpha g(\vec{\mathbf{s}}) f\left(\vec{\mathbf{s}},\left(\bm{\theta},\bm{\theta}_\mathcal{G}\right),I\otimes P_\alpha,op\left(\ket{0}\bra{0}\right)^{\otimes n}
\otimes\rho\right),
  \end{aligned}
\end{equation}  
\end{small}
where $g(\vec{\mathbf{s}})$ is the noise effect factor on the Pauli path $\vec{\mathbf{s}}$, defined as the product of the coefficients computed in \cref{eq:Pauli_noise_channel}.
\begin{equation}
g(\vec{\mathbf{s}}):=\tr{\mathbf{s}_{0}\mathcal{N}_{op}^\dagger\otimes\mathcal{I}(\mathbf{s}_{0})}\prod_{i=1}^m \tr{\mathbf{s}_{i}\mathcal{N}_{i}^\dagger(\mathbf{s}_{i})}\prod_{i=1}^3 \tr{\mathbf{s}_{\mathcal{G}_{i,1}}\mathcal{N}_{\mathcal{G}_{i}}^\dagger(\mathbf{s}_{\mathcal{G}_{i,1}})}.
\end{equation}

\subsection{Lower bounds of variance and gradient variance of the loss function of noisy MPQCs}
With the descrptions in the previous subsection, we can prove the following theorem
\begin{theorem}\label{thm:absence_BP_noisy_MPQC}
For an MPQC measured with a $k$-local observable $O = \sum_\alpha c_\alpha P_\alpha$, suppose the conditions stated in \cref{subapp:lower_bound_gradient} hold, then under Pauli noise with strength at most $\gamma<1/2$ applied after each block, the variance of the loss function is lower bounded by
  \[
  \operatorname{Var}_{(\bm{\theta},\bm{\theta}_\mathcal{G})}\left[\lossMNoi\right] \geq (1-2\gamma)^{2(\Fnumall+4)} \left(\frac{\tau}{4}\right)^{K}\left\|O\right\|_{HS}^2 = \Omega\left(\frac{1}{\mathrm{poly}(n)}\right).
  \]
\end{theorem}
\begin{proof}
  We first express the variance of the loss function of noisy MPQC:
  \begin{equation}\label{eq:variance_noisy_MPQC}
    \begin{aligned}
    &\operatorname{Var}_{(\bm{\theta},\bm{\theta}_\mathcal{G})}\left[\lossMNoi\right]\\ 
& = 
\mathbb{E}_{\left(\bm{\theta} , \bm{\theta}_\mathcal{G}\right)} 
\sum_{\alpha,\beta,\vec{\mathbf{s}}, \vec{\mathbf{s}}\hspace{0.1em}'} c_\alpha c_\beta g(\vec{\mathbf{s}})g(\vec{\mathbf{s'}})
f\left(\vec{\mathbf{s}},\left(\bm{\theta},\bm{\theta}_\mathcal{G}\right),I\otimes P_\alpha,op\left(\ket{0}\bra{0}\right)^{\otimes n}
\otimes\rho\right)  
f\left(\vec{\mathbf{s}}\hspace{0.1em}',\left(\bm{\theta},\bm{\theta}_\mathcal{G}\right),I\otimes P_\beta,op\left(\ket{0}\bra{0}\right)^{\otimes n}
\otimes\rho\right)\\
&\qquad- \left( 
\mathbb{E}_{\left(\bm{\theta} , \bm{\theta}_\mathcal{G}\right)} 
\sum_{\alpha,\vec{\mathbf{s}}} c_\alpha g(\vec{\mathbf{s}})
f\left(\vec{\mathbf{s}},\left(\bm{\theta},\bm{\theta}_\mathcal{G}\right),I\otimes P_\alpha,op\left(\ket{0}\bra{0}\right)^{\otimes n}
\otimes\rho\right) 
\right)^2.
    \end{aligned}
\end{equation}
Following the same proof of \cref{eq:variance_MPQC}, we can prove the orthogonality of different Pauli path and the second term in the above equation equals 0. More precisely, we have 
\begin{equation}
\operatorname{Var}_{(\bm{\theta},\bm{\theta}_\mathcal{G})}\left[\lossMNoi\right] = 
\mathbb{E}_{\left(\bm{\theta} , \bm{\theta}_\mathcal{G}\right)} 
\sum_{\alpha,\vec{\mathbf{s}}} c_\alpha^2 g(\vec{\mathbf{s}})^2
f\left(\vec{\mathbf{s}},\left(\bm{\theta},\bm{\theta}_\mathcal{G}\right),I\otimes P_\alpha,op\left(\ket{0}\bra{0}\right)^{\otimes n}
\otimes\rho\right)^2.
\end{equation}
Again similar with the proof of \cref{thm:absence_BP_MPQC}, we lower bound \cref{eq:variance_noisy_MPQC} by considering some specfic angle configurations of the gadget layer:
\begin{equation}\label{eq:lower_bound_variance_noisy}
  \begin{aligned}
    &\operatorname{Var}_{(\bm{\theta},\bm{\theta}_\mathcal{G})}\left[\lossMNoi\right] = 
      \mathbb{E}_{\left(\bm{\theta} , \bm{\theta}_\mathcal{G}\right)} 
      \sum_{\alpha}c_\alpha^2 g(\vec{\mathbf{s}})^2
f\left(\vec{\mathbf{s}},\left(\bm{\theta},\bm{\theta}_\mathcal{G}\right),I\otimes P_\alpha,op\left(\ket{0}\bra{0}\right)^{\otimes n}
\otimes\rho\right)^2\\
    &= \frac{1}{4^{m+3n}}\underset{\substack{\bm{\theta} \in \AngleSet^m\\
    \bm{\theta}_\mathcal{G} \in \AngleSet^{3n}}}{\sum}\sum_{\alpha}c_\alpha^2 g(\uniquePathM)^2
f\left(\uniquePathM,\left(\bm{\theta},\bm{\theta}_\mathcal{G}\right),I\otimes P_\alpha,op\left(\ket{0}\bra{0}\right)^{\otimes n}
\otimes\rho\right)^2\\
  & \geq \frac{1}{4^{m+3n}}\underset{\substack{\bm{\theta} \in \AngleSet^m\\
    \bm{\theta}_\mathcal{G} \in M_{\mathrm{swap}}(\bm{\theta})}}{\sum} \sum_{\alpha}c_\alpha^2 g(\uniquePathM)^2f\left(\uniquePathM,\left(\bm{\theta},\bm{\theta}_\mathcal{G}\right),I\otimes P_\alpha,op\left(\ket{0}\bra{0}\right)^{\otimes n}
\otimes\rho\right)^2\\
  & \geq \frac{(1-2\gamma)^{2(\Fnumall+4)}}{4^{m+3n}} \underset{\substack{\bm{\theta} \in \AngleSet^m\\
    \bm{\theta}_\mathcal{G} \in M_{\mathrm{swap}}(\bm{\theta})}}{\sum} \sum_{\alpha}c_\alpha^2 f\left(\uniquePathM,\left(\bm{\theta},\bm{\theta}_\mathcal{G}\right),I\otimes P_\alpha,op\left(\ket{0}\bra{0}\right)^{\otimes n}
\otimes\rho\right)^2\\
  & \geq (1-2\gamma)^{2(\Fnumall+4)} \left(\frac{\tau}{4}\right)^{K}\left\|O\right\|_{HS}^2 = \Omega\left(\frac{1}{\mathrm{poly}(n)}\right).
  \end{aligned}
\end{equation}  
Here, in the second-to-last inequality, we use the following result:
\begin{equation}
  g(\uniquePathM)\geq (1-2\gamma)^{2(\Fnumall+4)}.
\end{equation}
This inequality holds for the following reasons.
First, when we choose $\bm{\theta}_\mathcal{G} \in M_{\mathrm{swap}}(\bm{\theta})$, each gadget transforms the backward-propagated operator $IP$ into $PI$. This implies that for all $i \leq L$ (recall that the gadget layer is located right after $U_L(\theta_L)$), we have $\uniquePathMpos{i}|_{\geq n} = I$, i.e., the part of $\uniquePathMpos{i}$ on the system qubits is the identity operator.
Then we have that for all $i\leq L$,
\begin{equation}
\tr{\uniquePathMpos{i}\mathcal{N}_{i}^\dagger(\uniquePathMpos{i})} = 1.
\end{equation}

Second, when $U_{i'}(\theta_{i'})$ does not belong to the backward light cone of $P_\alpha$, the supports of $\uniquePathMpos{i'}$ and $\mathcal{N}_{i'}$ do not overlap, since $\mathcal{N}_{i'}$ acts on the same qubits as $U_{i'}(\theta_{i'})$. Hence, it follows that
\begin{equation}
\tr{\uniquePathMpos{i'}\mathcal{N}_{i'}^\dagger(\uniquePathMpos{i'})} = 1.
\end{equation}

By combining the two observations above, we obtain
\begin{equation}
  \begin{aligned}
g(\uniquePathM)&=\tr{\uniquePathMpos{0}\mathcal{N}_{op}^\dagger(\uniquePathMpos{0})}\prod_{i=1}^m \tr{\uniquePathMpos{i}\mathcal{N}_{i}^\dagger(\uniquePathMpos{i})}\prod_{i=1}^3 \tr{\uniquePathMpos{\mathcal{G}_{i,1}}\mathcal{N}_{\mathcal{G}_{i}}^\dagger(\uniquePathMpos{\mathcal{G}_{i,1}})}\\
&=\tr{\uniquePathMpos{0}\mathcal{N}_{op}^\dagger(\uniquePathMpos{0})}
\underset{\substack{i > L\\
    U_{i}(\theta_{i}) \in \mathrm{BLig}_{\alpha}^{\mathcal{C}}}}{\prod}
 \tr{\uniquePathMpos{i}\mathcal{N}_{i}^\dagger(\uniquePathMpos{i})}\prod_{i=1}^3 \tr{\uniquePathMpos{\mathcal{G}_{i,1}}\mathcal{N}_{\mathcal{G}_{i}}^\dagger(\uniquePathMpos{\mathcal{G}_{i,1}})}\\
& \geq (1-2\gamma)^{\Fnumall+4},
  \end{aligned}
\end{equation}
where we define $\mathrm{BLig}_{\alpha}^{\mathcal{C}}$ as the backward lightcone of $P_\alpha$ in $\mathcal{C}(\bm{\theta})$, and
the last inequality follows from the fact that at most $\Fnumall$ parameters in $\mathrm{BLig}_{\alpha}^{\mathcal{C}}$ lie after the gadget layer.

\end{proof}
Following similar techniques, we can prove that BP can be guaranteed to be avoided for some particular parameters of the noisy MPQC corresponding to \cref{thm:absence_BP_MPQC_parameters} and \cref{thm:activate_parameters}.
\begin{theorem}\label{app:noise_robustness_others}
  For an MPQC with a $k-$local observable $O = \sum_{\alpha} c_\alpha P_\alpha$, suppose that the conditions stated in \cref{thm:absence_BP_noisy_MPQC} hold, then under Pauli noise with strength at most $\gamma<1/2$ applied after each block, the variance of the gradient of the parameters $\bm{\theta}\in \left[0, 2\pi \right)^{m}$ in the circuit follows the following rules:
\begin{itemize}
\item  For parameter $\theta_j$ located after the gadget layer, if $ \operatorname{Var}_{\bm{\theta}} \left[\frac{\partial \Vexp{O}}{\partial{\theta_j}}\right]\neq 0$, then $\operatorname{Var}_{(\bm{\theta},\bm{\theta}_\mathcal{G})}\left[ \frac{\partial \lossM}{\partial{\theta_j}}\right]$ is lower bounded by
\begin{equation}
  \operatorname{Var}_{(\bm{\theta},\bm{\theta}_{\mathcal{G}})}\left[ \frac{\partial \lossMNoi}{\partial{\theta_j}}\right]\geq (1-2\gamma)^{\Fnumall+4}\left(\frac{1}{2}\right)^{\Fnum }\left(\frac{\tau}{4}\right)^{K}\left\|O\right\|_{min}^2= \Omega\left(\frac{1}{\mathrm{poly}(n)}\right).
\end{equation}

\item For parameters located before the gadget layer, the gradient variance of the MPQC is at least of the same scaling as that of the original PQC without the gadget layer, i.e.
\begin{equation}
  \operatorname{Var}_{(\bm{\theta},\bm{\theta}_{\mathcal{G}})}\left[ \frac{\partial \lossMNoi}{\partial{\theta_j}}\right]\geq  \Omega\left(\frac{1}{\mathrm{poly}(n)}\right) \operatorname{Var}_{\bm{\theta}}\left[ \frac{\partial \lossNoi}{\partial{\theta_j}}\right],
\end{equation}
where $\lossNoi$ is the loss function of the noisy original PQC.
\item Also for the noisy $T$-activating MPQC, we have 
\begin{equation}\label{eq:activate_parameter_lower_bound_noisy}
  \operatorname{Var}_{(\bm{\theta},\bm{\theta}_{\mathcal{G}},\bm{\theta}_{\mathcal{G}'_T})} \left[\frac{\partial\widetilde{L}_T^{\mathcal{C}}\left(\bm{\theta},\bm{\theta}_{\mathcal{G}},\bm{\theta}_{\mathcal{G}'_T}\right)}{\partial{\theta_T}}\right]\geq (1-2\gamma)^{\Fnumall+12}\left(\frac{1}{2}\right)^{\Fnumall+2}\left(\frac{\tau}{4}\right)^{K+1}\left\|O\right\|_{min}^2= \Omega\left(\frac{1}{\mathrm{poly}(n)}\right)
\end{equation}
where $\widetilde{L}_T^{\mathcal{C}}\left(\bm{\theta}, \bm{\theta}_{\mathcal{G}}\right)$ is the loss function of the noisy $T$-activating MPQC.
\end{itemize}
\end{theorem}

\section{Analysis of trainable $op$}\label{app:trainable_op}
Previously, the proofs of our results relied on a deterministic construction of $op$. In this section, we show that the alternative construction illustrated in \cref{fig:op_trainable_proof} preserves all the desirable properties of the corresponding MPQC.

\begin{figure}[H]
  \centering
  \begin{quantikz}
   & \gate[1, style={draw, shape=circle}]{Q} \qw  & \gate[1]{R_{X}(\theta_1)} & \gate[1]{R_{Y}(\theta_2)} &\gate[1, style={draw, shape=circle}]{P}
  \end{quantikz}
  \caption{Trainable construction of $op$. $Q$ and $P$ are Pauli operators.}\label{fig:op_trainable_proof}
\end{figure}
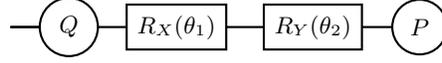
It is easy to verify from \cref{eq:gate_term_in_f_discrete} that under Heisenberg evolution, for any Pauli operator $P\neq I$, among the $4* 4= 16 $ possible combinations of $\theta_1, \theta_2 \in \AngleSet^2 $, at least 4 lead to the resulting operator $Q$ being equal to $Z$. 

In all the proofs, the only parts involving $op$ are as follows:
\begin{equation}
  \tr{s \cdot \left[ op \left(\ket{0}\bra{0}\right)^{\otimes n} \right]}^2,
\end{equation}
for some $n$-qubit Pauli word $s$ with weight at most $K$ (or $K + \order{\log n}$, which we denote simply as $K$ for clarity). As there are parameters in all $op$, we also need to take the average over these angles, namely,
\begin{equation}\label{eq:expectation_op}
 \mathbb{E}_{\bm{\theta}_{op}} \tr{s \bigotimes_{i=1}^n R_{Y_i}(\bm{\theta}_{op_{i,2}})R_{X_i}(\bm{\theta}_{op_{i,1}})\ket{0}\bra{0}R_{X_i}(-\bm{\theta}_{op_{i,1}})R_{Y_i}(-\bm{\theta}_{op_{i,2}})}^2,
\end{equation}
where we define $\bm{\theta}_{op} \in [0,2\pi)^{2n}$ for the parameters in $op$. Without loss of generality, we assume that the first $K' \leq K$ qubits of the Pauli word $s$ are nontrivial. By employing the property of the rotation 2-design stated in \cref{cor:two_design}, the expression in \cref{eq:expectation_op} can be reformulated and lower bounded as
\begin{equation}\label{eq:expectation_op_lower_bound}
  \begin{aligned}
& \mathbb{E}_{\bm{\theta}_{op}} \tr{s \bigotimes_{i=1}^n R_{Y_i}(\bm{\theta}_{op_{i,2}})R_{X_i}(\bm{\theta}_{op_{i,1}})\ket{0}\bra{0}R_{X_i}(-\bm{\theta}_{op_{i,1}})R_{Y_i}(-\bm{\theta}_{op_{i,2}})}^2\\    
& =\mathbb{E}_{\bm{\theta}_{op}\in\AngleSet^{3n}} \tr{ \left[\bigotimes_{i=1}^nR_{X_i}(-\bm{\theta}_{op_{i,1}})R_{Y_i}(-\bm{\theta}_{op_{i,2}})\right] s \left[\bigotimes_{i=1}^n R_{Y_i}(\bm{\theta}_{op_{i,2}})R_{X_i}(\bm{\theta}_{op_{i,1}})\right]\ket{0}\bra{0} }^2\\
& \geq \frac{4^{2(n-K')}4^{K'}}{4^{2n}}\tr{Z^{\otimes K'}\otimes I \ket{0^n}\bra{0^n}}^2 = (\frac{1}{4})^{K'}\geq(\frac{1}{4})^{K}.
  \end{aligned}
\end{equation}  
By substituting \cref{eq:expectation_op_lower_bound} into all Theorems, we obtain the lower bounds on the variance and gradient variance of the MPQC loss function when employing a trainable $op$, simply by replacing $\tau$ with $1/4$.

\section{Hardness of classical simulation of MPQC}\label{app:hardness_MPQC}
In this section, we demonstrate that the classical simulation hardness of MPQCs is not easier than that of the original PQCs. We evaluate two widely used metrics—the worst-case error (WCE) and the mean squared error (MSE)—and prove that, even in the average case, adding the gadget layer does not compromise the classical intractability of the circuit. 
\subsection{Worst case error}
Here we prove that if we can design an efficient classical algorithm which can compute the loss function $\lossM$ with low error for all $(\bm{\theta},\bm{\theta}_{\mathcal{G}}) \in \left[0, 2\pi\right)^{m+3n}$, then we can simulate the loss function of the original PQC efficiently.
\begin{theorem}\label{thm:same_wce}
  For an arbitrary MPQC $\channelM{\bm{\theta},\bm{\theta}_\mathcal{G}}$ and an arbitrary obserable $O$, suppose there exists a classical algorithm that outputs $\mathcal{D}_{{\Phi^{\mathcal{C}}}}\left(\left(\bm{\theta},\bm{\theta}_\mathcal{G}\right),O\right)$ in $\order{\mathrm{poly}(n, \frac{1}{\epsilon})}$ time such that for any $(\bm{\theta},\bm{\theta}_{\mathcal{G}}) \in \left[0, 2\pi\right)^{m+3n}$
\begin{equation}
  \left|\lossM - \mathcal{D}_{{\Phi^{\mathcal{C}}}}\left(\left(\bm{\theta},\bm{\theta}_\mathcal{G}\right),O\right)\right| \leq \epsilon.
\end{equation}
Then, there exists a classical algorithm that outputs an estimate of $\loss = \tr{\mathcal{C}(\bm{\theta}) \rho \mathcal{C}^\dagger(\bm{\theta})O}$ for any $\bm{\theta}\in \left[0, 2\pi\right)^{m}$ with error at most $\epsilon$ in $\order{\mathrm{poly}(n, \frac{1}{\epsilon})}$ time.
\end{theorem}
\begin{proof}
  To arpproximate $\loss$, we directly output $\mathcal{D}_{{\Phi^{\mathcal{C}}}}\left(\left(\bm{\theta},\bm{0}\right),O\right)$. Since $\mathcal{C}(\bm{\theta}) \rho \mathcal{C}^\dagger(\bm{\theta}) = \channelM{\bm{\theta},\bm{0}}(\rho)$, we have
\begin{equation} \left|\loss - \mathcal{D}_{{\Phi^{\mathcal{C}}}}\left(\left(\bm{\theta},\bm{0}\right),O\right)\right| = \left|L^{\mathcal{C}}\left(\bm{\theta},\bm{0}\right)-\mathcal{D}_{{\Phi^{\mathcal{C}}}}\left(\left(\bm{\theta},\bm{0}\right),O\right)\right| \leq \epsilon, \end{equation}
The running time of the above algorithm is also $\order{\mathrm{poly}(n, \frac{1}{\epsilon})}$.
\end{proof}


\subsection{Average case error}
In this subsection, we demonstrate that simulating an MPQC in the average case is no easier than simulating the original PQC.
The the proof idea is as follows: starting from an efficient classical algorithm that approximates $\lossM$ with small average error, we estimate the value of $\loss = L^{\mathcal{C}}\left(\bm{\theta}, \bm{0}\right)$ by randomly sampling $\bm{\theta}_\mathcal{G}$ within a small hypercube centered at $\bm{0}$. Since $\lossM$ is a continuous function of $\bm{\theta}_\mathcal{G}$, the obtained value will be close to $L^{\mathcal{C}}\left(\bm{\theta}, \bm{0}\right)$, with high probability.

We first establish the continuity of $\lossM$, as summarized in the following lemma.
\begin{lemma}\label{lem:Lipschitz_constant}
For an MPQC measured with a local Pauli word $P$, regard its loss funtion $\lossM = \tr{\channelM{\bm{\theta},\bm{\theta}_\mathcal{G}}(\rho)P}$ as a function of $\bm{\theta}_\mathcal{G}$. Then, for any fixed $\bm{\theta}$, the function $\lossM$ is Lipschitz continuous with respect to $\bm{\theta}_\mathcal{G}$, with Lipschitz constant $l_{\bm{\theta}}$ upper bounded by $\sqrt{3K}$, {where $K$ is the support size of $P$'s backward light cone at the gadget layer.}
\end{lemma}
\begin{proof}
Without loss of generalization, we assume that the first $K$ gadgets lie in the backward light cone of $P$, i.e., the remaining $n-K$ gadgets do not affect the value of $\lossM$. Therefore, we set their parameters to zero and denote $\lossM=\lossMP$. 

For the function $\lossMP$, it is Lipschitz continuous since it can be expressed as a finite linear combination of products of $\sin\theta_{\mathcal{G}_{ij}}$ and $\cos\theta_{\mathcal{G}_{ij}}$, each of which is a smooth function.
Consequently, its Lipschitz constant $l_{\bm{\theta}}$ can be upper bounded by the supremum of the $\ell_2$-norm of its gradient with respect to $\bm{\theta}_{\mathcal{G}}$, namely,
\begin{equation}
  \begin{aligned}
  l_{\bm{\theta}} &= \sup_{\theta_{\mathcal{G}}} \|\nabla \lossM\| = \sup_{\theta_{\mathcal{G}}} \|\nabla \lossMP\|_2\\
 &= \sup_{\theta_{\mathcal{G}}} \left\|\left(\frac{\partial}{\partial \theta_{\mathcal{G}_{11}}}\lossMP\right),\ldots,\left(\frac{\partial}{\partial \theta_{\mathcal{G}_{k3}}}\lossMP\right),0,\ldots\right\|_2. 
  \end{aligned}
\end{equation}
For each element, the parameter-shift rule gives 
\begin{equation}
  \begin{aligned}
  \frac{\partial}{\partial \theta_{\mathcal{G}_{ij}}}\lossMP &= \frac{L^{\mathcal{C}}\left(\bm{\theta},\left(\theta_{\mathcal{G}_{11}},\ldots, \theta_{\mathcal{G}_{ij}}+\pi/2,\ldots,\bm{0}\right)\right) - L^{\mathcal{C}}\left(\bm{\theta},\left(\theta_{\mathcal{G}_{11}},\ldots, \theta_{\mathcal{G}_{ij}}-\pi/2,\ldots,\bm{0}\right)\right)}{2}\\
  & \leq 1.
  \end{aligned}
\end{equation}
Consequently, $\lossMP$ is a Lipschitz continuous function of $\bm{\theta}_\mathcal{G}$ with Lipschitz constant
\begin{equation}
  l_{\bm{\theta}} \leq \sqrt{3K}.  
\end{equation}
\end{proof}
The above theorem implies that, for arbitrary $\bm{\theta}$, if $\|\bm{\theta}_{\mathcal{G}} - \bm{\theta_}{\mathcal{G}}'\|_2\leq \epsilon$, then $|\lossM - L^{\mathcal{C}}\left(\bm{\theta},\bm{\theta}_{\mathcal{G}}'\right)| \leq \sqrt{3K}\epsilon$. With this result, we are able to prove the following theorem:

\begin{lemma}\label{lem:classical_hardness_PQC_pauli_word}
  For an arbitrary MPQC $\channelM{\bm{\theta},\bm{\theta}_\mathcal{G}}$ measured with a local Pauli word $P$, suppose there exists a classical algorithm that outputs $\mathcal{A}_{{\Phi^{\mathcal{C}}}}\left(\left(\bm{\theta},\bm{\theta}_\mathcal{G}\right),P\right)$ in $\order{\mathrm{poly}(n, \frac{1}{\epsilon})}$ time such that
\begin{equation}\label{eq:mse_whole}
  \ExpMC\left[\lossM - \mathcal{A}_{{\Phi^{\mathcal{C}}}}\left(\left(\bm{\theta},\bm{\theta}_\mathcal{G}\right),P\right)\right]^2 \leq \epsilon.
\end{equation}
Then, there exists a randomized classical algorithm $A_{\mathrm{rand}}(\bm{\theta})$ that outputs $\mathcal{A}_{\mathcal{C}}(\bm{\theta})$ such that
\begin{equation}
\Pr_{\substack{\bm{\theta}\in\left[0,2\pi\right)^{m}}}\left\{\Pr\left\{\left|\loss - \mathcal{A}_{\mathcal{C}}\left(\bm{\theta}\right)\right|\geq \epsilon_{error}\right\}\geq \delta\right\}\leq 1-\epsilon_{rate}.
\end{equation}
The runtime of $A_{\mathrm{rand}}(\bm{\theta})$ scales as $K^{3K}\order{\mathrm{poly}(n, \frac{1}{\delta}, \frac{1}{\epsilon_{error}}, \frac{1}{\epsilon_{rate}})}$, when {the support size of $P$'s backward light cone at the gadget layer $K$ satisfies} $K = \order{\log n}$.
\end{lemma}
\begin{proof}
Similarly to the proof of \cref{lem:Lipschitz_constant}, without loss of generality, we assume that the first $K$ gadgets lie within the backward light cone of $P$. For simplicity, we denote by $\bm{\theta}_\mathcal{G_P}=(\theta_{\mathcal{G}_{11}},\ldots,\theta_{\mathcal{G}_{K3}})$ the set of rotation angles that directly affect the computation of $\lossM$, and by $\bar{\bm{\theta}}_{\mathcal{G_P}} =(\theta_{\mathcal{G}_{{(K+1)}1}},\ldots,\theta_{\mathcal{G}_{n3}})$ the remaining gadget parameters that do not influence it.

We then expand \cref{eq:mse_whole} over the entire parameter space and while fixing $\bar{\bm{\theta}}_{\mathcal{G_P}} = \bm{0}$ and restricting $\bm{\theta}_\mathcal{G_P}$ to a hypercube $\left[0, \epsilon_1\right)^{3K}$ for some $\epsilon_1 > 0$, which will be determined later. Note that the output of the classical algorithm might depend on $\bar{\bm{\theta}}_{\mathcal{G_P}}$. However, without loss of generality, we assume that when $\bar{\bm{\theta}}_{\mathcal{G_P}} = \bm{0}$, the error with respect to $L^{\mathcal{C}}\left(\bm{\theta},\bm{\theta}_\mathcal{G_P},\bm{0}\right)$ is minimized. This assumption implies that the MSE must be smaller when $\bar{\bm{\theta}}_{\mathcal{G_P}} = \bm{0}$.
As a consequnce, we have 
\begin{equation}
  \begin{aligned}
    \epsilon &\geq \left(\frac{1}{2\pi}\right)^{m+3n}\int_{\bm{\theta}\in\left[0,2\pi\right)^{m}}\int_{\bm{\theta}_\mathcal{G_P}\in\left[0,2\pi\right)^{3K}}\int_{\bar{\bm{\theta}}_\mathcal{G_P}\in\left[0,2\pi\right)^{3(n-K)}}
\left[\lossM - \mathcal{A}_{{\Phi^{\mathcal{C}}}}\left(\left(\bm{\theta},\bm{\theta}_\mathcal{G_P},\bar{\bm{\theta}}_\mathcal{G_P}\right),P\right)\right]^2 d\bm{\theta} d\bm{\theta}_\mathcal{G_P}d\bar{\bm{\theta}}_\mathcal{G_P}\\
&\geq \left(\frac{1}{2\pi}\right)^{m+3n}\int_{\bm{\theta}\in\left[0,2\pi\right)^{m}}\int_{\bm{\theta}_\mathcal{G_P}\in\left[0,2\pi\right)^{3K}}\int_{\bar{\bm{\theta}}_\mathcal{G_P}\in\left[0,2\pi\right)^{3(n-K)}}
\left[L^{\mathcal{C}}\left(\bm{\theta},\bm{\theta}_\mathcal{G_P},\bm{0}\right) - \mathcal{A}_{{\Phi^{\mathcal{C}}}}\left(\left(\bm{\theta},\bm{\theta}_\mathcal{G_P},\bm{0}\right),P\right)\right]^2 d\bm{\theta} d\bm{\theta}_\mathcal{G_P}d\bar{\bm{\theta}}_\mathcal{G_P}\\
&= \left(\frac{1}{2\pi}\right)^{m+3K}\int_{\bm{\theta}\in\left[0,2\pi\right)^{m}}\int_{\bm{\theta}_\mathcal{G_P}\in\left[0,2\pi\right)^{3K}}
\left[L^{\mathcal{C}}\left(\bm{\theta},\bm{\theta}_\mathcal{G_P},\bm{0}\right) - \mathcal{A}_{{\Phi^{\mathcal{C}}}}\left(\left(\bm{\theta},\bm{\theta}_\mathcal{G_P},\bm{0}\right),P\right)\right]^2 d\bm{\theta} d\bm{\theta}_\mathcal{G_P}\\
&\geq \left(\frac{1}{2\pi}\right)^{m+3K}\int_{\bm{\theta}\in\left[0,2\pi\right)^{m}}\int_{\bm{\theta}_\mathcal{G_P}\in\left[0,\epsilon_1\right)^{3K}}
\left[L^{\mathcal{C}}\left(\bm{\theta},\bm{\theta}_\mathcal{G_P},\bm{0}\right) - \mathcal{A}_{{\Phi^{\mathcal{C}}}}\left(\left(\bm{\theta},\bm{\theta}_\mathcal{G_P},\bm{0}\right),P\right)\right]^2 d\bm{\theta} d\bm{\theta}_\mathcal{G_P}\\
&=\left(\frac{\epsilon_1}{2\pi}\right)^{3K} \left(\frac{1}{2\pi}\right)^{m}\left(\frac{1}{\epsilon_1}\right)^{3K}\int_{\bm{\theta}\in\left[0,2\pi\right)^{m}}\int_{\bm{\theta}_\mathcal{G_P}\in\left[0,\epsilon_1\right)^{3K}}
\left[L^{\mathcal{C}}\left(\bm{\theta},\bm{\theta}_\mathcal{G_P},\bm{0}\right) - \mathcal{A}_{{\Phi^{\mathcal{C}}}}\left(\left(\bm{\theta},\bm{\theta}_\mathcal{G_P},\bm{0}\right),P\right)\right]^2 d\bm{\theta} d\bm{\theta}_\mathcal{G_P}\\
&=\left(\frac{\epsilon_1}{2\pi}\right)^{3K} 
\underset{\substack{\bm{\theta}\in\left[0,2\pi\right)^{m} \\
\bm{\theta}_\mathcal{G_P}\in\left[0,\epsilon_1\right)^{3K}}}{\mathbb{E}}
\left[L^{\mathcal{C}}\left(\bm{\theta},\bm{\theta}_\mathcal{G_P},\bm{0}\right) - \mathcal{A}_{{\Phi^{\mathcal{C}}}}\left(\left(\bm{\theta},\bm{\theta}_\mathcal{G_P},\bm{0}\right),P\right)\right]^2. \\
  \end{aligned}
\end{equation}
It implies that
\begin{equation}\label{eq:upper_E_small_cube}
\underset{\substack{\bm{\theta}\in\left[0,2\pi\right)^{m} \\
\bm{\theta}_\mathcal{G_P}\in\left[0,\epsilon_1\right)^{3K}}}{\mathbb{E}}
\left[L^{\mathcal{C}}\left(\bm{\theta},\bm{\theta}_\mathcal{G_P},\bm{0}\right) - \mathcal{A}_{{\Phi^{\mathcal{C}}}}\left(\left(\bm{\theta},\bm{\theta}_\mathcal{G_P},\bm{0}\right),P\right)\right]^2 \leq \left(\frac{2\pi}{\epsilon_1}\right)^{3K}\epsilon.
\end{equation}
\cref{eq:upper_E_small_cube} provides an upper bound on the MSE of the given classical algorithm over the hypercube $\bm{\theta}\in\left[0,2\pi\right)^{m}, \bm{\theta}_\mathcal{G_P}\in\left[0,\epsilon_1\right)^{3K}$ and $\bar{\bm{\theta}}_{\mathcal{G_P}} = \bm{0}$. Then, by applying Markov's inequality, we obtain
\begin{equation}\label{eq:error_probability_small_cube}
\begin{aligned}
\Pr_{\substack{\bm{\theta}\in\left[0,2\pi\right)^{m} \\
\bm{\theta}_\mathcal{G_P}\in\left[0,\epsilon_1\right)^{3K}}}\left\{\left|L^{\mathcal{C}}\left(\bm{\theta},\bm{\theta}_\mathcal{G_P},\bm{0}\right) - \mathcal{A}_{{\Phi^{\mathcal{C}}}}\left(\left(\bm{\theta},\bm{\theta}_\mathcal{G_P},\bm{0}\right),P\right)\right|\geq \epsilon_2\right\}&\leq \underset{\substack{\bm{\theta}\in\left[0,2\pi\right)^{m} \\
\bm{\theta}_\mathcal{G_P}\in\left[0,\epsilon_1\right)^{3K}}}{\mathbb{E}}
\left|L^{\mathcal{C}}\left(\bm{\theta},\bm{\theta}_\mathcal{G_P},\bm{0}\right) - \mathcal{A}_{{\Phi^{\mathcal{C}}}}\left(\left(\bm{\theta},\bm{\theta}_\mathcal{G_P},\bm{0}\right),P\right)\right|/\epsilon_2\\
&\leq \left(\frac{2\pi}{\epsilon_1}\right)^{3K/2}\frac{\sqrt{\epsilon}}{\epsilon_2},
\end{aligned}
\end{equation}
where $\epsilon_2>0$ is a constant to be determined later. The last inequality follows from the fact that for any random variable $X$, we have $\left(\mathbb{E}[X]\right)^2 \leq \mathbb{E}[X^2]$. 

\cref{eq:error_probability_small_cube} implies that the output of the classical algorithm, when evaluated in the hypercube $\bm{\theta} \in \left[0, 2\pi \right)^m, \bm{\theta}_\mathcal{G_P} \in \left[0, \epsilon_1 \right)^{3K}$ and $\bar{\bm{\theta}}_{\mathcal{G_P}} = \bm{0}$, will have very low error with high probability. Hence, by choosing $\epsilon_1$ to be small and leveraging the fact that $L^{\mathcal{C}}\left(\bm{\theta}, \bm{\theta}_\mathcal{G_P}, \bm{0}\right)$ is a Lipschitz continuous function, the output will also be close to $\loss = L^{\mathcal{C}}\left(\bm{\theta}, \bm{0}\right)$ with high probability.

Based on this observation, we construct a randomized classical algorithm $A_{\text{rand}}(\bm{\theta})$ to compute the value of $L^{\mathcal{C}}\left(\bm{\theta}, \bm{0}\right)$. For an arbitrary $\bm{\theta}$, this algorithm randomly selects a $\bm{\theta}_\mathcal{G_P} \in \left[0, \epsilon_1 \right)^{3K}$ and $\bar{\bm{\theta}}_{\mathcal{G_P}} = \bm{0}$, runs the classical algorithm under the assumptions of the theorem, and outputs $\mathcal{A}_{\mathcal{C}}\left(\bm{\theta}\right)\coloneq\mathcal{A}_{{\Phi^{\mathcal{C}}}}\left(\left(\bm{\theta}, \bm{\theta}_\mathcal{G_P}, \bm{0}\right), P\right)$ with some given MSE $\epsilon$. We now analyze this randomized algorithm $A_{\text{rand}}(\bm{\theta})$ and show that it can achieve the performance stated in the theorem for a suitably chosen $\epsilon$.

We first calculate the probability of $\bm{\theta}$ for which our algorithm can achieve low error with high probability. To this end, we define a function $p(\bm{\theta})$, which represents the probability that $A_{\text{rand}}(\bm{\theta})$ incurs high error for a given $\bm{\theta}$.
\begin{equation}\label{eq:condition_theta_G}
  p(\bm{\theta})\coloneq \Pr_{\substack{
\bm{\theta}_\mathcal{G_P}\in\left[0,\epsilon_1\right)^{3K}}}\left\{\left|L^{\mathcal{C}}\left(\bm{\theta},\bm{\theta}_\mathcal{G_P},\bm{0}\right) - \mathcal{A}_{{\Phi^{\mathcal{C}}}}\left(\left(\bm{\theta},\bm{\theta}_\mathcal{G_P},\bm{0}\right),P\right)\right|\geq \epsilon_2\right\}.
\end{equation}
It is easy to verify that 
\begin{equation}
\underset{\substack{\bm{\theta}\in\left[0,2\pi\right)^{m}}}{\mathbb{E}}\left[p(\bm{\theta})\right] = \Pr_{\substack{\bm{\theta}\in\left[0,2\pi\right)^{m} \\
\bm{\theta}_\mathcal{G_P}\in\left[0,\epsilon_1\right)^{3K}}}\left\{\left|L^{\mathcal{C}}\left(\bm{\theta},\bm{\theta}_\mathcal{G_P},\bm{0}\right) - \mathcal{A}_{{\Phi^{\mathcal{C}}}}\left(\left(\bm{\theta},\bm{\theta}_\mathcal{G_P},\bm{0}\right),P\right)\right|\geq \epsilon_2\right\} \leq \left(\frac{2\pi}{\epsilon_1}\right)^{3K/2}\frac{\sqrt{\epsilon}}{\epsilon_2}.
\end{equation}
Again, applying Markov's inequality, we obtain
\begin{equation}\label{eq:lower_bound_frac_Arand_good}
\Pr_{\substack{\bm{\theta}\in\left[0,2\pi\right)^{m}}}\left\{p(\bm{\theta})\geq \delta\right\}\leq \underset{\substack{\bm{\theta}\in\left[0,2\pi\right)^{m}}}{\mathbb{E}}\left[p(\bm{\theta})\right]/\delta\leq \left(\frac{2\pi}{\epsilon_1}\right)^{3K/2}\frac{\sqrt{\epsilon}}{\epsilon_2\delta}.
\end{equation}
\cref{eq:lower_bound_frac_Arand_good} implies that, with probability at least
$\displaystyle 1 - \left(\frac{2\pi}{\epsilon_1}\right)^{3K/2} \frac{\sqrt{\epsilon}}{\epsilon_2 \delta}$
over $\bm{\theta}\in[0,2\pi)^m$, the output of $A_{\mathrm{rand}}(\bm{\theta})$ has an error smaller than $\epsilon_2$ with probability at least $1-\delta$.
Focusing on these values of $\bm{\theta}$ and on those $\bm{\theta}_\mathcal{G_P}$ that violates the condition in \cref{eq:condition_theta_G},
$A_{\mathrm{rand}}(\bm{\theta})$ produces an output $\mathcal{A}_{{\Phi^{\mathcal{C}}}}\left(\left(\bm{\theta}, \bm{\theta}_\mathcal{G_P}, \bm{0}\right), P\right)$ that satisfies
\begin{equation}
\left|L^{\mathcal{C}}\left(\bm{\theta},\bm{\theta}_\mathcal{G_P},\bm{0}\right) - \mathcal{A}_{{\Phi^{\mathcal{C}}}}\left(\left(\bm{\theta},\bm{\theta}_\mathcal{G_P},\bm{0}\right),P\right)\right|\leq \epsilon_2.
\end{equation}
Then, by employing the triangle inequality and the Lipschitz continuity of $\lossM$ as a function of $\bm{\theta}_\mathcal{G}$, we have
\begin{equation}
  \begin{aligned}   \left|\mathcal{A}_{{\Phi^{\mathcal{C}}}}\left(\left(\bm{\theta},\bm{\theta}_\mathcal{G_P},\bm{0}\right),P\right) -\loss\right| &=\left|\mathcal{A}_{{\Phi^{\mathcal{C}}}}\left(\left(\bm{\theta},\bm{\theta}_\mathcal{G_P},\bm{0}\right),P\right) - L^{\mathcal{C}}\left(\bm{\theta},\bm{0},\bm{0}\right)\right|\\
& \leq \left|\mathcal{A}_{{\Phi^{\mathcal{C}}}}\left(\left(\bm{\theta},\bm{\theta}_\mathcal{G_P},\bm{0}\right),P\right)-L^{\mathcal{C}}\left(\bm{\theta},\bm{\theta}_\mathcal{G_P},\bm{0}\right) \right| + \left|L^{\mathcal{C}}\left(\bm{\theta},\bm{\theta}_\mathcal{G_P},\bm{0}\right) -L^{\mathcal{C}}\left(\bm{\theta},\bm{0},\bm{0}\right)\right|\\
&\leq \epsilon_2 + \epsilon_1 \sqrt{3K}l_{\bm{\theta}} \leq \epsilon_2 + \epsilon_13K.
  \end{aligned}
\end{equation}
In the end, we determine the unfixed parameters introduced earlier. To ensure that $A_{\text{rand}}(\bm{\theta})$ achieves an error of at most $\epsilon_{\text{error}}$, we need to set $\epsilon_2 + 3K\epsilon_1 \leq \epsilon_{\text{error}}$. Hence, we choose $\epsilon_2 = \frac{\epsilon_{\text{error}}}{2}$ and $\epsilon_1 = \frac{\epsilon_{\text{error}}}{6K}$ to satisfy the error condition. Next, to ensure that $A_{\text{rand}}(\bm{\theta})$ works with probability at least $1 - \epsilon_{\text{rate}}$ over $\bm{\theta}$, we apply \cref{eq:lower_bound_frac_Arand_good} and set $\left(\frac{2\pi}{\epsilon_1}\right)^{3K/2} \frac{\sqrt{\epsilon}}{\epsilon_2 \delta} \leq \epsilon_{\text{rate}}$. Substituting the values of $\epsilon_1$ and $\epsilon_2$ into the inequality, we obtain

\begin{equation}
  \epsilon\leq \frac{{\left(12\pi K\right)}^{-3K}}{4}\epsilon_{rate}^2\epsilon_{error}^{3K+2}\delta^2 = K^{-3K} \frac{\epsilon_{rate}^2\epsilon_{error}^{\order{\log n}}\delta^2}{\order{\mathrm{poly}(n)}}.
\end{equation}
The above calculation implies that, to achieve a randomized classical algorithm that satisfies the conditions in the theorem, we first randomly pick each $\theta_{\mathcal{G}_{ij}} \in \bm{\theta}_\mathcal{G_P}$ from the interval $\left[0, \frac{\epsilon_{\text{error}}}{6K}\right)$. We then run the classical algorithm to compute $\lossM$ with MSE at most $K^{-3K}\left(\frac{3}{\pi}\right)^{3K} \epsilon_{\text{rate}}^2 \epsilon_{\text{error}}^{3K} \delta^2$. The running time of $A_{\text{rand}}(\bm{\theta})$ is $\mathcal{O}\left(\mathrm{poly}(n, K^{3K}\frac{1}{\epsilon_{\text{rate}}^2 \epsilon_{\text{error}}^{\log n} \delta^2})\right) = K^{3K} \mathcal{O}\left(\mathrm{poly}(n, \frac{1}{\delta}, \frac{1}{\epsilon_{\text{error}}}, \frac{1}{\epsilon_{\text{rate}}})\right)$, which satisfies the condition in the theorem.
\end{proof}

Based on the above lemma, we can directly extend the observable from a single local Pauli word to an arbitrary $k$-local observable $O$:
\begin{theorem}\label{thm:same_average_simulate_formall}
Suppose there exists a classical algorithm running in $\order{\mathrm{poly}\left(n, \tfrac{1}{\zeta}\right)}$ time that can estimate the expectation value of any local observable for arbitrary MPQCs with MSE no larger than $\zeta$.
Then, for any PQC $\mathcal{C}(\bm{\theta})$ and any local observable $O = \sum_\alpha c_\alpha P_\alpha$ consisting of a polynomial number of Pauli terms, there exists a randomized classical algorithm that, with probability at least $1 - 1/\mathrm{poly}(n)$ over $\bm{\theta}$, outputs an estimate of $\loss = \tr{\mathcal{C}(\bm{\theta}) \rho \mathcal{C}^\dagger(\bm{\theta})O}$ with error at most $\epsilon$ with success probability at least $1-\delta$.
The runtime of this algorithm scales as $\order{\mathrm{poly}\left(n, \frac{1}{\delta}, \frac{1}{\epsilon}\right)}$.
\end{theorem}
\begin{proof}
To compute $\loss = \tr{\mathcal{C}(\bm{\theta}) \rho \mathcal{C}^\dagger(\bm{\theta})O} = \Vexp{O}$, we first construct an MPQC $\channelM{\bm{\theta},\bm{\theta}_\mathcal{G}}$ from $\mathcal{C}(\bm{\theta})$ such that the backward-propagated support of all Pauli components $P_\alpha$ in $O$ through the gadget layer is at most $K$.

Using the classical algorithm established in the theorem, which efficiently computes $\tr{\channelM{\bm{\theta},\bm{\theta}_\mathcal{G}}(\rho)P_\alpha}$,
we estimate the expectation value $\Vexp{P_\alpha}$ for each Pauli term and then reconstruct $\Vexp{O}$ via the weighted sum $\sum_\alpha c_\alpha \Vexp{P_\alpha}$.
Let $\#O = \order{\mathrm{poly}(n)}$ denote the number of Pauli terms in $O$.
For each $\Vexp{P_\alpha}$, according to \cref{lem:classical_hardness_PQC_pauli_word}, we can design a randomized classical algorithm $A_{\text{rand}}^{P_\alpha}(\bm{\theta})$ that estimates $\Vexp{P_\alpha}$ within error $\frac{\epsilon}{|c_\alpha| \#O}$ and with success probability at least $1-\frac{\delta}{\#O}$, for probability at least $1-\frac{1}{\#O\mathrm{poly}(n)}$ over $\bm{\theta}$.
The runtime of $A_{\text{rand}}^{P_\alpha}(\bm{\theta})$ is
\begin{equation}
  K^{3K} \mathcal{O}\left(\mathrm{poly}(n, \frac{\#O}{\delta}, \frac{|c_\alpha| \#O}{\epsilon},\#O\mathrm{poly}(n))\right) = K^{3K}\#O \cdot \mathrm{poly}(n)\order{\mathrm{poly}(n, \frac{1}{\delta}, \frac{1}{\epsilon})}.
\end{equation}
Summing the outputs of all $A_{\text{rand}}^{P_\alpha}(\bm{\theta})$ yields the final algorithm $A_{\text{rand}}^{O}(\bm{\theta})$ for $\loss$, whose runtime is upper bounded by $K^{3K}\left(\#O\right)^2\mathrm{poly}(n)\order{\mathrm{poly}(n, \frac{1}{\delta}, \frac{1}{\epsilon})} = K^{3K}\order{\mathrm{poly}(n, \frac{1}{\delta}, \frac{1}{\epsilon})}$.

By the union bound, the probability (over $\bm{\theta}$) that any $A_{\mathrm{rand}}^{P_\alpha}(\bm{\theta})$ fails to satisfy the required condition is at most
$\#O \cdot \frac{1}{\#O\mathrm{poly}(n)} = \frac{1}{\mathrm{poly}(n)}$.
Hence, focusing on those $\bm{\theta}$—which occur with probability at least $1-\frac{1}{\mathrm{poly}(n)}$—for which each $A_{\mathrm{rand}}^{P_\alpha}(\bm{\theta})$ outputs an estimate with error at most $\tfrac{\epsilon}{|c_\alpha|\#O}$ and success probability at least $1-\tfrac{\delta}{\#O}$,
we obtain that, again by the union bound, the probability that any single $A_{\mathrm{rand}}^{P_\alpha}(\bm{\theta})$ exceeds its error threshold $\frac{\epsilon}{|c_\alpha|\#O}$ is at most
$\#O \cdot \tfrac{\delta}{\#O} = \delta$.
Conditioning on the successful instances, 
the total error is bounded by
\begin{equation}
  \sum_\alpha |c_\alpha| \frac{\epsilon}{|c_\alpha| \#O} =\epsilon,
\end{equation} 
which satisfies the theorem's conditions. 
Because the assumption classical algorithm works for arbitrary MPQCs, and we can always construct an MPQC with $K = \order{1}$ for any PQC,
the runtime of the final algorithm scales as $\order{\mathrm{poly}\left(n, \frac{1}{\delta}, \frac{1}{\epsilon}\right)}$.
\end{proof}
\begin{remark}
  \Cref{thm:same_average_simulate_formall} implies that if MPQCs are classically simulable on average, then arbitrary PQCs would also be efficiently simulable by a $\mathrm{BPP}$ Turing machine on average.
In other words, the average-case classical simulation of PQCs would belong to the heuristic complexity class $\mathrm{HeurBPP}$~\cite{bogdanov2006average,perez2024classical}.
Notably, existing works on classical simulation of quantum circuits typically rely on specific assumptions about the distribution of circuit gates~\cite{angrisani2024classically}, and whether general PQCs are classically simulable on average remains an open question. 
\end{remark}
\section{Numerical experiments}\label{app:numerical}
In this section, we provide details of numerical results in the manuscript. 

First, we construct a deliberately designed example in which the original PQC becomes untrainable even at small system sizes, while the corresponding MPQC remains trainable and is able to recover the optimal solution. 
Second, we consider the task of approximating the ground state of a complex Hamiltonian, where we also show that, by employing the activation strategy introduced in \cref{app:activate_multiple_paras}, MPQC can further reduce the loss and achieve a better ground-state approximation.

\subsection{Effectiveness of MPQC under a poorly designed PQC ansatz}

Owing to the limitations of current classical simulation methods for variational quantum algorithms, the system sizes that can be explored numerically are relatively small. 
In particular, for MPQC, the additional ancilla qubits introduced by the gadget layers further constrain the maximum system size accessible to simulation, typically to at most a few tens of qubits. 
In this regime, although gradients may scale exponentially with system size in principle, their magnitudes are not necessarily extremely small, and PQCs can still be trainable.

Nevertheless, extremely small gradients can still occur even at these moderate system sizes, depending on the specific circuit architecture. To clearly illustrate the advantage brought by MPQC, we construct an artificial yet representative example in which a PQC becomes untrainable due to an unfavorable circuit design. 
Specifically, we construct the PQC as follows, which is obtained by replacing all rotation gates in the circuit in the manuscript into $R_x$ in Section VI. 
\begin{figure}[H]
    \centering
    \includegraphics[width=0.6\textwidth]{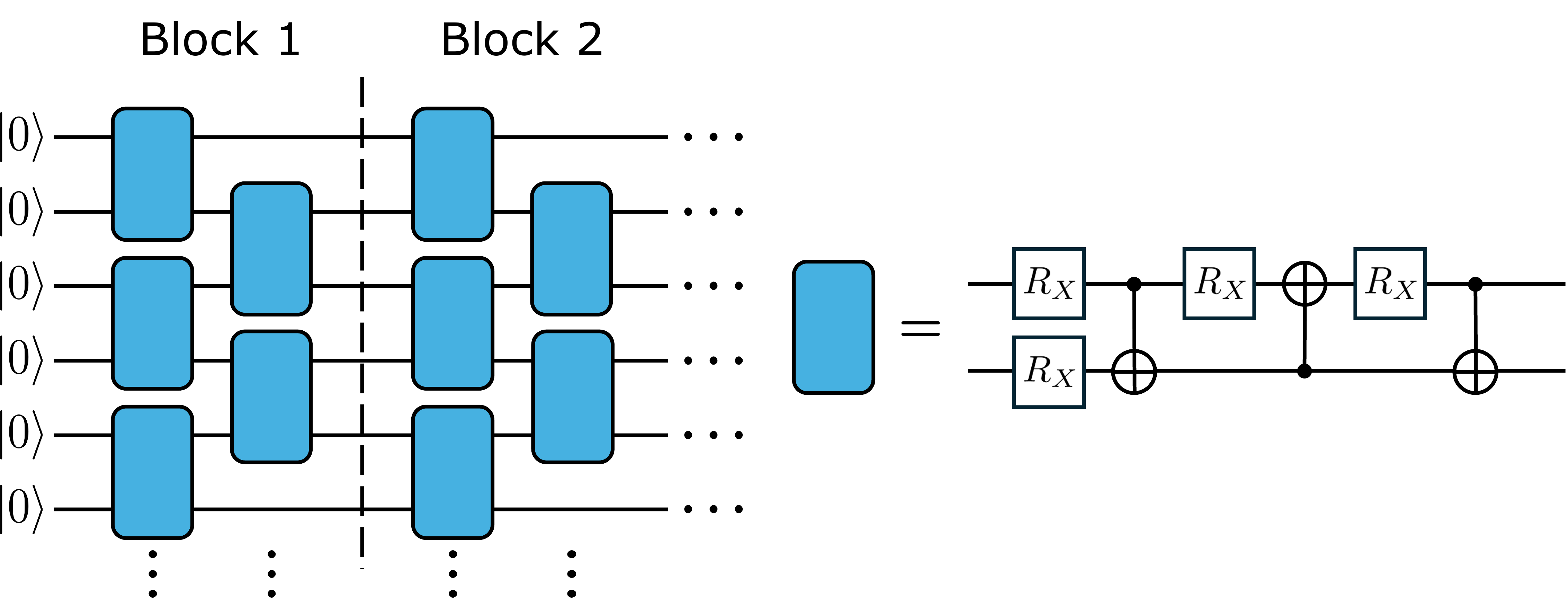}
    \caption{An example of a poorly designed PQC obtained by restricting all rotation gates with $R_x$ gates.}
    \label{fig:bad_pqc}
\end{figure} 
We consider the task of finding the ground state of the following two-local transverse-field Ising Hamiltonian:
\begin{equation}\label{eq:cheating_H}
  H_{\mathrm{TFI}}=-\sum_{j=1}^n X_j X_{j+1}-h\sum_{j=1}^n Z_j
\end{equation}
defined on a periodic one-dimensional chain, where $h>0$ is treated as a tunable parameter.
With the above choice of rotation gates, the Pauli-operator evolution governed by \cref{eq:gate_term_in_f_discrete} shows that the circuit parameters fail to influence the $XX$ terms in the Hamiltonian, since the rotation gate generators commute with the Pauli operators backward-propagated from the $XX$ terms. Consequently, when $h$ is chosen sufficiently small, the gradient variance with respect to all circuit parameters becomes uniformly small. This allows us to artificially construct barren-plateau-like behavior even for circuits of small depth and modest system size. In contrast, after inserting the gadget layer, the diversity of Pauli paths is significantly enhanced, thereby restoring nontrivial couplings between the circuit parameters and the $XX$ terms in the Hamiltonian.

We then perform numerical experiments to demonstrate that MPQC remains capable of finding the ground state even when the original PQC suffers from such an unfavorable design. Concretely, we set $n=6$ and consider $h=0.01$ and $h=0.5$ in \cref{eq:cheating_H}. The original PQC consists of six blocks, each corresponding to the structure shown in \cref{fig:bad_pqc}. The associated MPQC is obtained by inserting a gadget layer after the fourth block, i.e., two blocks before the final measurement. In addition, we construct a ``shallow'' PQC containing only a single block, illustrating that even very shallow circuits with this unfavorable design remain untrainable.

For all three circuit architectures, parameters are initialized randomly from the uniform distribution $[0,2\pi)$, and we perform ten independent training runs with different random seeds to mitigate the effect of unlucky initializations. Optimization is carried out using the Adam optimizer~\cite{diederik2014adam} with a learning rate of $0.01$ for $1000$ training epochs. All simulations are performed using \emph{PennyLane}~\cite{bergholm2018pennylane}. Detailed numerical results and further discussion are presented in the main text.



\subsection{Application of parameter activation strategy}
In this subsection, we demonstrate the employment of the activation strategy and present additional numerical evidence showing that MPQC achieves substantially better performance than the original PQC.
We consider a random-sign
2-local XYZ Hamiltonian of the form
\begin{equation}\label{eq:random_sign_H}
  H_G=\sum_{\{i,j\}\in E}\big(J^{(x)}_{ij}X_iX_j+J^{(y)}_{ij}Y_iY_j+J^{(z)}_{ij}Z_iZ_j\big),
\end{equation}
where $G=(V,E)$ is an undirected graph with vertex set $V=\{1,2,\dots,n\}$ and $X_i,Y_i,Z_i$ denote Pauli operators acting on qubit $i$ and identity on all other qubits.
The couplings in $H_G$ are i.i.d.\ random signs, e.g.\ $J^{(\alpha)}_{ij}\in\{-1,+1\}$ with equal probability for each $\alpha\in\{x,y,z\}$ and each edge $\{i,j\}\in E$. Here to ensure the hardness of the optimization problem, we choose $G$ to be the complete graph on $12$ vertices.

Random-sign spin Hamiltonians are canonical models of disorder and frustration, widely used to study spin-glass physics and as challenging benchmark instances for quantum optimization and variational ground-state preparation~\cite{SherringtonKirkpatrick1975,Venturelli2015}. From the computational-complexity viewpoint, the task of estimating (or deciding) the ground-state energy of generic 2-local quantum Hamiltonians is QMA-complete~\cite{KempeKitaevRegev2006}, and hardness persists under physically motivated restrictions such as geometrically local interactions~\cite{OliveiraTerhal2008}.

To address this task, we extend the PQC architecture from a one-dimensional chain to a two-dimensional lattice, reflecting the structure of the target Hamiltonian $H_G$. The resulting ansatz is composed of repeated blocks, each of which is shown in \cref{fig:2d_ansatz}. Starting from a PQC consisting of eight such blocks, we construct the corresponding MPQC by inserting a gadget layer in the middle of the circuit, i.e., after the fourth block.

To further improve the optimization performance, we activate the parameters in the first block, as illustrated in \cref{fig:2d_mpqc_actmpqc}. Here, our goal is to activate the entire block. According to the strategy described in \cref{app:activate_multiple_paras}, this would in principle require introducing an additional gadget layer before the first block. To reduce the complexity of the numerical simulations, we reuse the ancilla qubits introduced by the gadget layer in \cref{fig:2d_mpqc_actmpqc}(a).

\begin{figure}[H]
    \centering
    \includegraphics[width=0.6\textwidth]{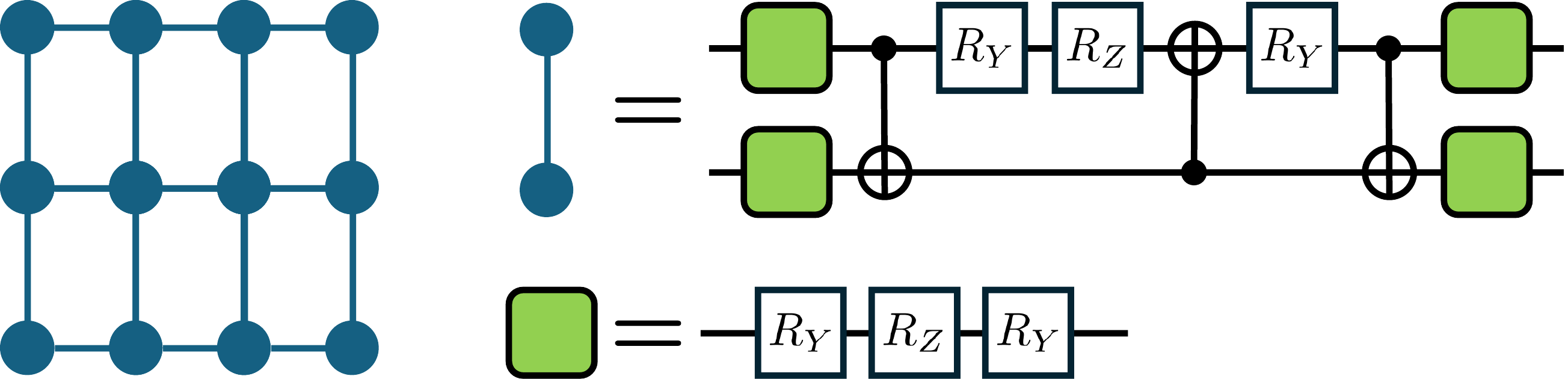}
    \caption{One block of the 2D lattice ansatz used to approximate the ground state of $H_G$. The complete circuit is constructed by repeating this block multiple times.}
    \label{fig:2d_ansatz}
\end{figure}

\begin{figure}[H]
    \centering
    \includegraphics[width=1\textwidth]{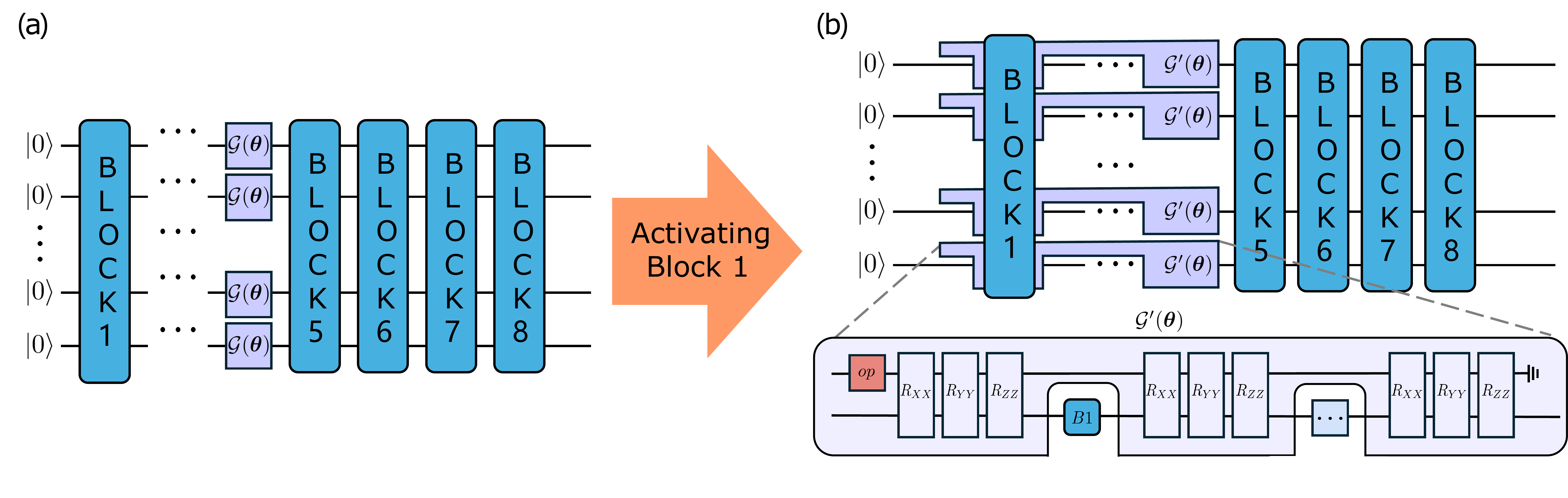}
    \caption{(a) Construction of MPQC, where a gadget layer is inserted in the middle of the original PQC. (b) Strategy for activating the parameters in the first block. Gates denoted by ``$\cdots$'' correspond to blocks 2, 3, and 4, while $B1$ represents the gates in the first block.}
    \label{fig:2d_mpqc_actmpqc}
\end{figure}

To demonstrate that MPQC outperforms the original PQC, we perform numerical simulations using the 2D ansatz with different circuit depths. Specifically, we train a family of PQCs with the number of blocks ranging from 1 to 8. As in the previous section, all circuit parameters are initialized independently from the uniform distribution $[0,2\pi)$, and 10 independent training runs with different random seeds are performed for each setting. Optimization is carried out using the Adam optimizer with a learning rate of $0.01$ for $3000$ iterations for all PQCs.

For MPQC, we first optimize the circuit shown in \cref{fig:2d_mpqc_actmpqc}(a) for $2000$ iterations. We then further minimize the loss by activating the parameters as in \cref{fig:2d_mpqc_actmpqc}(b) and continuing the optimization for an additional $1000$ iterations. The newly introduced parameters are initialized to zero so that the second-stage optimization starts from the state obtained in the first stage. We emphasize that the activation strategy is not optimized in this example, and further performance improvements may still be possible.

\end{document}